\documentclass[smallextended]{svjour3}       
\smartqed

\usepackage{amsmath}
\usepackage{times}
\usepackage{bm}
\usepackage{natbib}
\usepackage{listings,algorithmic}
\usepackage{amsmath,amssymb,color}
\usepackage{verbatim}
\usepackage{amsmath,amssymb,color}
\usepackage{color} 
\RequirePackage[OT1]{fontenc}
\RequirePackage{amsmath, amssymb, enumerate}
\usepackage{eufrak}
\usepackage{graphicx, url,pstricks,fancyhdr}
\RequirePackage{bbm}
\RequirePackage{amssymb}

\usepackage{float}
\RequirePackage{enumitem}
\RequirePackage{bigints}
\RequirePackage{mathtools}
\RequirePackage[stable]{footmisc}
\usepackage{bm}
\RequirePackage{amsmath, amssymb, enumerate}
\usepackage{eufrak}
\usepackage{longtable}
\usepackage{graphicx, url,pstricks,fancyhdr}
\usepackage{float}
\RequirePackage{enumitem}
\RequirePackage{bigints}
\RequirePackage{mathtools, amsfonts}
\RequirePackage[stable]{footmisc}
\usepackage{tikz}
\usetikzlibrary{external}
\usetikzlibrary{arrows,shapes,trees} 
\usetikzlibrary{positioning}
\usetikzlibrary{fit,chains}
\usetikzlibrary{backgrounds}

\newcommand{\E}{{\mathbb{E}} }
\newcommand{\Prob}{\mathbb{P}\,}

\begin{document}

\title{Detecting multiple generalized change-points by isolating single ones
}


\author{Andreas Anastasiou\footnote{On behalf of all authors, the corresponding author states that there is no conflict of interest.}         \and
        Piotr Fryzlewicz\footnote{Piotr Fryzlewicz's work was supported by the Engineering and Physical Sciences Research Council grant No. EP/L014246/1.} 
}


\institute{A. Anastasiou \at
              Department of Mathematics and Statistics, University of Cyprus, P.O. Box: 20537, 1678, Nicosia, Cyprus. \\
              \email{anastasiou.andreas@ucy.ac.cy}           
           \and
           P. Fryzlewicz \at
              Department of Statistics, The London School of Economics and Political Science, Columbia House, Houghton Street, London, WC2A 2AE.
}

\date{Received: date / Accepted: date}

\maketitle

\begin{abstract}
\indent We introduce a new approach, called Isolate-Detect (ID), for the consistent estimation of the number and location of multiple generalized change-points in noisy data sequences. Examples of signal changes that ID can deal with are changes in the mean of a piecewise-constant signal and changes, continuous or not, in the linear trend. The number of change-points can increase with the sample size. Our method is based on an isolation technique, which prevents the consideration of intervals that contain more than one change-point. This isolation enhances ID's accuracy as it allows for detection in the presence of frequent changes of possibly small magnitudes. In ID, model selection is carried out via thresholding, or an information criterion, or SDLL, or a hybrid involving the former two. The hybrid model selection leads to a general method with very good practical performance and minimal parameter choice. In the scenarios tested, ID is at least as accurate as the state-of-the-art methods; most of the times it outperforms them. ID is implemented in the R packages \textbf{IDetect} and \textbf{breakfast}, available from CRAN.
\end{abstract}

\section{Introduction}
\label{sec:intro}
Change-point detection is an active area of statistical research that has attracted a lot of interest in recent years.
Our work's focus is on \textit{a posteriori} change-point detection, where the aim is to estimate the number and locations of certain changes in the behaviour of the data. 
We work in the model
\begin{equation}
\label{our_model}
X_t = f_t + \sigma\epsilon_t, \quad t=1,2,\ldots,T,
\end{equation}
where $X_t$ are the observed data and $f_t$ is a one-dimensional, deterministic signal with structural changes at certain points. Two examples are: change-points in the level when $f_t$ is seen as piecewise-constant, and change-points in the first derivative when $f_t$ is piecewise-linear. We highlight, however, that our methodology and analysis apply to more general scenarios, for instance the detection of knots in a piecewise polynomial signal of order $k$, where $k$ is not necessarily equal to zero (piecewise-constant mean) or one (piecewise-linear mean). The number $N$ of change-points as well as their locations $r_1, r_2, \ldots, r_N$ are unknown and our aim is to estimate them. In addition, $N$ can grow with $T$. The random variables $\epsilon_t$ in \eqref{our_model} have mean zero and variance one; further assumptions will be given in Section \ref{subsec:theory}.

When $f_t$ is assumed to be piecewise-constant, the existing change-point detection techniques are mainly split into two categories based on whether the change-points are detected all at once or one at a time. The former category mainly includes optimization-based methods, in which the estimated signal is chosen based on its least squares or log-likelihood fit to the data, penalized by a complexity rule in order to avoid overfitting. The most common example of a penalty function is the Schwarz Information Criterion (SIC); see \cite{Yao} for details. To solve the implied penalization problem, dynamic programming approaches, such as the Segment Neighborhood (SN) and Optimal Partitioning (OP) methods of \cite{Auger_Lawrence} and \cite{Jackson}, have been developed. In an attempt to improve on OP's computational cost, \cite{Killick_PELT} introduce the PELT method, based on a pruning step applied to OP's dynamic programming approach. A non-parametric adaptation of PELT is given in \cite{Haynes_npPelt}. \cite{Rigaill} introduces an improvement over classical SN algorithms, through a pruning approach called PDPa, while \cite{Maidstone_et_al} give two algorithms by combining ideas from PELT and PDPa. \cite{SMUCE} propose the simultaneous multiscale change-point estimator (SMUCE) for the change-point problem in the case of exponential family regression; solving an optimization problem is also required. The FDRSeg method of \cite{LiMunk} is a combination of False Discovery Rate (FDR) control and global segmentation methods in a multiscale way; the change-points are again detected all at once.

In the latter category, in which change-points are detected one at a time, a popular method is binary segmentation, which performs an iterative binary splitting of the data on intervals determined by the previously obtained splits. \cite{Vostrikova} introduces and proves the validity of binary segmentation in the setting of change-point detection for piecewise-constant signals. 
The main advantages of binary segmentation are its conceptual simplicity and low computational cost. However, at each step of the algorithm, binary segmentation looks for a single change-point, which leads to its suboptimality in terms of accuracy, especially for signals with frequent change-points. Some variants of binary segmentation that work towards solving this issue are the Circular Binary Segmentation (CBS) of \cite{Olshen_Venkatraman}, the Wild Binary Segmentation (WBS) of \cite{Fryzlewicz_WBS} as well as its second version (WBS2) of \cite{Fryzlewicz_WBS2}, the Narrowest-Over-Threshold (NOT) method of \cite{NOT_paper}, and the Seeded Binary Segmentation (SeedBS) of \cite{Seeded_2020}. CBS searches for at most two change-points at each step. Instead of initially calculating the contrast value for the whole data sequence, WBS and NOT are based on a random draw of subintervals of the domain of the data, on which an appropriate statistic is tested against a threshold. The draw of all the subintervals takes place at the beginning of the algorithm. In contrast, WBS2 draws first only a small number, $\tilde{M}$, of data subsamples. It then uses the first change-point candidate to split the data into two parts, and again recursively draws the same number $\tilde{M}$ of subsamples to the left and to the right of this change-point candidate, and so on. A major difference between WBS and WBS2 is that the latter adaptively decides where to recursively draw the next subsamples, based on the change-point candidates detected so far; this adds to the detection power of the method. SeedBS is an approach, similar to WBS and NOT, that relies instead on a deterministic construction of background intervals in which single change points are searched.
 Apart from binary-segmentation-related approaches, the category in which the change-points are detected one at a time also includes methods that control the False Discovery Rate. For instance, the ``pseudo-sequential'' (PS) procedure of \cite{Venkatraman}, as well as the CPM method of \cite{Ross_CPM} are based on an adaptation of online detection algorithms to a posteriori situations and work by bounding the Type I error rate of falsely detecting change-points. Some methods do not fall in either category. For example, the tail-greedy algorithm in \cite{Fryzlewicz_tail-greedy} achieves a multiscale decomposition of the data using Unbalanced Haar wavelets in an agglomerative way. In addition, \cite{Kirch} use moving sum (MOSUM) statistics in order to detect multiple change-points. For a more thorough review of the literature on the detection of multiple change-points in the mean of univariate data sequences, see \cite{Cho_Kirch} and \cite{YiYu}. \cite{Truong_Oudre_Vayatis} also present a survey of various a posteriori change-point detection algorithms; the focus is, however, on multivariate time series.


Beyond the piecewise-constant signal model, existing methods mainly minimize the residual sum of squares taking into account a penalty, with the most common being the SIC. This is used in \cite{BaiPerron}, in the trend filtering (TF) approach \citep{KimTF,TibshiraniTF}, and in the dynamic programming algorithm CPOP \citep{FearnheadCPOP}. 
\cite{Friedman} introduces the Multivariate Adaptive Regression Splines (MARS) method for regression analysis based on splines with the number and the location of the knots being determined by the data. \cite{Spiriti} propose two methods for optimizing knot locations in spline smoothing, where either the number of knots is fixed or an upper bound for it needs to be given. The NOT approach \citep{NOT_paper} detects change-points one at a time in various scenarios including piecewise-linear mean signals.

In general, change-point detection becomes easier in situations where there is at most one change-point to be detected in a given interval; in such cases the detection power of the contrast function (more details are in Section \ref{subsec:theory}) is maximised. Therefore, it makes sense to decouple the multiple change-point detection problem into many single change-point detections. To achieve this, we propose a generic technique, Isolate-Detect (ID), for generalized change-point detection in various different structures, such as piecewise-constant or piecewise-linear signals. The concept behind ID is simple and is split into two stages; firstly, the isolation of each of the true change-points within subintervals of the domain $[1,2,\ldots,T]$, and secondly their detection. From now on, the terms \textit{subinterval} and \textit{interval} will be used interchangeably. Although a detailed explanation of our methodology is provided in Section \ref{subsec:methodology}, the basic idea is that for an observed data sequence of length $T$ and with $\lambda_T$ a positive constant, ID first creates two ordered sets of $K = \left\lceil T/\lambda_T \right\rceil$ right- and left-expanding intervals as follows. The $j^{th}$ right-expanding interval is $R_j = [1,\min\left\lbrace j\lambda_T, T\right\rbrace]$, while the $j^{th}$ left-expanding interval is $L_{j} = [\max\left\lbrace 1,T - j\lambda_T + 1\right\rbrace,T]$. We collect these intervals in the ordered set $S_{RL} = \left\lbrace R_1, L_1, R_2, L_2, \ldots,R_K,L_K\right\rbrace$. For a suitably chosen contrast function (more details are in Section \ref{subsec:theory}), ID identifies the point with the maximum contrast value in $R_1$. If its value exceeds a threshold, denoted by $\zeta_T$, then it is taken as a change-point. If not, then the next interval in $S_{RL}$ is tested. Upon detection, ID makes a new start from the end-point (or start-point) of the right- (or left-) expanding interval where the detection occurred. Upon correct choice of $\zeta_T$, ID ensures that we work on intervals with at most one change-point, which was our aim.

We would like to highlight the importance of the change-point isolation aspect present in our method as explained in the previous paragraph. There are various advantages. First, it enables detection in higher-order polynomial signals. Second, it is carried out in a fixed and systematic way, which eliminates any randomness in the selection of the intervals and, by extension, in the final results. Third, the way the isolation is carried out in ID makes it quicker than other localisation-focused algorithms, such as NOT, due to the fact that it needs to work on fewer intervals; more details on this advantage of our proposed methodology are in Section \ref{subsec:Compcost}. We note here that, even though the default methodology described in \cite{Fryzlewicz_WBS} and \cite{NOT_paper} is based on the construction of random intervals, the same approaches can be applied to a fixed grid of intervals.
However, as noted in \cite{Seeded_2020}, the latter implementation can be quite slow. Fourth, the pseudo-sequential nature of the attempted isolation, makes our proposed methodology suitable for online change-point detection. This is one of the various different ways that ID is different from existing techniques in the literature which also attempt change-point isolation; a more thorough comparison with seemingly similar, but still different, methods is given in the next section.

The paper is organized as follows. Section \ref{sec:motivation_illustration} is a motivating illustration of our proposed method through examples. Section \ref{sec:methodology_theory} gives a formal explanation of the ID methodology along with two different scenarios of use and the associated theory. In Section \ref{sec:variants}, we first discuss the computational aspects of ID and the choice of parameter values. ID variants which lead to improved practical performance are also explained. In Section \ref{sec:simulation}, we provide a thorough simulation study to compare ID with state-of-the-art methods. Real-life data examples are provided in Section \ref{sec:real_data}. The paper is concluded with reflections on the proposed method. The theoretical, as well as practical, merits and weaknesses of ID when compared against state-of-the-art methods are discussed throughout the paper. However, for the sake of clarity these are also brought together in Section \ref{sec:Discussion}. ID is implemented in the R packages \textbf{IDetect} and \textbf{breakfast}, available from CRAN.

\section{Motivating illustration of Isolate-Detect}
\label{sec:motivation_illustration}
The fact that each change-point is sequentially detected using an interval that contains no other change-points leads to high detection power, especially in difficult structures, such as limited spacings between consecutive change-points and/or higher-order piecewise-polynomial signals. Two examples follow in order to make clear the importance of the isolation step and to illustrate the power of ID compared to other change-point detection methods (some of those also attempt localisation) in capturing even small movements in the data that are close to each other. Table \ref{table_first_comparison} provides results on 100 replications of the continuous piecewise-linear signal (S1) and the piecewise-constant signal (S2), where
\begin{itemize}[leftmargin = 0.32in]
\item[(S1)] $T = 5200$, with 21 change-points in the slope at locations $100, 140, \cdots, 900$. The standard deviation is $\sigma=0.25$;
\item[(S2)] $T = 5200$, with 21 change-points in the mean at locations $100, 105, \cdots, 200$. The standard deviation is $\sigma=0.1$.
\end{itemize}
As a measure of the accuracy of the estimated number we give $\hat{N} - N$, while as a measure of the accuracy of the detected locations, we give Monte-Carlo estimates of the mean square error, ${\rm MSE} = T^{-1}\sum_{t=1}^{T}\E\left(\hat{f}_t - f_t\right)^2$. The methods compared are ID, NOT, and MARS for (S1) and ID, WBS, NOT, and PELT for (S2). For the ID related results in Table \ref{table_first_comparison}, we used the hybrid version of ID explained in Section \ref{subsec:hybrid}. The choice of the parameters is described in Section \ref{subsec:parameter_choice}. As already mentioned, WBS and NOT also work on subintervals of the data, chosen though in a completely different manner than in ID. More comparative simulation and real-life studies will be given in Sections \ref{sec:simulation} and \ref{sec:real_data}, respectively.
\begin{table}[h]
\centering
\caption{Distribution of $\hat{N} - N$ over 100 simulated data sequences from (S1). The average MSE is also given.}
{\small{
\begin{tabular}{|l|l|l|l|l|l|l|l|}
\cline{1-8}
 & & \multicolumn{5}{|c|}{•} & \\
 & & \multicolumn{5}{|c|}{$\hat{N} - N$} & \\
Signal & Method  & $\leq - 15$ & $(-15, -5] $ & $[-4,4]$ & $[5, 15)$ & $\geq 15$ &  MSE\\
\hline
& \textbf{ID}  & 0 & 0 & {\textbf{100}} & 0 & 0 & $13 \times 10^{-5}$\\
(S1) & NOT & 5 & 86 & 9 & 0 & 0 &  $141 \times 10^{-5}$\\
 & MARS  & 100 & 0 & 0 & 0 & 0 & $284 \times 10^{-5}$\\
\hline
\hline
& \textbf{ID}  & 0 & 1 & {\textbf{97}} & 2 & 0 & $94 \times 10^{-5}$\\
(S2) & NOT & 100 & 0 & 0 & 0 & 0 &  $485 \times 10^{-5}$\\
& PELT  & 78 & 22 & 0 & 0 & 0 & $437 \times 10^{-5}$\\
& WBS & 27 & 71 & 2 & 0 & 0 &  $413 \times 10^{-5}$\\
\hline
\end{tabular}}}
\label{table_first_comparison}
\end{table}
We notice from Table \ref{table_first_comparison} that ID offers an important increase in the change-point detection power, especially under limited spacings between consecutive change-points. Figures \ref{comparison_small_cpt_spacing_linear} and \ref{comparison_small_cpt_spacing} give a graphical representation of the results for the first out of the 100 repetitions for signals (S1) and (S2), respectively. For better presentation of the results, in (S1) the signals are presented up to $t=1000$, since after $t = 900$ there is no change-point and in all methods the estimated signal continues linearly beyond that point. For the same reason, in Figure \ref{comparison_small_cpt_spacing} which is related to (S2), the results are presented up to $t = 250$.
\begin{figure}[h]
\centering\includegraphics[width=1\textwidth,height=0.32\textheight]{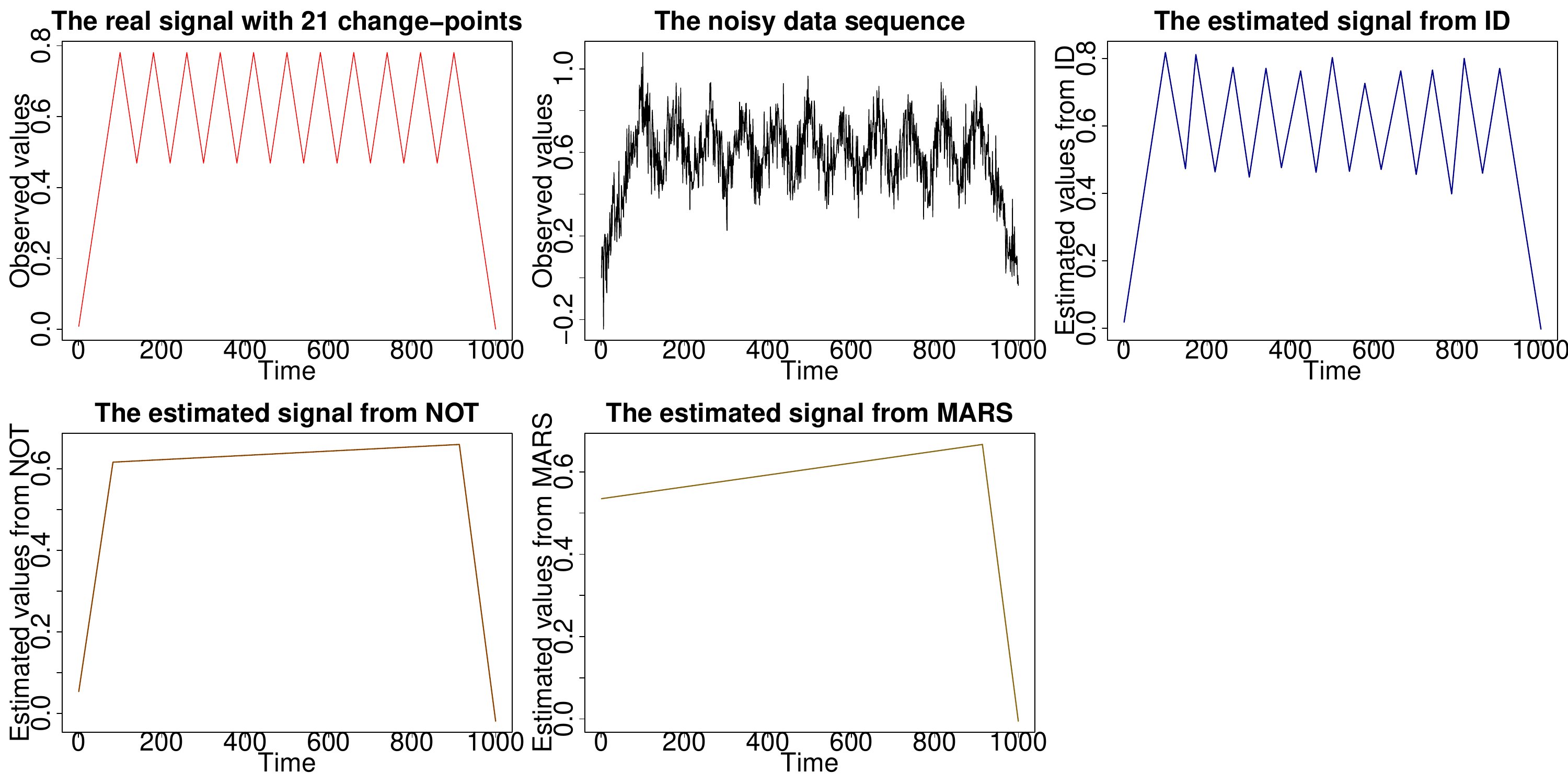}
\caption{Results (up to $t = 1000$) on estimated signals obtained by different change-point detection methods. Top row: the true signal (S1) and the data sequence, and the estimated signal using ID. Bottom row: The estimated signals from NOT, and MARS.} 
\label{comparison_small_cpt_spacing_linear}
\end{figure}
\begin{figure}[h]
\centering\includegraphics[width=1\textwidth,height=0.32\textheight]{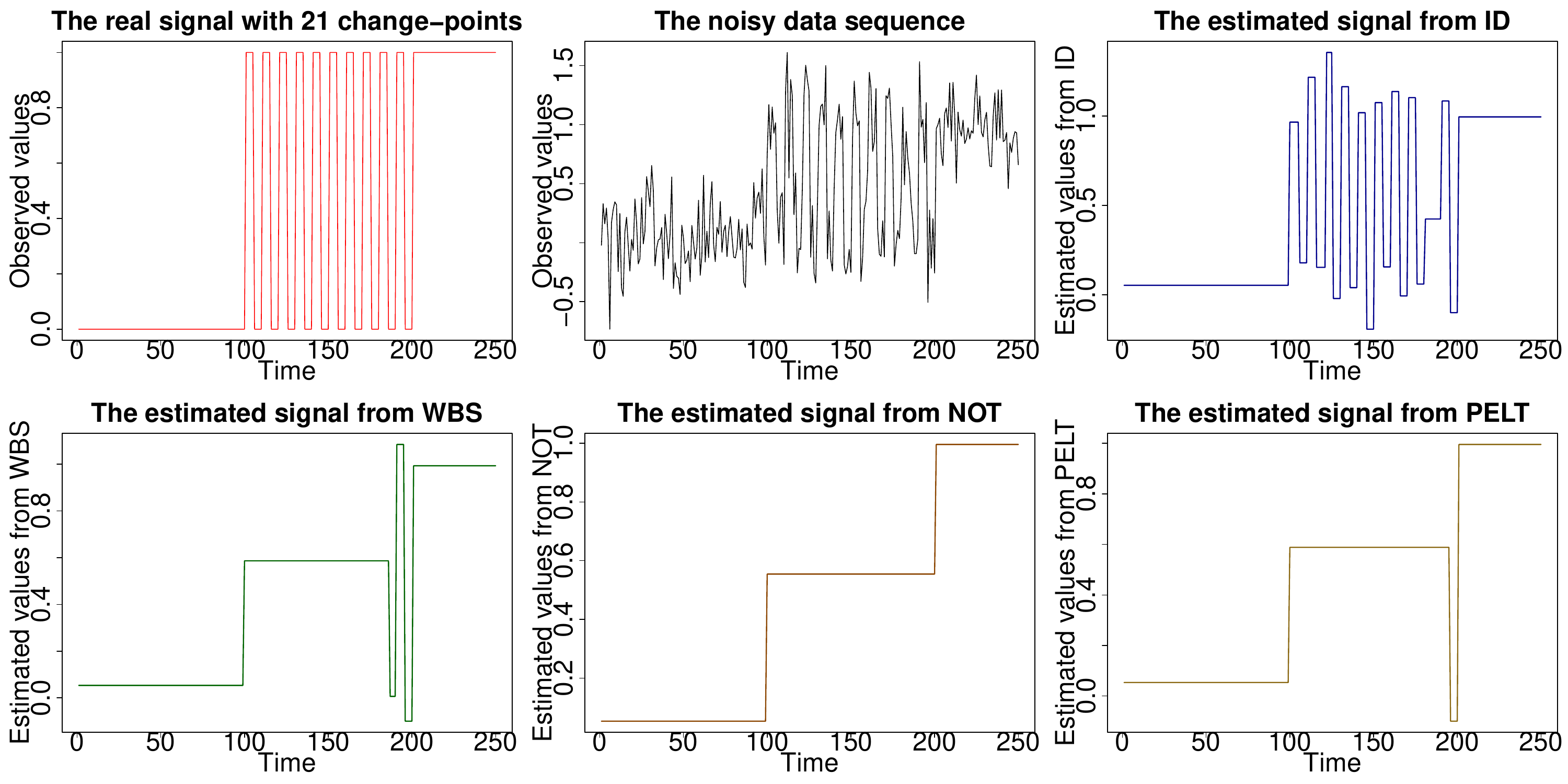}
\caption{Results (up to $t=250$) on estimated signals obtained by different change-point detection methods. Top row: the true signal (S2), the data sequence, and the estimated signal using ID. Bottom row: The estimated signals from WBS, NOT, and PELT.} 
\label{comparison_small_cpt_spacing}
\end{figure}
The NOT and WBS methods also operate on sub-intervals of the data. However, the nature of the fixed, certain (we can expand one data point at each time), localization in ID means that it is of an order of magnitude faster than the aforementioned methods, which have high computational cost that increases linearly with the number of the randomly drawn intervals. This is an issue of fundamental importance, especially in signals with a large number of change-points, in which NOT and WBS need to increase the number $M$ of intervals drawn. However, doing this also increases the computational cost. More specifically, one could try and draw all possible combinations of start- and end-points of the intervals; however, the computational complexity turns out to be cubic in $T$. In contrast, due to the explained interval expansion approach, in ID no choice of $M$ is required, which leads to better practical performance with more predictable execution times, while at the same time ID examines all possible change-point locations. We recall that unlike ID and NOT, the principle of WBS does not extend to models other than piecewise-constant. To be more precise, this generality of Isolate-Detect with respect to its applicability in many different signal structures is a main distinction between our method and recently published competing methods which, with the exception of NOT, have been developed to cover only the detection of level-changes.
\section{Methodology and Theory}
\label{sec:methodology_theory}
\subsection{Methodology}
\label{subsec:methodology}
The model is given in \eqref{our_model} and the unknown number, $N$, of change-points $r_j$ can possibly grow with $T$. Let $r_{0} = 0$ and $r_{N+1} = T$ and let $\delta_T = \min_{j=1,2,\ldots,N+1} \left|r_{j} - r_{j-1}\right|$. For clarity of exposition, we start with a simple example before providing a more thorough explanation of how ID works. Figure \ref{idea_intro} covers a specific case of two change-points, $r_1=38$ and $r_2=77$. We will be referring to Phases 1 and 2 involving six and four intervals, respectively. These are clearly indicated in the figure and they are only related to this specific example, as for cases with more change-points we would have more such phases. At the beginning, $s=1$, $e=T=100$, and we take $\lambda_T = 10$ (how to choose $\lambda_T$ will be described in Section \ref{subsec:parameter_choice}). Suppose the threshold $\zeta_T$ has been chosen well enough (more details in Section \ref{subsec:parameter_choice}) so that $r_2$ gets detected in $\left\lbrace X_{s^*}, X_{s^*+1},\ldots,X_{e}\right\rbrace$, where $s^*=71$. After the detection, $e$ is updated as the start-point of the interval where the detection occurred; therefore, $e=71$. In Phase 2 indicated in the figure, ID is applied in $[s,e]=[1,71]$. Intervals 1, 3 and 5 of Phase 1 will not be re-examined in Phase 2 and $r_1$ gets, upon a good choice of $\zeta_T$, detected in $\left\lbrace X_{s}, X_{s+1},\ldots,X_{e^*}\right\rbrace$, where $e^*=40$. After the detection, $s$ is updated as the end-point of the interval where the detection occurred; therefore, $s=40$. Our method is then applied in $[s,e] = [40,71]$; supposing there is no interval $[s^*,e^*]\subseteq [40,71]$ on which the contrast function value exceeds $\zeta_T$, the process will terminate.
\begin{figure}[h]
\centering\includegraphics[width=0.9\textwidth,height=0.25\textheight]{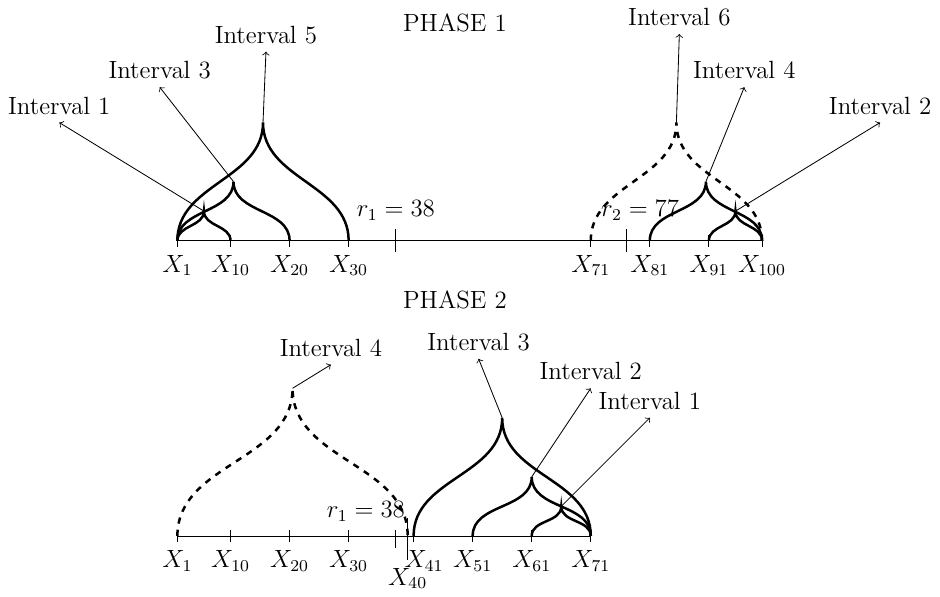}
\caption{An example with two change-points; $r_1=38$ and $r_2 = 77$. The dashed line is the interval in which the detection took place in each phase.} 
\label{idea_intro}
\end{figure}

We now describe ID more generically. For each change-point, $r_j$, ID works in two stages: Firstly, we isolate $r_j$ in an interval that contains no other change-point. To ensure this, the expansion parameter $\lambda_T$ can be taken to be as small as equal to 1. If $\lambda_T > 1$, then isolation is guaranteed with high probability. Theoretically for large $T$, the chosen value for $\lambda_T$ (this typically will be small; see Section \ref{subsec:parameter_choice} for more details) is guaranteed to be smaller than the minimum distance $\delta_T$ (which has to grow with $T$) between two consecutive change-points and isolation will be guaranteed. For an explanation on the rate of $\delta_T$ with respect to the sample size $T$, see the discussion that follows Theorem \ref{consistency_theorem}. (Of course when asymptotics is put aside, in finite samples anything can happen, and in some configurations no method can be guaranteed to detect change-points if they are arbitrarily close.) The second stage is to detect $r_j$ through the use of an appropriate contrast function. This function is, from now on, denoted by $C_{s,e}^b(\boldsymbol{X})$, and it is defined for any integer triple $(s,e,b)$, with $1\leq s \leq b < e \leq T$. Heuristically, the value of $C_{s,e}^b(\boldsymbol{X})$ is small if $b$ is not a change-point and large otherwise. In piecewise-constant signals, the contrast function reduces to the absolute value of the CUSUM statistic defined in \eqref{CUSUM}, while for continuous, piecewise-linear signals, the contrast function is given in Section \ref{subsec:theory}. For the better understanding of the method, we provide its step-by-step simple outline through pseudocode, followed by a succinct narrative of the purpose of each step.
The threshold to be used, in order to decide if a change has occurred at a specific data point, is denoted by $\zeta_T$. Practical choices for $\lambda_T$ and $\zeta_T$ are given in Section \ref{subsec:parameter_choice}. For $K=\lceil T/\lambda_T\rceil$, let $c_{j}^r = j\lambda_T$ and $c_{j}^l = T - j\lambda_T + 1$ for $j=1,2,\ldots, K-1$, while $c_{K}^r = T$ and $c_{K}^l = 1$. For a generic interval $[s,e]$, define the sequences
\begin{equation}
\label{expanding_points_s_e}
{\rm R}_{s,e} = \left[c_{k_1}^r, c_{k_1+1}^r, \ldots,e \right], \quad {\rm L}_{s,e} = \left[c_{k_2}^l, c_{k_2+1}^l, \ldots, s\right],
\end{equation}
where $k_1 := {\rm argmin}_{j \in \left\lbrace 1,2\ldots,K \right\rbrace}\left\lbrace j\lambda_T > s \right\rbrace$ and $k_2 := {\rm argmin}_{j\in \left\lbrace 1,2\ldots,K \right\rbrace}\left\lbrace T-j\lambda_T+1 < e \right\rbrace$. Denoting by $|A|$, the cardinality of any sequence $A$, and by $A(j)$ its $j^{th}$ element, the pseudocode of the main function is as below:\\
{\textbf{{\textit{Pseudocode explaining the proposed ID algorithm}}}}\\
{\textbf{function}} {\textbf{I}}{\small{SOLATE}}{\textbf{D}}{\small{ETECT}}($s, e, \lambda_T, \zeta_T$)
\begin{algorithmic}
\IF{$e - s < 1$}\STATE{STOP} \ELSE \STATE{For $j \in \left\lbrace 1,2,\ldots,\left|R\right|\right\rbrace,$ denote $\left[s_{2j-1}, e_{2j-1}\right] := \left[s, {\rm R}_{s,e}(j)\right]$\\
For $j \in \left\lbrace 1,2,\ldots,\left|{\rm L}_{s,e}\right|\right\rbrace,$ denote $\left[s_{2j}, e_{2j}\right] := \left[{\rm L}_{s,e}(j), e\right]$\\
$i = 1$\\
\textbf{(Main part)}\\
$b_{2i-1} := \underset{t \in [s_{2i-1}, e_{2i-1})}{\rm argmax}C_{s_{2i-1},e_{2i-1}}^{t}(\boldsymbol{X})$
\IF{$C_{s_{2i-1},e_{2i-1}}^{b_{2i-1}} > \zeta_T$}\STATE{add $b_{2i-1}$ to the set of estimated change-points.\\
{\textbf{I}}{\small{SOLATE}}{\textbf{D}}{\small{ETECT}}($e_{2i-1}, e, \lambda_T, \zeta_T$)}\ELSE \STATE{$b_{2i} := {\rm argmax}_{t \in [s_{2i}, e_{2i})}C_{s_{2i},e_{2i}}^{t}(\boldsymbol{X})$
\IF{$C_{s_{2i},e_{2i}}^{b_{2i}} > \zeta_T$} \STATE{add $b_{2i}$ to the set of estimated change-points.\\
{\textbf{I}}{\small{SOLATE}}{\textbf{D}}{\small{ETECT}}($s, s_{2i}, \lambda_T, \zeta_T$)}\\
 \ELSE \STATE{$i = i+1$\\
\IF {$i \leq \max\left\lbrace|{\rm L}_{s,e}|,|{\rm R}_{s,e}|\right\rbrace$}\STATE{Go back to {\textbf{(Main part)}} and repeat}\ELSE \STATE{STOP}\ENDIF}
\ENDIF}\ENDIF
}\ENDIF\\
\end{algorithmic}
{\textbf{end function}}
\vspace{0.1in}
\\
A brief explanation of the pseudocode follows. With $K$ already defined above, the intervals $[s_1,e_1], [s_2,e_2], \ldots, [s_{2K},e_{2K}]$ are those used for the isolation step. Notice that in the odd intervals $[s_1, e_1], [s_3, e_3], \ldots, [s_{2K - 1}, e_{2K - 1}]$ the start-point is fixed, unchanged, and equal to $s$, meaning that $s_1 = s_3 = \ldots = s_{2K-1} = s$. In the even intervals $[s_2, e_2], [s_4, e_4], \ldots, [s_{2K}, e_{2K}]$, it is the end-point that is kept fixed and equal to $e$, meaning that $e_2 = e_4 = \ldots = e_{2K} = e$. The process will follow until there are intervals to check. The term ``expanding intervals'' that is used throughout the paper is due to this one-sided expansion (of magnitude $\lambda_T$) of the intervals. The pseudocode makes it also clear that ID is looking for change-points interchangeably in {\textit{right}}- and {\textit{left-expanding}} intervals which, with high probability, contain at most one change-point. The Isolate-Detect procedure is launched by the call {\textbf{I}}{\small{SOLATE}}{\textbf{D}}{\small{ETECT}}($1, T, \lambda_T, \zeta_T$). 

The idea of a-posteriori change-point detection, in which change-points are detected sequentially, has appeared previously in the literature. The PS method of \cite{Venkatraman} studies the multiple change-point detection problem for the case of piecewise-constant mean signals, as well as for changes in the rate of an exponential process. The CPM method of \cite{Ross_CPM} treats change-point detection in the mean or variance of a sequence of random variables when their distribution is known. In addition, CPM can be used for distributional changes. \cite{Fang_Siegmund2020}, in a work completed after the first version of the current paper appeared on arXiv, search for significant change-points in settings such as piecewise-linear, and one of their algorithms, labelled Seq, bears some resemblance to ID; we note, however, that in addition to some algorithmic differences our aim is different as we focus on consistent estimation while \cite{Fang_Siegmund2020} on testing.

ID is conceptually and in practice different from these methods in a number of ways related to the threshold choice, the construction of the estimated change-point locations as well as the way PS, CPM, and Seq restart upon detection. Furthermore, ID's isolation technique does not appear in CPM. By contrast, we use this isolation property of ID as a device enabling its use in piecewise-(higher-order-) polynomial models. Indeed, as shown in \cite{NOT_paper}, fast segmentation of signals of the latter type is difficult to achieve unless any change-point present can be isolated away from neighbouring change-points before detection is performed, which is exactly what ID sets out to do. In particular, this paper demonstrates the use of ID in continuous piecewise-linear models. A comparison between the performance of ID and that of state-of-the-art methods is given in Section \ref{sec:simulation}. 
\subsection{Theoretical behavior of ID}
\label{subsec:theory}
The assumption of the random sequence $\left\lbrace\epsilon_t\right\rbrace_{t=1,2,\ldots,T}$ being independent and identically distributed (i.i.d.) from the Gaussian distribution is widely used in the literature. In this paper, the Gaussianity assumption is only made for technical convenience with respect to the proofs of Theorems \ref{consistency_theorem} and \ref{consistency_theorem_trend}. Relaxing both the Gaussianity and the independence assumptions in order to have time-dependent errors is a more complicated issue in terms of theory development. Recently, \cite{DepSMUCE} have attempted to treat this issue, specifically for the SMUCE approach of \cite{SMUCE}, using a reliable estimate for the long run variance, $\sigma_*^2 := \sum_{k \in \mathbb{Z}}{\rm Cov}(\epsilon_0, \epsilon_k)$, of the error distribution, which is not necessarily Gaussian.

Apart from the well-studied i.i.d. Gaussian noise structure, Isolate-Detect is explored under a variety of settings including i.i.d. non-Gaussian (see Section \ref{subsec:ht_variant}), and auto-correlated noise structures; see \cite{Fearnhead2020} who conclude that ``IDetect has very strong performance for many scenarios when either we have auto-correlated or heavy-tailed noise''.

If the standard deviation, $\sigma$, of $\epsilon_t$ is unknown, then we need to estimate it and in the cases of independent errors with the signal being piecewise-constant or piecewise-linear, $\sigma$ can be estimated via the Median Absolute Deviation (MAD) method proposed in \cite{Hampel}. For $\boldsymbol{x} = (x_1, x_2, \ldots, x_T)$, the proposed estimator, denoted by $\hat{\sigma} := C\times {\rm median}\left|\boldsymbol{x} - {\rm median}(\boldsymbol{x})\right|$, has been shown to be, for $C = 1.4826$, a consistent estimator of the population standard deviation $\sigma$ in the case of Gaussian data \citep{Rousseeuw1993}. It is very robust as evidenced by its bounded influence function and its 50\% breakdown point. 
For simplicity, let $\sigma = 1$, and \eqref{our_model} becomes
\begin{equation}
\label{new_model}
X_t = f_t + \epsilon_t, \quad t=1,2,\ldots,T.
\end{equation}
With $r_0 = 0$ and $r_{N+1} = T$, and for $j=1,2,\ldots,N+1$, we examine the theoretical behaviour of ID in the following two illustration cases:\\
\textbf{Piecewise-constant signals:} $f_t = \mu_j$ for $t = r_{j-1}+1,\ldots,r_j$, and $f_{r_j}\neq f_{r_j+1}$.\\
{\textbf{Continuous, piecewise-linear signals:}} $f_{t} = \mu_{j,1} + \mu_{j,2}t$, for $t = r_{j-1} + 1, \ldots, r_{j}$ with the additional constraint of $\mu_{k,1} + \mu_{k,2}r_{k} = \mu_{k+1,1} + \mu_{k+1,2}r_{k}$ for $k=1,2,\ldots,N$. The change-points, $r_k$, satisfy $f_{r_k-1} + f_{r_k+1}\neq 2f_{r_k}$.\\
The above scenarios are only examples of settings in which the ID methodology can be applied. The isolation aspect of the method allows its application to various different cases, such as the estimation of the number and the position of knots in piecewise polynomial signals (with or without the continuity constraint).
\vspace{0.1in}
\\
\textbf{Piecewise-constant signals.}
Under piecewise-constancy, the contrast function used is the absolute value of the CUSUM statistic, the latter being
\begin{equation}
\label{CUSUM}
\tilde{X}_{s,e}^b = \sqrt{\frac{e-b}{n(b-s+1)}}\sum_{t=s}^{b}X_t - \sqrt{\frac{b-s+1}{n(e-b)}}\sum_{t=b+1}^{e}X_t,
\end{equation}
where $1\leq s \leq b < e\leq T$ and $n=e-s+1$. Under the i.i.d. Gaussian framework used for the theoretical results presented in this paper, it can be shown that ${\rm argmax}_b\left|\tilde{X}_{s,e}^b\right| = {\rm argmax}_b \mathcal{R}_{s,e}^b(\boldsymbol{X})$, where $\mathcal{R}_{s,e}^b(\boldsymbol{X})$ is the generalized log-likelihood ratio statistic for all potential single change-points within $[s,e]$. For the main result of Theorem \ref{consistency_theorem}, we also make the following assumption.
\begin{itemize}[leftmargin = 0.32in]
\item[(A1)] The minimum distance, $\delta_T$, between two change-points and the minimum magnitude of jumps, $\underline{f}_T$, are connected by $\sqrt{\delta_T}\underline{f}_T \geq \underline{C}\sqrt{\log T}$, for a large enough constant $\underline{C}$.
\end{itemize}
The number of change-points, $N$, is assumed to be neither known nor fixed. It can grow with $T$ and the only indirect assumption on $N$ is due to the minimum distance, $\delta_T$, between two change-points in the sense that $N + 1\leq T/\delta_T$. Below, we give the theoretical result for the consistency of the number and location of the estimated change-points. The proof is in Section \ref{sec:supp_proofs} of the supplementary material.
\begin{theorem}
\label{consistency_theorem}
Let $\left\lbrace X_t \right\rbrace_{t=1,2,\ldots,T}$ follow model \eqref{new_model}, with $f_t$ being a piecewise-constant signal and assume that the random sequence $\left\lbrace\epsilon_t\right\rbrace_{t=1,2,\ldots,T}$ is independent and identically distributed (i.i.d.) from the normal distribution with mean zero and variance one and also that (A1) holds. Let $N$ and $r_j, j=1,2,\ldots,N$ be the number and locations of the change-points, while $\hat{N}$ and $\hat{r}_j, j=1,2,\ldots,\hat{N}$ are their estimates sorted in increasing order. In addition, $\Delta_j^f = \left|f_{r_j+1} - f_{r_j}\right|$, $j=1,2,\ldots, N$. Then, there exist positive constants $C_1, C_2, C_3, C_4$, which do not depend on $T$, such that for $C_1\sqrt{\log T}\leq \zeta_T < C_2\sqrt{\delta_T}\underline{f}_T$ and for a sufficiently large $T$, we obtain
\begin{equation}
\label{mainresult_theorem}
\mathbb{P}\,\left(\hat{N} = N, \max_{j=1,2,\ldots,N}\left(\left|\hat{r}_j - r_j\right|\left(\Delta_j^f\right)^2\right) \leq C_3\log T\right) \geq 1 - \frac{C_4}{T}.
\end{equation}
\end{theorem}
The isolation aspect of Isolate-Detect helps us to prove consistency under the conditions used in Theorem \ref{consistency_theorem} (and later in Theorem \ref{consistency_theorem_trend}). From \eqref{mainresult_theorem}, we notice that in order to be able to match the estimated change-point locations with the true ones, $\delta_T$ should be larger than $\underset{j=1,2,\ldots,N}{\max}\left|\hat{r}_j - r_j\right|$, meaning that $\delta_T$ must be at least $\mathcal{O}(\log T)$. For this order of $\delta_T$, \cite{Chan_Walther} argue that the smallest possible $\delta_T\underline{f}_T^2$ that allows change-point detection is $\mathcal{O}\left(\log T - \log(\log T)\right)$. In our case, assumption (A1) ensures that the $\mathcal{O}(\log T)$ rate for $\delta_T\underline{f}_T^2$ is attained, which is nearly optimal (up to the double logarithmic term). This provides evidence that ID allows for detection in complex scenarios, such as limited spacings between change-points. We mention that if $\delta_T$ is of higher order than $\mathcal{O}\left(\log T\right)$, then Assumption (A1) implies that $\underline{f}_T$ could decrease with $T$.

The quantity on the right-hand side of \eqref{mainresult_theorem} is $1-\mathcal{O}(1/T)$; the same order as in WBS and NOT. However, ID gives a provably lower constant $C_4$ for the bound. To understand this consistency advantage of our method over, for example, NOT see our proof in Section \ref{sec:supp_proofs} of the supplement and compare \eqref{mainresult_theorem2} with the result in Equation (19), p.28 in the online supplementary material of \cite{NOT_paper}. The rate of the lower bound for the threshold $\zeta_T$ is $\mathcal{O}\left(\sqrt{\log T}\right)$ and this is what will be used in practice as the default rate: we use
\begin{equation}
\label{threshold}
\zeta_T = C\sqrt{2\log T}.
\end{equation}
and the choice of the constant $C$ will be explained in Section \ref{sec:variants}. Furthermore, \eqref{mainresult_theorem} indicates that $\delta_T$ does not affect the rate of convergence of the estimated change-point locations; these only depend on $\Delta_j^f$.
\vspace{0.1in}
\\
{\raggedright{\textbf{Continuous, piecewise-linear signals.}}}
Under Gaussianity and with $R_{s,e}^b(\boldsymbol{X})$ being the generalized log-likelihood ratio for all possible single change-points within $[s,e)$, the idea is to find a contrast function $C_{s,e}^b(\boldsymbol{X})$, which is maximized at the same point as $R_{s,e}^b(\boldsymbol{X})$. The contrast function is constructed by taking inner products of the data with a contrast vector. In the case of continuous piecewise-linear signals, \cite{NOT_paper} show that the contrast vector to be used is $\boldsymbol{\phi_{s,e}^b} = \left(\phi_{s,e}^b(1),\ldots,\phi_{s,e}^b(T)\right)$, where
\begin{equation}
\label{contrast_vectorCPLM}
\phi_{s,e}^b(t) = \begin{cases}
    \alpha_{s,e}^b\beta_{s,e}^b\left[(e+2b-3s+2)t - (be +bs - 2s^2+2s)\right], & \hspace{-0.06in} t=s,\ldots,b,\\
    -\frac{\alpha_{s,e}^b}{\beta_{s,e}^b}\left[(3e-2b-s+2)t - (2e^2+2e-be-bs)\right], &\hspace{-0.06in} t=b+1,\ldots,e,\\
    0, & \hspace{-0.04in}\text{otherwise},
  \end{cases}
\end{equation}
where $n=e-s+1$, $\alpha_{s,e}^b = (6/[n(n^2-1)(1+(e-b+1)(b-s+1)+(e-b)(b-s))])^{\frac{1}{2}}$ and $\beta_{s,e}^b = ([(e-b+1)(e-b)]/[(b-s+1)(b-s)])^{\frac{1}{2}}$. The contrast function is $C_{s,e}^b(\boldsymbol{X}) = \left|\left\langle\boldsymbol{X},\boldsymbol{\phi_{s,e}^b}\right\rangle\right|$. To explain the reasoning behind the choice of the triangular function $\phi_{s,e}^b(\cdot)$, we define, for the interval $[s,e]$, the linear vector 
$\gamma_{s,e}(t) = \left(\frac{1}{12}(e-s+1)\left(e^2 - 2es +2e +s^2-2s\right)\right)^{-\frac{1}{2}}\left(t-\frac{e+s}{2}\right)$, $t=s,\ldots,e$ (and 0 otherwise)
as well as the constant vector
$1_{s,e}(t) = \left(e-s+1\right)^{-\frac{1}{2}}$, $t=s,\ldots,e$ (and 0 otherwise).
On the vector
$\tilde{\phi}_{s,e}^b(t) = t - b$, $t=b+1,\ldots,e$ (and 0 otherwise), 
which is linear with a kink at $b+1$, we apply the Gram-Schmidt orthogonalization with respect to $\gamma_{s,e}(t)$ and $1_{s,e}(t)$. Normalizing the obtained vector such that $\|.\|_2 = 1$ returns the contrast vector $\phi_{s,e}^b(t)$ defined in \eqref{contrast_vectorCPLM}. The best approximation, in terms of the Euclidean distance, of $X_t$ in $[s,e]$ is a linear combination of $\gamma_{s,e}(t)$, $1_{s,e}(t)$, and $\phi_{s,e}(t)$, which are mutually orthonormal \citep{NOT_paper}. This orthonormality leads to $R_{s,e}^b(\boldsymbol{X}) = \left|\left\langle\boldsymbol{X},\boldsymbol{\phi_{s,e}^b}\right\rangle\right| = C_{s,e}^b(\boldsymbol{X})$. For the consistency of ID in continuous piecewise-linear signals, we make the following assumption.
\begin{itemize}[leftmargin = 0.32in]
\item[(A2)] The minimum distance, $\delta_T$, between two change-points and the minimum magnitude of jumps, $\underline{f}_T = \min_{j=1,2,\ldots,N}\left|2f_{r_j} - f_{r_{j}+1} - f_{r_{j}-1}\right|$, are connected by the requirement $\delta_T^{3/2}\underline{f}_T \geq C^*\sqrt{\log T}$, for a large enough constant $C^*$.
\end{itemize}
The term $\delta_T^{3/2}\underline{f}_T$ characterizes the difficulty level of the detection problem and is analogous to $\sqrt{\delta_T}\underline{f}_T$ in the scenario of piecewise-constant signals. Theorem \ref{consistency_theorem_trend} gives the consistency result for the case of continuous piecewise-linear signals. The proof is in Section \ref{sec:supp_proofs} of the supplement.
\begin{theorem}
\label{consistency_theorem_trend}
Let $\left\lbrace X_t \right\rbrace_{t=1,2,\ldots,T}$ follow model \eqref{new_model} with $f_t$ being a continuous, piecewise-linear signal and assume that the random sequence $\left\lbrace\epsilon_t\right\rbrace_{t=1,2,\ldots,T}$ is independent and identically distributed (i.i.d.) from the normal distribution with mean zero and variance one and that (A2) holds. We denote by $N$ and $r_j, j=1,2,\ldots,N$ the number and locations of the change-points, while $\hat{N}$ and $\hat{r}_j, j=1,2,\ldots,\hat{N}$ are their estimates sorted in increasing order. Also, we denote $\Delta_j^f = \left|2f_{r_j} - f_{r_{j}+1} - f_{r_{j}-1}\right|$. Then, there exist positive constants $C_1, C_2, C_3, C_4$, which do not depend on $T$, such that for $C_1\sqrt{\log T}\leq \zeta_T < C_2\delta_T^{3/2}\underline{f}_T$ and for sufficiently large $T$,
\begin{equation}
\label{mainresult_theorem_trend}
\mathbb{P}\,\left(\hat{N} = N, \max_{j=1,2,\ldots,N}\left(\left|\hat{r}_j - r_j\right|\left(\Delta_j^f\right)^{2/3}\right) \leq  C_3(\log T)^{1/3}\right) \geq 1 - \frac{C_4}{T}.
\end{equation}
\end{theorem}
The quantity on the right-hand side of \eqref{mainresult_theorem_trend} is $1-\mathcal{O}\left(1/T\right)$. In addition, in the case of $\underline{f}_T \sim T^{-1}$, ID's change-point detection accuracy is $\mathcal{O}\left(T^{2/3}\left(\log T\right)^{1/3}\right)$, as can be seen from \eqref{mainresult_theorem_trend}. This differs from the $\mathcal{O}\left(T^{2/3}\right)$ rate derived in \cite{Raimondo} only by the logarithmic factor. The lower bound of the threshold is $\mathcal{O}\left(\sqrt{\log T}\right)$. Therefore,
\begin{equation}
\label{threshold_linear}
\zeta_T = \tilde{C}\sqrt{2\log T},
\end{equation}
where $\tilde{C}$ is a constant and we will comment on its choice in Section \ref{subsec:parameter_choice}.

ID is flexible because it does not depend on the structure of the signal; what changes is the choice of an appropriate contrast function. Adopting a similar approach as the one for the case of continuous piecewise-linear signals, one can construct contrast functions for the detection of other types of features.
\subsection{Information Criterion approach}
\label{subsec:sSIC}
Misspecification of the threshold in the ID algorithm can lead to the misestimation of the number of change-points. To remedy this, we develop an approach which starts by possibly overestimating the number of change-points and then creates a solution path, with the estimates ordered according to a certain predefined criterion. The best fit is then chosen, based on the optimization of a model selection criterion.
\vspace{0.05in}
\\
{\textbf{The solution path algorithm:}} The estimated number of change-points depends on $\zeta_T$ and this allows us to denote $\hat{N} = \hat{N}(\zeta_T)$. For given data, we employ ID using first $\zeta_T$ and then $\tilde{\zeta}_T$, where $\tilde{\zeta}_T < \zeta_T$. Let $C_{\tilde{\zeta}_T}$ and $\tilde{C}_{\tilde{\zeta}_T}$ be the $\tilde{\zeta}_T$-associated constants in \eqref{threshold} and \eqref{threshold_linear}, respectively. With $J \geq \hat{N}(\zeta_T)$, we estimate $\tilde{r}_j, j=1,2,\ldots, J$, which are sorted in increasing order in $\tilde{S} = \left[\tilde{r}_1, \tilde{r}_2, \ldots, \tilde{r}_{J}\right]$. Our aim is to prune the estimates through an iterative procedure, where at each iteration the estimation most likely to be spurious is removed. The algorithm is split into four parts, with their descriptions being fairly technical. We note however that the different parts are very similar and are based on the idea of removing change-points according to their contrast function values as well as their distance to neighbouring estimates. Even though the full explanation of each part is in Section \ref{supp:solution_path_algorithm_parts} of the supplement, we now provide a brief summary for the framework of the solution path algorithm. With $\tilde{r}_0 = 1$ and $\tilde{r}_{J+1} = T$, we first collect triplets $(\tilde{r}_{j-1},\tilde{r}_j,\tilde{r}_{j+1})$, $\forall\left\lbrace 1,2,\ldots, J\right\rbrace$ and we calculate ${\rm CS}(\tilde{r}_j) := C_{\tilde{r}_{j-1},\tilde{r}_{j+1}}^{\tilde{r}_j}(\boldsymbol{X}),$ with $C_{s,e}^b(\boldsymbol{X})$ being the relevant contrast function. For $m = {\rm argmin}_{j}\left\lbrace {\rm CS}(\tilde{r}_j) \right\rbrace$ we check whether $CS(\tilde{r}_m) \leq \tilde{\tilde{C}}\sqrt{\log T}$, for $\tilde{\tilde{C}} > 0$; in the proofs of Theorems \ref{consistency_theorem2} and \ref{consistency_theorem_trend2}, $\tilde{\tilde{C}} = 2\sqrt{2}$ but smaller values could be sufficient; see for example Corollary \ref{cor:bound}. If $CS(\tilde{r}_m) \leq \tilde{\tilde{C}}\sqrt{\log T}$, we remove $\tilde{r}_m$ from $\tilde{S}$, reduce $J$ by 1, relabel the remaining estimates (in increasing order) in $\tilde{S}$, and repeat this estimate removal process, which is carried out in a way such that once the set $\tilde{S}$ contains $N$ estimates, then for $j=1,2,\ldots,N$, each $\tilde{r}_j$ is within a distance of $C^*\left(\log T\right)^{\alpha}$ from the true change-point $r_j$. We keep removing estimates until $\tilde{S} = \emptyset$. 

At the end of this change-point removal approach, we collect the estimates in 
\begin{equation}
\label{orderedcpts}
\boldsymbol{b} = \left(b_1,b_2,\ldots,b_J\right),
\end{equation}
where $b_J$ is the estimate that was removed first, $b_{J-1}$ is the one that was removed second, and so on. From now on, the vector $\boldsymbol{b}$ is called the solution path and is used to give a range of different fits. We define the collection $\left\lbrace\mathcal{M}_j\right\rbrace_{j = 0,1,\ldots,J}$ where $\mathcal{M}_{0} = \emptyset$ and $\mathcal{M}_j = \left\lbrace b_1,b_2,\ldots,b_j\right\rbrace$. For $j = 2,\ldots,J$, let $\tilde{b}_1 < \ldots < \tilde{b}_j$ be the sorted elements of $\mathcal{M}_j$. Among the collection of models $\left\lbrace \mathcal{M}_j\right\rbrace_{j=0,1,\ldots,J}$, we propose to select the one that minimizes the strengthened Schwarz Information Criterion \citep{Liu, Fryzlewicz_WBS}, defined as
\begin{equation}
\label{sSIC}
{\rm sSIC}(j) = -2\sum_{k=1}^{j+1}\ell\left(X_{\tilde{b}_{k-1} + 1},\ldots,X_{\tilde{b}_k};\hat{\theta}_k\right) + n_{j}\left(\log T\right)^{\alpha},
\end{equation}
where $\tilde{b}_0 = 0$ and for each collection $\mathcal{M}_j$, $\tilde{b}_{j+1} = T$ and $\hat{\theta}_1,\hat{\theta}_2,\ldots,\hat{\theta}_{j+1}$ are the maximum likelihood estimators of the segment parameters for the model \eqref{new_model} with change-point locations $b_1,b_2,\ldots,b_j$. The quantity $n_j$ is the total number of estimated parameters related to $\mathcal{M}_j$. For example, if we do not consider the change-point locations as free parameters, then in the scenario of piecewise-constant mean $n_j = j+1$ (the constant values for each of the $j+1$ segments), while in the scenario of continuous and piecewise-linear signals $n_j = j+2$ (the starting intercept and slope and the $j$ changes in the slope). We mention that if the continuity constraint is to be removed, then $n_j$ would be equal to $2j+2$ (the constant and slope values for the $j+1$ segments). If now we consider the change-point locations to be free parameters, then we just need to add $j$ in the above values for $n_j$ in the different scenarios.

In the algorithm we have referred to three parameters: $C^*, \tilde{\tilde{C}}$, and $\alpha$. Although we do not give a recipe for the choice of $C^*$ and $\tilde{\tilde{C}}$, Section \ref{subsec:var_proposed} describes how to circumvent their choice. With respect to $\alpha$, taking its value to be equal to 1 in \eqref{sSIC} gives the standard SIC penalty, but our theory requires $\alpha > 1$. In practice we use $\alpha = 1.01$ in order to remain close to SIC. Theorems \ref{consistency_theorem2} and \ref{consistency_theorem_trend2} below give the consistency results for the piecewise-constant and continuous piecewise-linear models, based on the sSIC approach. The proof of Theorem \ref{consistency_theorem2} is in the supplementary material and the same approach can be followed to prove Theorem \ref{consistency_theorem_trend2}.
\begin{theorem}
\label{consistency_theorem2}
Let $\left\lbrace X_t \right\rbrace_{t=1,2,\ldots,T}$ follow model \eqref{new_model} under piecewise-constancy and let the assumptions of Theorem \ref{consistency_theorem} hold. Let $N$ and $r_j, j=1,2,\ldots,N$ be the number and locations of the change-points. Let $N \leq J$, where $J$ can also grow with $T$. In addition, let $\alpha > 1$ be such that $\left(\log T\right)^\alpha = o(\delta_T\underline{f}_T^2)$ is satisfied, where $\delta_T$ and $\underline{f}_T$ are defined in (A1). With $\left\lbrace\mathcal{M}_{j}\right\rbrace_{j=0,1,\ldots,J}$ being the set of candidate models obtained by the solution path algorithm, we define $\hat{N} = {\rm argmin}_{j=0,1,\ldots,J}\;{\rm sSIC}(j)$. Then, there exist positive constants $C_1, C_2$, which do not depend on $T$, such that for $\Delta_j^f = \left|f_{r_j+1} - f_{r_j}\right|$,
\begin{equation}
\label{mainresult_theorem_ic}
\mathbb{P}\,\left(\hat{N} = N, \max_{j=1,2,\ldots,N}\left(\left|\hat{r}_j - r_j\right|\left(\Delta_j^f\right)^2\right) \leq C_1\left(\log T\right)^{\alpha}\right) \geq 1 - \frac{C_2}{T}.
\end{equation}
\end{theorem}
\begin{theorem}
\label{consistency_theorem_trend2}
Let $\left\lbrace X_t \right\rbrace_{t=1,2,\ldots,T}$ follow model \eqref{new_model} under continuous piecewise-linearity and let the assumptions of Theorem \ref{consistency_theorem_trend} hold. Let $N$ and $r_j, j=1,2,\ldots,N$ be the number and locations of the change-points. Let $N \leq J$, where $J$ can also grow with $T$. In addition, let $\alpha > 1$ be such that $\left(\log T\right)^\alpha = o(\delta_T^3\underline{f}_T^2)$ is satisfied, where $\delta_T$ and $\underline{f}_T$ are defined in (A2). With $\left\lbrace\mathcal{M}_{j}\right\rbrace_{j=0,1,\ldots,J}$ being the set of candidate models obtained by the solution path algorithm, we define $\hat{N} = {\rm argmin}_{j=0,1,\ldots,J}\;{\rm sSIC}(j)$. Then, there exist positive constants $C_1, C_2$, which do not depend on $T$, such that for $\Delta_j^f = \left|2f_{r_j} - f_{r_{j}+1} - f_{r_{j}-1}\right|$,
\begin{equation}
\label{mainresult_theorem_trend_ic}
\mathbb{P}\,\left(\hat{N} = N, \max_{j=1,2,\ldots,N}\left(\left|\hat{r}_j - r_j\right|\left(\Delta_j^f\right)^{2/3}\right) \leq C_1(\log T)^{\alpha/3}\right) \geq 1 - \frac{C_2}{T}.
\end{equation}
\end{theorem}
We note that our solution path algorithm, explained in detail in Section \ref{supp:solution_path_algorithm_parts} of the supplementary material, allows $J$, the number of the detections from the already explained overestimation process, to grow with $T$. The quantities on the right hand sides of \eqref{mainresult_theorem_ic} and \eqref{mainresult_theorem_trend_ic} are $1-\mathcal{O}\left(1/T\right)$; the same order as those in \eqref{mainresult_theorem} and \eqref{mainresult_theorem_trend}. The lowest admissible $\delta_T\underline{f}_T^2$ and $\delta_T^3\underline{f}_T^2$ in Theorems \ref{consistency_theorem2} and \ref{consistency_theorem_trend2}, respectively, are slightly larger than the same quantities in the thresholding approach. Our empirical expertise suggests that SIC-based approaches tend to exhibit better practical behaviour for signals that have a moderate number of change-points and/or large spacings between them. A hybrid that combines the advantages of the thresholding and the SIC-based approach is introduced in Section \ref{subsec:hybrid}.
\section{Computational complexity and practicalities}
\label{sec:variants}
\subsection{Computational cost}
\label{subsec:Compcost}
With $\delta_T$ being the minimum distance between two change-points, and $\lambda_T$ the interval-expansion parameter, we use $\lambda_T < \delta_T$. We note that while $\delta_T$ is unknown, choosing $\lambda_T$ small enough guarantees with high probability that this requirement holds; see Section \ref{subsec:parameter_choice} for how to choose $\lambda_T$ in order to obtain good accuracy performance and at the same time low computational cost. Now, since $K = \lceil T/\lambda_T\rceil > \lceil T/\delta_T \rceil$ and the total number, $M_{ID}$, of intervals required to scan the data is no more than $2K$ ($K$ intervals from each expanding direction), in the worst case scenario we have $M_{ID} = 2K > 2\left\lceil\frac{T}{\delta_T}\right\rceil$.
As a comparison, in WBS and NOT one needs to draw at least $M$ intervals where $M \geq \left(9T^2/\delta_T^2\right)\log\left(T^2/\delta_T\right)$. The lower bound for $M$ in WBS and NOT is $\mathcal{O}\left(T^2/\delta_T^2\right)$ up to a logarithmic factor, whereas the lower bound for $M_{ID}$ is $\mathcal{O}\left(T/\delta_T\right)$. This results in great speed gains of ID over WBS and NOT. The reason behind this significant difference in the computational complexity of the methods is that in WBS and NOT both the start- and end-points of the randomly drawn intervals have to be chosen, whereas in ID, depending on the expanding direction, we keep the start- or the end-point fixed.
\subsection{Parameter choice}
\label{subsec:parameter_choice}
{\textbf{Choice of the threshold constant.}} We start with an upper bound on the constant $C$, as defined in \eqref{threshold}, for the case of piecewise-constant signals when the error terms $\epsilon_t$ are i.i.d. from the Gaussian distribution. We note that this result is of independent interest. Our model is as in \eqref{our_model} for stationary $\epsilon_t$. For any vector $\boldsymbol{y} \in \mathbb{R}^{T}$, we define
\begin{equation}
\tilde{y}_{s,e}^b  = \sqrt{\frac{e-b}{n(b-s+1)}}\sum_{t=s}^{b}y_t - \sqrt{\frac{b-s+1}{n(e-b)}}\sum_{t=b+1}^{e}y_t;\quad
\tilde{\tilde{y}}_{s,e} = \frac{\sum_{t=s}^e y_t}{(e-s+1)^{1/2}},
\label{notation_Corollary}
\end{equation}
where $1 \leq s \leq b < e \leq T$ and $n = e-s+1$. It can be shown that if $\epsilon_t$, are serially independent and their distribution is symmetric about zero (for example i.i.d. standard Gaussian random variables), then the sequence $\left\lbrace\epsilon_t\right\rbrace_{t=1}^T$ satisfies
\begin{equation}
\label{assumption_corollary}
\forall\,\gamma>0,\quad P\left(\min_{s,b,e}\,\, \tilde{\tilde{\epsilon}}_{s,b} \tilde{\tilde{\epsilon}}_{b+1,e} < -\gamma \right) \le  P\left(\max_{s,b,e}\,\, \tilde{\tilde{\epsilon}}_{s,b} \tilde{\tilde{\epsilon}}_{b+1,e} > \gamma \right)
\end{equation}
The following corollary indicates that as $T\rightarrow\infty$, we have that $C \leq \sqrt{3/2}$, meaning that the threshold can be taken to be at most $\sqrt{3\log T}$. This value of $\sqrt{3}$ is smaller than the constant used in the solution path algorithm of Section \ref{subsec:sSIC} ($\tilde{\tilde{C}} = 2\sqrt{2}$), which can however be used to give explicit upper bounds on the consistency results as explained in Theorems \ref{consistency_theorem} and \ref{consistency_theorem2}; in contrast, Corollary \ref{cor:bound} does not give an explicit upper bound for the probability related to the consistency result as expressed in \eqref{corollary_result}. We highlight that the aforementioned bound on the constant and its proof are simpler than the results presented in \cite{Fang_Li_Siegmund} which involve the manipulation of complex distributions. The proof is in the supplementary material.
\begin{corollary}
\label{cor:bound}
Let $\{\epsilon_t\}_{t=1}^T$ be i.i.d. $N(0, \sigma^2)$. For any $\delta > 0$,
\begin{equation}
\label{corollary_result}
\Prob\left( \exists_{s,a,e}\,\, \left(\tilde{\epsilon}_{s,e}^b\right)^2 > 3 \sigma^2 (1 + \delta) \log\,T   \right) \xrightarrow[T \rightarrow \infty]{} 0.
\end{equation}
\end{corollary}
For the practical choice of the values of $C$ and $\tilde{C}$, in \eqref{threshold} and \eqref{threshold_linear}, respectively, we ran a large-scale simulation study involving a wide range of signals. The number of change-points, $N$, was generated from the Poisson distribution with rate parameter $N_{\alpha}\in \left\lbrace 4,8,12 \right\rbrace$. For  $T \in \left\lbrace 100,200,500,1000,2000,5000 \right\rbrace$, we uniformly distributed the change-points in $\left\lbrace 1,2,\ldots,T \right\rbrace$. Then, for piecewise-constant (or continuous piecewise-linear) signals, at each change-point location we introduced a jump (or a slope change) which followed the normal distribution with mean zero and variance $\sigma^2 \in \left\lbrace 1,3,5 \right\rbrace$. Standard Gaussian noise was then added onto the simulated signal. For each value of $N_\alpha$, $\sigma^2$ and $T$ we generated 1000 replicates and estimated the number of change-points using ID with threshold $\zeta_T$ as in \eqref{threshold} and \eqref{threshold_linear} for a variety of constant values $C$ and $\tilde{C}$. The best behaviour occurred when, approximately, $C=1.05$ and $\tilde{C} = 1.4$. These values will be referred to as the default constants and they hold true for all signals that satisfy the assumption of the error terms $\epsilon_t$ being i.i.d. Gaussian. We note that the value of $\tilde{C} = 1.4$ does not violate Corollary \ref{cor:bound} because the result expressed in the latter is only for piecewise-constant signals, while the constant $\tilde{C}$ applies to the scenario of continuous, piecewise-linear signals.
Due to the fact that the contrast function used is based on local averaging, the CLT can be used to show that for sufficiently large sample size $T$, ID is robust when the normality assumption is not satisfied; this has also been explored in \cite{Fearnhead2020}. Also, pre-averaging is a practical approach that we employed in Section \ref{subsec:ht_variant} for such cases with error departures from Gaussianity.

In the SIC-based approach of Section \ref{subsec:sSIC}, we started by detecting change-points using threshold $\tilde{\zeta}_T < \zeta_T$. In practice, we take the constants related to $\tilde{\zeta}_T$, namely $C_{\tilde{\zeta}_T}$ and $\tilde{C}_{\tilde{\zeta}_T}$ as defined in Section \ref{subsec:sSIC}, to be 0.9 and 1.25, respectively.
\vspace{0.05in}
\\
{\textbf{Choice of the expansion parameter $\lambda_T$.}} We start by highlighting that our numerical experience suggests that ID is robust to small changes in the value of $\lambda_T$; for a small-scale simulation study when the value of $\lambda_T$ changes significantly ($\lambda_T \in \left\lbrace 5,20,80\right\rbrace$), see Section \ref{sec:investigation_lambda} of the supplementary material. Theoretically, for a given signal, the change-point detection results obtained from ID are the same for any value of $\lambda_T$ used which is less than the minimum spacing between two successive change-points. The computational cost of running ID is inversely proportional to the size of the expansion parameter; the smaller the $\lambda_T$, the more intervals we need to work on. However, the low computational complexity of our algorithm allows us to take $\lambda_T$ to be as small as the value of three leading to very good accuracy even for signals with frequent change-points. We now give example execution times for two models, (T1) and (T2) defined below, on a 3.60GHz CPU with 16 GB of RAM. We employed the ID-variant for long signals explained in Section \ref{subsec:var_proposed}.
\begin{itemize}[leftmargin = 0.32in]
\item[(T1)] Length $l_j = 7\times 10^{j}, j=3,4,5$, with change-points at $7, 14, \ldots, l_j - 7$ and values between them $0,4,0,4,\ldots, 0, 4$. The standard deviation is $\sigma=0.5$.
Execution times: 0.31s ($j=3$), 2.25s ($j=4$), 26.41s ($j=5$).
\item[(T2)] Length $l_j = 7 \times 10^{j}, j=3,4,5$, with no change-points. We use $\sigma=1$. Execution times: 0.64s ($j=3$), 3.01s ($j=4$), 30.35s ($j=5$).
\end{itemize}
\subsection{Variants}
\label{subsec:var_proposed}
Here, we describe three different ways to further improve ID's practical performance.
\vspace{0.05in}
\\ 
{\textit{Long signals:}} If $T$ is large, we split the given data sequence uniformly into smaller parts (windows), to which ID is then applied. In practical implementations, the length of the window is 3000 and we apply this structure only when $T > 12000$, because for smaller values of $T$ there are no significant differences in the execution times of ID and its window-based variant. The computational improvement that this structure offers is explained in Section \ref{sec:Supp_wind_impr} of the supplement.
\vspace{0.05in}
\\
{\textit{Restarting after detection:}} In practice, instead of starting from the end-point $e^*$ (or start-point $s^*$) of the right-expanding (or left-expanding) interval where a detection occurred, we could start from the estimated change-point, $\hat{b}$. This alternative, labelled ${\rm ID}_{det}$, leads to accuracy improvement without affecting the speed of the method.
\vspace{0.05in}
\\
{\textit{Faster solution path algorithm:}} In practice, we use only Part 4 of the solution path algorithm described in Section \ref{supp:solution_path_algorithm_parts} of the supplement because it is quicker and conceptually simpler; it requires only the choice of $\alpha$ and tends not to affect ID's accuracy.
\subsection{Alternative model selection criteria}
\label{subsec:hybrid}
{\textit{A hybrid between thresholding and SIC stopping rules:}} 
For signals with a large number of regularly occurring change-points, the threshold-based ID tends to behave better than the SIC-based procedure. As explained after Theorems \ref{consistency_theorem2} and \ref{consistency_theorem_trend2}, this is unsurprising because SIC-based approaches typically perform better on signals with a moderate number of change-points separated by larger spacings. This difference in ID's behaviour between the threshold- and SIC-based versions is what motivates us to introduce a hybrid of these two stopping rules with minimal parameter choice, which works as follows. Firstly, we estimate the change-points using the threshold approach ${\rm ID}_{det}$ with $\lambda_T^{th} = 3$. If the estimated change-points are more than a constant $J^*$, then the result is accepted and we stop. Otherwise, the hybrid method proceeds to detect the change-points using the SIC-based approach with $\lambda_T > \lambda_T^{th}$, since the already-applied thresholding rule has not suggested a signal with many change-points. In the simulations, we use $J^* = 100$, $\lambda_T=10$.\\
\vspace{0.05in}
\\
{\textit{Steepest Drop to Low Levels (SDLL):}}
We also combine ID with the SDLL model selection method introduced in \cite{Fryzlewicz_WBS2}.

\subsection{Extension to different noise structures}
\label{subsec:ht_variant}
This section describes how to use ID when the noise is not Gaussian. We pre-process the data in order to obtain a noise structure that is closer to Gaussianity. For a given scale number $s$ and data $\left\lbrace X_t \right\rbrace_{t=1,2,\ldots,T}$, let $Q= \lceil T/s \rceil$ and $\tilde{X}_q = \frac{1}{s}\sum_{t=(q-1)s+1}^{qs}X_t$, for $q=1,2,\ldots,Q-1$, while $\tilde{X}_Q = (T - (Q-1)s)^{-1}\sum_{t=(Q-1)s+1}^{T}X_t$. We apply ID on $\left\lbrace\tilde{X}_q\right\rbrace_{q=1,2,\ldots,Q}$ to obtain the estimated change-points, namely $\tilde{\tilde{r}}_1,\tilde{\tilde{r}}_2,\ldots,\tilde{\tilde{r}}_{\hat{N}}$, in increasing order. To estimate the original locations of the change-points we define $\hat{r}_k = \left(\tilde{\tilde{r}}_k-1\right)s + \left\lfloor \frac{s}{2} + 0.5 \right\rfloor, \; k=1,2,\ldots,\hat{N}$. The larger the value of $s$, the closer the distribution of the noise to normal, but the more the amount of pre-processing. In simulations presented in Section \ref{sec:simulation}, we use $s=3$ for the case of Student-$t_5$ distributed noise, while if the tails are heavier (Student-$t_3$), we set $s=5$. The hybrid version of ID will be employed on $\left\lbrace\tilde{X}_q\right\rbrace_{q=1,2,\ldots,Q}$ and in order to be consistent with the choice of the expansion parameter, we take $\lambda^*_T = \left\lfloor \lambda_T/s\right\rfloor$. In practice, for unknown noise, our recommendation is to set $s=5$.\\
\section{Simulations}
\label{sec:simulation}
This section compares the performance of ID with competitors. The main change-point detection $\mathsf{R}$ functions in the competing packages were called using their default input arguments, which does not always allow direct like-for-like comparisons of the methods.
Whenever needed (difficult signal structures), and in order to help the competitors capture their best possible performance, the input values were adjusted accordingly. The $\mathsf{R}$ code used for the simulation study is available from Github at \url{https://github.com/Anastasiou-Andreas/IDetect/blob/master/R/Simulations_used.R}.  Table \ref{table_methods} shows the competitors used. 
\begin{table}[h]
\centering
\caption{The competing methods used in the simulation study}
\begin{tabular}{|l|l|l|l|}
\cline{1-4}
 & & & \\
Type of signal & Method notation & Reference & \textsf{R} package\\
\hline
 & PELT & \cite{Killick_PELT} & \textbf{changepoint}\\
 & NP.PELT & \cite{Haynes_npPelt} & \textbf{changepoint.np}\\
 & S3IB  & \cite{Rigaill} & \textbf{Segmentor3IsBack}\\
 & CumSeg  & \cite{Muggeo_Adelfio} & \textbf{cumSeg}\\
Piecewise-constant & CPM & \cite{Ross_CPM} & \textbf{cpm}\\
 & WBS & \cite{Fryzlewicz_WBS} & \textbf{wbs}\\
 & WBS2 & \cite{Fryzlewicz_WBS2} & \textbf{breakfast}\\
 & NOT & \cite{NOT_paper} & \textbf{not}\\
 & FDR & \cite{LiMunk} & \textbf{FDRSeg}\\
 & TGUH & \cite{Fryzlewicz_tail-greedy} & \textbf{breakfast}\\
\hline
 & NOT & \cite{NOT_paper} & \textbf{not}\\
 & TF & \cite{KimTF} & -\\
Continuous piecewise-linear & CPOP & \cite{FearnheadCPOP} & -\\
 & MARS & \cite{Friedman} & \textbf{earth}\\
 & FKS & \cite{Spiriti} & \textbf{freeknotsplines}\\
\hline
\end{tabular}
\label{table_methods}
\end{table}
CPOP is employed based on $\mathsf{R}$ code found in \url{http://www.research.lancs.ac.uk/portal/en/datasets/cpop(56c07868-3fe9-4016-ad99-54439ec03b6c).html} and TF in \url{https://stanford.edu/~boyd/l1_tf}. For WBS, we give results based on both the information criterion and the thresholding (for $C=1$) stopping rules. The notation is WBSIC and WBSC1, respectively. With respect to WBS2, its performance is investigated based on the SDLL model selection criterion introduced in \cite{Fryzlewicz_WBS2}. In the \textbf{cpm} package, the threshold is decided through the average run length (ARL) until a false positive occurs. In our simulations, we give results for ${\rm ARL}=500$ (the default value) and if the signal length, $l_s$, is greater than 500, results are also given for ${\rm ARL} = 1000\lceil ls/1000 \rceil$. The notation is CPM.$l.A$, with $A$ the value of ARL. For FKS, when the number of knots is unknown (the scenario we work in), we need to specify the maximum allowed number of knots. We take this to be 2$N$, with $N$ the true number of change-points. Also, the estimated change-points by FKS are positive real numbers; we take as estimation the closest integer. The proposed ID version  is the hybrid described in Section \ref{subsec:hybrid}. However, we also present the results for two more variants: SDLL and thresholding with constant $\sqrt{3/2}$ (see \eqref{threshold}), which is the upper bound proven in Corollary \ref{cor:bound}. The notation for these variants is ID.SDLL and ID$_{\sqrt{3/2}}$, respectively.
\vspace{0.05in}
\\
{\textbf{A seemingly difficult structure for ID:}} Signals that present the most difficulty to ID are ones in which change-points are concentrated in the middle part of the data and offset each other, as in Figure \ref{hatfunction}. The reason is that due to the left- and right- expanding feature of ID, where one of the two end-points of the interval is kept fixed, the change-points need to be detectable based on relatively ``unbalanced'' (explanation follows directly below) tests, which typically tend to offer poor power. For example, referring again to Figure \ref{hatfunction}, the change-point at 490 will need to be isolated and detected by comparing the means of the data over the long interval $[1,490]$ and a short interval of the form $[491,e_j]$, where $e_j \leq 510$ is the end-point of a right-expanding interval $[1,e_j]$. To be more precise, if the expansion parameter $\lambda = 3$, then $e_j \in \left\lbrace 492, 495, \ldots, 510\right\rbrace$ and therefore our procedure will have seven opportunities to detect the change-point 490 while it is still isolated in intervals that do not contain any other change-points. Even though ID would be expected to struggle in detecting the change-points in such unbalanced intervals, our numerical experience suggests that its performance on such challenging signals is in fact very good and matches or surpasses that of the best competitors; see for example the results in Table \ref{models_Piotr} for the model (M4), which follows this structure.
\begin{figure}[h]
\centering
\includegraphics[width=0.8\textwidth,height=0.18\textheight]{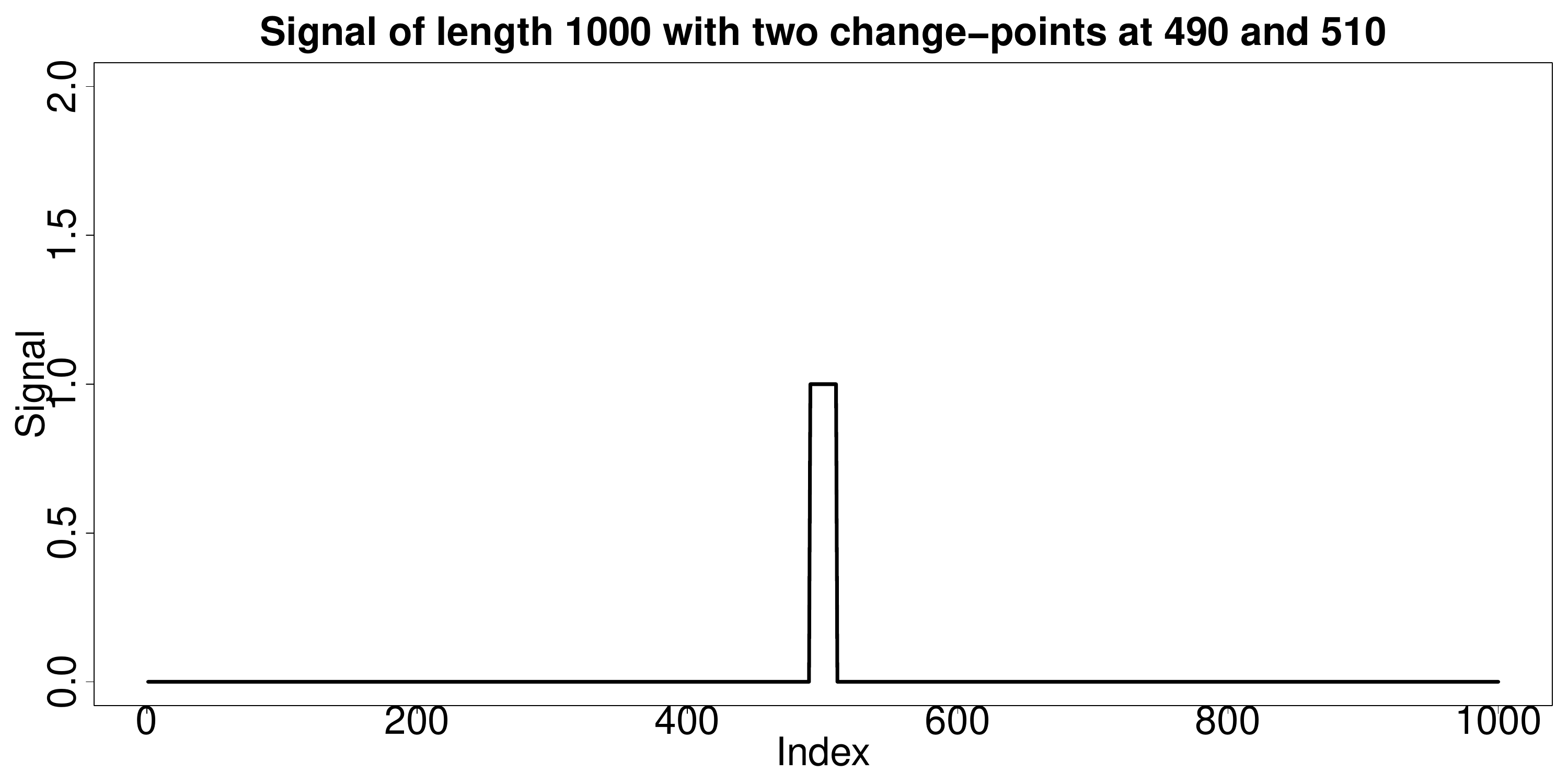}
\caption{Example of a signal of length 1000 with change-points at 490 and 510 offsetting each other.} 
\label{hatfunction}
\end{figure}
All the signals are fully specified in Section \ref{supp:signals} of the supplementary material. Figure \ref{figure_signals} shows examples of the data generated by models (M1) \textit{blocks}, (M2) \textit{teeth}, (M4) \textit{middle-points}, and (W1) \textit{wave 1}. Tables \ref{models_Piotr}--\ref{models_linearW3} summarize the results in the case of i.i.d$.$ Gaussian noise. Table \ref{models_Piotr_ht} presents the behaviour of ID under the setting of i.i.d$.$ scaled Student-$t_d$ noise, where $d=3,5$. More examples are in the supplement.
\begin{figure}[h]
\centering
\includegraphics[width=1\textwidth,height=0.29\textheight]{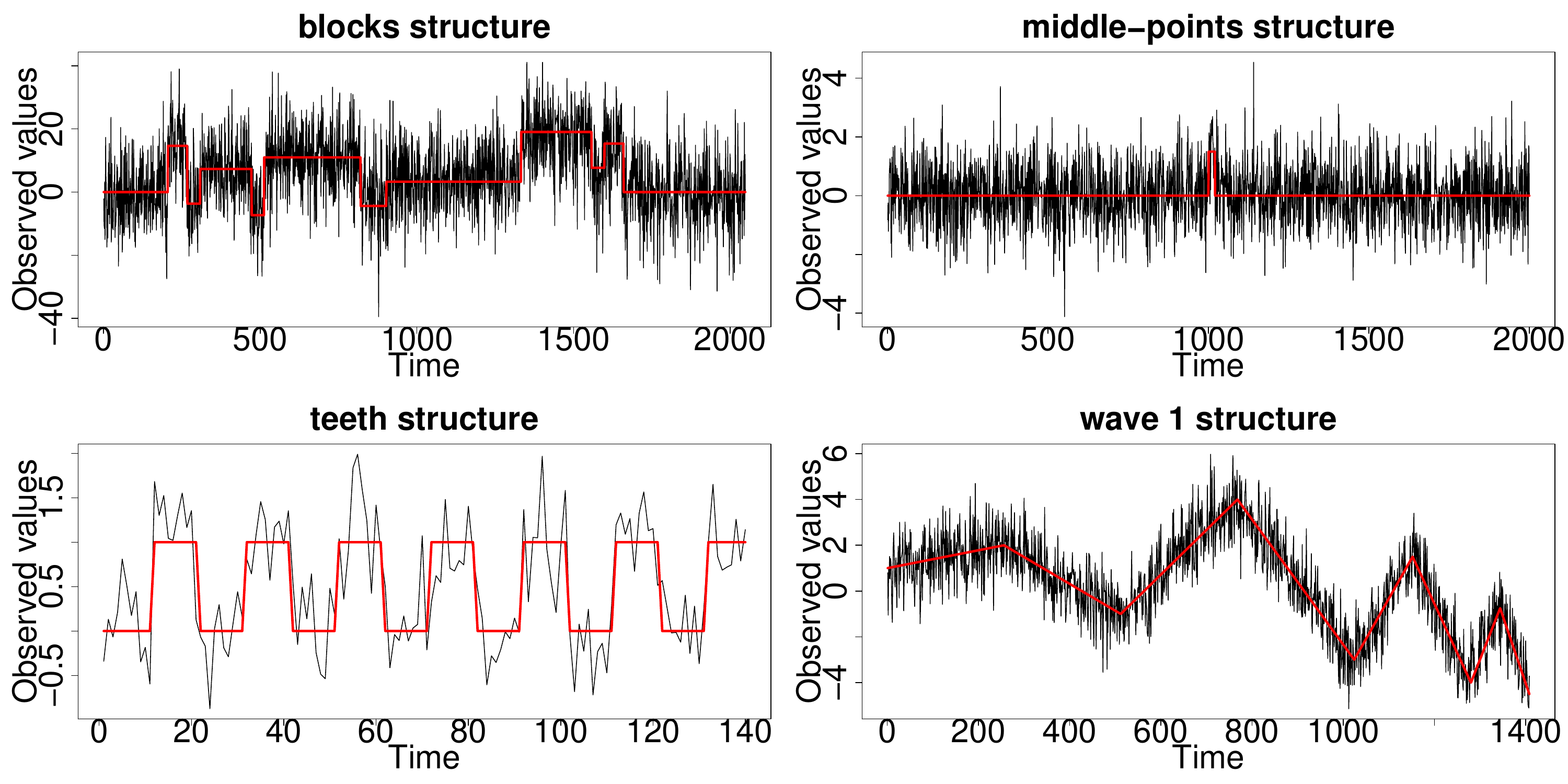}
\caption{Examples of data series, used in simulations. The true signal, $f_t$, is in red.} 
\label{figure_signals}
\end{figure}

We highlight that the NOT, WBSIC, and S3IB methods require the specification of the maximum number, $K_{max}$, of change-points allowed to be detected. If the default values in these methods are lower than the true number of change-points in the simulated examples, then we take $K_{max} = \lceil T/\delta_T\rceil$, where $\delta_T$ is the minimum distance between two change-points. We ran 100 replications for each signal and the frequency distribution of $\hat{N} - N$ for each method is presented. The methods with the highest empirical frequency of $\hat{N} - N = 0$ (or in a neighbourhood of zero, depending on the example) and those within $10\%$ off the highest are given in bold. As a measure of the accuracy of the detected locations, we provide Monte-Carlo estimates of the mean squared error, ${\rm MSE} = T^{-1}\sum_{t=1}^{T}\E\left(\hat{f}_t - f_t\right)^2$, where $\hat{f}_t$ is the ordinary least square approximation of $f_t$ between two successive change-points. In continuous piecewise-linear signals, $\hat{f}_t$ is the splines fit obtained using the \textbf{splines} package in \textsf{R}. The scaled Hausdorff distance, $d_H = n_s^{-1}\max\left\lbrace \max_j\min_k\left|r_j-\hat{r}_k\right|,\max_k\min_j\left|r_j-\hat{r}_k\right|\right\rbrace,$ where $n_s$ is the length of the largest segment, is also given in all examples apart from the signal (NC) in Table \ref{sim_supp2}, which is a constant-mean signal with no change-points.

The average computational time for all methods, apart from FDR, is also provided. FDR is excluded due to its non-uniform procedure in terms of the execution speed for each signal (if a newly obtained signal has length greater than previously treated signals, then FDR estimates the threshold by 5000 Monte-Carlo simulations, which makes it slow). In some cases the average computational time for FKS is not given. We have already explained that we need to pre-specify the maximum allowed number of knots in order for FKS to work. The method is somewhat slow and we exclude the results for FKS when the true change-points are more than 10, as in such cases it would take a significant amount of time to finish all the 100 simulations.
\begin{table}[H]
\centering
\caption{Distribution of $\hat{N} - N$ over 100 simulated data sequences of the piecewise-constant signals (M1)-(M4). The average MSE, $d_H$ and computational time are also given}
{\small{
\begin{tabular}{|l|l|l|l|l|l|l|l|l|l|l|l|}
\cline{1-12}
 & & \multicolumn{7}{|c|}{•} & & & \\
 & & \multicolumn{7}{|c|}{$\hat{N} - N$} & & & \\
Method & Model & $\leq -3$ & $-2$ & $-1$ & 0 & 1 & 2 & $\geq 3$ & MSE & $d_H$ & Time (ms)\\
\hline
PELT &  & 6 & 32 & 50 & 12 & 0 & 0 & 0 & 3.23 & 0.14 & 3\\
NP.PELT &  & 0 & 2 & 27 & 49 & 15 & 5 & 2 & 2.82 & 0.10 & 211.8\\
S3IB  &  & 0 & 7 & 38 & 54 & 1 & 0 & 0 & 2.49 & 0.08 & 343.2\\
CumSeg  &  & 39 & 21 & 38 & 2 & 0 & 0 & 0 & 6.37 & 0.20 & 62.3\\
CPM.$l.$500  &  & 0 & 0 & 0 & 3 & 3 & 4 & 90 & 4.45 & 0.44 & 2.3\\
CPM.$l.$3000  &  & 0 & 0 & 8 & 41 & 26 & 19 & 6 & 3.03 & 0.19 & 3.3\\
WBSC1 & (M1) & 0 & 0 & 11 & 32 & 27 & 19 & 11 & 2.79 & 0.25 & 99.3\\
WBSIC &  & 0 & 3 & 37 & 53 & 7 & 0 & 0 & 2.59 & 0.08 & 99.3\\
WBS2 &  & 0 & 3 & 54 & 31 & 8 & 2 & 2 & 2.64 & 0.09 & 623.3\\
NOT  &  & 0 & 3 & 51 & 43 & 3 & 0 & 0 & 2.61 & 0.10 & 80.7\\
FDR  &  & 0 & 0 & 33 & 54 & 12 & 1 & 0 & 2.51 & 0.09 & -\\
TGUH &  & 0 & 5 & 37 & 49 & 7 & 1 & 1 & 3.30 & 0.08 & 127.4\\
\textbf{ID} &  & 0 & 3 & 30 & \textbf{62} & 5 & 0 & 0 & 2.66 & 0.08 & 23.9\\
ID.SDLL &  & 1 & 2 & 59 & 28 & 5 & 3 & 2 & 2.80 & 0.10 & 20\\
ID$_{ \sqrt{3/2}}$ &  & 0 & 9 & 62 & 28 & 1 & 0 & 0 & 2.75 & 0.09 & 22.3\\
\hline
PELT &  & 85 & 6 & 0 & 9 & 0 & 0 & 0 & 181 $\times 10^{-3}$ & 6.62 & 1.1\\
NP.PELT &  & 84 & 12 & 3 & 1 & 0 & 0 & 0 & 165 $\times 10^{-3}$ & 4.26 & 3.1\\
S3IB  &  & 41 & 15 & 1 & 43 & 0 & 0 & 0 & 117 $\times 10^{-3}$ & 3.73 & 15.2\\
CumSeg  &  & 100 & 0 & 0 & 0 & 0 & 0 & 0 & 251 $\times 10^{-3}$ & - & 3.9\\
CPM.$l.$500 &  & 78 & 4 & 15 & 3 & 0 & 0 & 0 & 145 $\times 10^{-3}$ & 2.96 & 0.4\\
{\textbf{WBSC1}} & (M2) & 1 & 2 & 7 & {\textbf{72}} & 12 & 6 & 0 & 53 $\times 10^{-3}$ & 0.33 & 38.2\\
{\textbf{WBSIC}} &  & 7 & 8 & 1 & {\textbf{68}} & 13 & 3 & 0 & 64 $\times 10^{-3}$ & 1.00 & 38.2\\
{\textbf{WBS2}} &  & 3 & 3 & 4 & {\textbf{71}} & 10 & 4 & 5 & 58 $\times 10^{-3}$ & 0.363 & 30.5\\
{\textbf{NOT}} &  & 9 & 7 & 4 & {\textbf{73}} & 6 & 1 & 0 & 65 $\times 10^{-3}$ & 0.97 & 43.4\\
FDR &  & 14 & 11 & 11 & 55 & 7 & 2 & 0 & 71 $\times 10^{-3}$ & 0.80 & -\\
{\textbf{TGUH}} &  & 4 & 18 & 3 & {\textbf{68}} & 7 & 0 & 0 & 64 $\times 10^{-3}$ & 0.47 & 22.8\\
\textbf{ID} &  & 7 & 7 & 1 & \textbf{74} & 11 & 0 & 0 & 60 $\times 10^{-3}$ & 0.87 & 8.8\\
ID.SDLL &  & 5 & 5 & 6 & 63 & 8 & 4 & 9 & 62 $\times 10^{-3}$ & 0.43 & 3.7\\
ID$_{ \sqrt{3/2}}$ &  & 28 & 13 & 9 & 47 & 3 & 0 & 0 & 84 $\times 10^{-3}$ & 0.90 & 5.3\\
\hline
\textbf{PELT} &  & 0 & 2 & 7 & \textbf{90} & 1 & 0 & 0 & 23 $\times 10^{-3}$ & 0.15 & 1.1\\
NP.PELT &  & 100 & 0 & 0 & 0 & 0 & 0 & 0 & 781 $\times 10^{-3}$ & 1.78 & 4.2\\
S3IB  &  & 98 & 1 & 1 & 0 & 0 & 0 & 0 & 213 $\times 10^{-3}$ & 0.91 & 20.2\\
CumSeg  &  & 0 & 3 & 16 & 72 & 9 & 0 & 0 & 65 $\times 10^{-3}$ & 0.32 & 5.2\\
CPM.$l.$500 &  & 1 & 6 & 87 & 6 & 0 & 0 & 0 & 51 $\times 10^{-3}$ & 0.85 & 0.2\\
WBSC1 & (M3) & 0 & 0 & 0 & 66 & 26 & 7 & 1 & 24 $\times 10^{-3}$ & 0.19 & 37.3\\
WBSIC &  & 0 & 0 & 0 & 64 & 27 & 9 & 0 & 24 $\times 10^{-3}$ & 0.18 & 37.3\\
{\textbf{WBS2}} &  & 0 & 0 & 1 & {\textbf{87}} & 8 & 2 & 2 & 25 $\times 10^{-3}$ & 0.17 & 34.7\\
\textbf{NOT} &  & 0 & 0 & 0 & \textbf{93} & 7 & 0 & 0 & 21 $\times 10^{-3}$ & 0.13 & 118.3\\
FDR &  & 0 & 0 & 2 & 77 & 15 & 5 & 1 & 23 $\times 10^{-3}$ & 0.17 & -\\
\textbf{TGUH} &  & 0 & 0 & 1 & \textbf{91} & 6 & 2 & 0 & 25 $\times 10^{-3}$ & 0.15 & 25.2\\
\textbf{ID} &  & 0 & 0 & 0 & \textbf{91} & 8 & 1 & 0 & 22 $\times 10^{-3}$ & 0.13 & 9.8\\
\textbf{ID.SDLL} &  & 0 & 0 & 1 & \textbf{97} & 1 & 0 & 1 & 24 $\times 10^{-3}$ & 0.14 & 6.8\\
{\textbf{ID$\mathbf{_{ \sqrt{3/2}}}$}} &  & 0 & 0 & 2 & {\textbf{94}} & 3 & 1 & 0 & 23 $\times 10^{-3}$ & 0.15 & 4.4\\
\hline
PELT & & - & 53 & 0 & 47 & 0 & 0 & 0 & 14 $\times 10^{-3}$ & 0.54 & 6.7\\
NP.PELT & & - & 0 & 0 & 21 & 3 & 34 & 42 & 14 $\times 10^{-3}$ & 0.47 & 395.2\\
\textbf{S3IB} & & - & 12 & 0 & \textbf{87} & 1 & 0 & 0 & 7 $\times 10^{-3}$ & 0.12 & 292.1\\
CumSeg & & - & 100 & 0 & 0 & 0 & 0 & 0 & 23 $\times 10^{-3}$ & -
 & 84.6\\
CPM.$l.$500 & & - & 0 & 0 & 0 & 0 & 6 & 94 & 31 $\times 10^{-3}$ & 0.76 & 14\\
CPM.$l.$2000 & & - & 0 & 0 & 35 & 11 & 22 & 32 & 13 $\times 10^{-3}$ & 0.39 & 20.4\\
WBSC1 & (M4) & - & 0 & 0 & 23 & 20 & 17 & 40 & 13 $\times 10^{-3}$ & 0.50 & 120.8\\
\textbf{WBSIC} & & - & 4 & 0 & \textbf{96} & 1 & 0 & 0 & 5 $\times 10^{-3}$ & 0.04 & 119.2\\
WBS2 & & - & 0 & 1 & 83 & 10 & 4 & 2 & 5 $\times 10^{-3}$ & 0.09 & 666.4\\
\textbf{NOT} & & - &  8 & 0 & \textbf{92} & 0 & 0 & 0 & 6 $\times 10^{-3}$ & 0.08 & 61.8\\
FDR & & - & 0 & 19 & 70 & 10 & 1 & 0 & 9 $\times 10^{-3}$ & 0.07 & -\\
TGUH & & - & 0 & 51 & 40 & 7 & 2  & 0 & 23 $\times 10^{-3}$ & 0.28 & 169.2\\
\textbf{ID} &  & - & 7 & 0 & \textbf{93} & 0 & 0 & 0 & 6 $\times 10^{-3}$ & 0.07 & 42.3\\
ID.SDLL & & - & 0 & 0 & 81 & 4 & 10 & 5 & 7 $\times 10^{-3}$ & 0.10 & 28.7\\
{\textbf{ID$\mathbf{_{ \sqrt{3/2}}}$}} & & - & 1 & 0 & {\textbf{98}} & 1 & 0 & 0 & 5 $\times 10^{-3}$ & 0.05 & 66.4\\
\hline
\end{tabular}}}
\label{models_Piotr}
\end{table}
\begin{table}[H]
\centering
\caption{Distribution of $\hat{N} - N$ over 100 simulated data sequences from the piecewise-constant signal (M5). The average MSE, $d_H$ and computational time are also given}
{\small{
\begin{tabular}{|l|l|l|l|l|l|l|l|l|}
\cline{1-9}
 & \multicolumn{5}{|c|}{•} & & & \\
 & \multicolumn{5}{|c|}{$\hat{N} - N$} & & & \\
Method  & $\leq -500$ & $(-500,-50]$ & $(-50,-10)$ & $[-10,10]$ & $> 10$ & MSE & $d_H$ & Time (s)\\
\hline
PELT  & 100  & 0 & 0 & 0 & 0 & 1.97 & 114.92 & 0.033\\
NP.PELT  & 100 & 0 & 0 & 0 & 0 & 2.25 & 551.89 & 8.976\\
S3IB  & 99 & 1 & 0 & 0 & 0 & 2.23 & 1979.95 & 332.841\\
CumSeg  & 100 & 0 & 0 & 0 & 0 & 2.25 & 1999 & 0.551\\
CPM.$l.$500 & 0 & 45 & 54 & 1 & 0 & 0.19 & 9.00 & 0.002\\
CPM.$l.$20000  & 100 & 0 & 0 & 0 & 0 & 2.23 & 1999 & 1.245\\
WBSC1  & 100 & 0 & 0 & 0 & 0 & 1.51 & 35.26 & 12.272\\
WBSIC & 100 & 0 & 0 & 0 & 0 & 2.25 & 1999 & 12.272\\
WBS2 & 0 & 0 & 0 & {\textbf{100}} & 0 & 0.14 & 0.54 & 5.796\\
NOT  & 100 & 0 & 0 & 0 & 0 & 2.25 & 1999 & 0.484\\
FDR  & 0 & 0 & 0 & 5 & 95 & 0.14 & 0.51 & -\\
\textbf{TGUH} & 0 & 0 & 0 & \textbf{100} & 0 & 0.16 & 0.84 & 0.794\\
\textbf{ID} & 0 & 0 & 0 & \textbf{100} & 0 & 0.14 & 0.99 & 0.785\\
\textbf{ID.SDLL} & 0 & 0 & 0 & \textbf{100} & 0 & 0.14 & 0.71 & 120.601\\
ID$_{ \sqrt{3/2}}$ & 0 & 82 & 18 & 0 & 0 & 0.22 & 2.48 & 1.363\\
\hline
\end{tabular}}}
\label{distribution_M2}
\end{table}
\begin{table}[H]
\centering
\caption{Distribution of $\hat{N} - N$ over 100 simulated data sequences from (NC). Also the average MSE and computational times for each method are given}
{\small{
\begin{tabular}{|l|l|l|l|l|l|l|}
\cline{1-7}
 & \multicolumn{4}{|c|}{•} & & \\
 & \multicolumn{4}{|c|}{$\hat{N} - N$} & & \\
Method  &  0 & 1 & 2 & $\geq 3$ & MSE & Time (s)\\
\hline
\textbf{PELT}  & \textbf{100} & 0 & 0 & 0 & 39 $\times 10^{-5}$ & 0.004\\
NP.PELT & 8 & 1 & 23 & 68 & 999 $\times 10^{-5}$ & 1.077\\
\textbf{S3IB} & \textbf{100} & 0 & 0 & 0 & 39 $\times 10^{-5}$ & 0.715\\
\textbf{CumSeg}  & \textbf{100} & 0 & 0  & 0 & 39 $\times 10^{-5}$ & 0.115\\
CPM.$l.$500  & 0 & 0 & 0 & 100 & 2957 $\times 10^{-5}$ & 0.011\\
CPM.$l.$3000  & 28 & 6 & 39 & 27 & 628 $\times 10^{-5}$ & 0.031\\
WBSC1  & 15 & 18 & 20 & 47 & 653 $\times 10^{-5}$ & 0.149\\
\textbf{WBSIC}  & \textbf{99} & 1 & 0 & 0 & 44 $\times 10^{-5}$ & 0.149\\
WBS2  & 89 & 5 & 4 & 2 & 82 $\times 10^{-5}$ & 0.958\\
\textbf{NOT}  & \textbf{99} & 1 & 0 & 0 & 44 $\times 10^{-5}$ & 0.089\\
\textbf{FDR}  & \textbf{96} & 4 & 0 & 0 & 47 $\times 10^{-5}$ & -\\
\textbf{TGUH} & \textbf{100} & 0 & 0 & 0 & 39 $\times 10^{-5}$ & 0.217\\
\textbf{ID}  & \textbf{100} & 0 & 0 & 0 & 39 $\times 10^{-5}$ & 0.172\\
{\textbf{ID.SDLL}}  & {\textbf{90}} & 4 & 0 & 6 & 182 $\times 10^{-5}$ & 0.069\\
{\textbf{ID$\mathbf{_{ \sqrt{3/2}}}$}} & {\textbf{99}} & 0 & 1 & 0 & 41 $\times 10^{-5}$ & 0.259\\
\hline
\end{tabular}}}
\label{sim_supp2}
\end{table}
\begin{table}[H]
\centering
\caption{Distribution of $\hat{N} - N$ over 100 simulated data sequences from the continuous piecewise-linear signals (W1), (W3), and (W4). The average MSE, $d_H$ and computational time for each method are also given}
{\small{
\begin{tabular}{|l|l|l|l|l|l|l|l|l|l|l|l|}
\cline{1-12}
 & & \multicolumn{7}{|c|}{•} &  &  & \\
 & & \multicolumn{7}{|c|}{$\hat{N} - N$} &  &  & \\
Method & Model & $\leq -3$ & $-2$ & $-1$ & 0 & 1 & 2 & $\geq 3$ & MSE & $d_H$ & Time (s)\\
\hline
\textbf{NOT} & & 0 & 0 & 0 & \textbf{99} & 1 & 0 & 0 & 0.016 & 0.063 & 0.343\\
TF  & & 0 & 0 & 0 & 0 & 0 & 0 & 100 & 0.029 & 0.451 & 1.125\\
\textbf{CPOP} & & 0 & 0 & 0 & \textbf{99} & 1 & 0 & 0 & 0.013 & 0.055 & 23.190\\
MARS  & (W1) & 0 & 0 & 2 & 9 & 42 & 39 & 8 & 0.034 & 0.200 & 0.011\\
FKS  & & 0 & 0 & 0 & 72 & 22 & 6 & 0 & 0.015 & 0.109 & 270.385\\
\textbf{ID} & & 0 & 0 & 0 & \textbf{91} & 9 & 0 & 0 & 0.030 & 0.104 & 0.036\\
\textbf{ID.SDLL} & & 0 & 0 & 0 & \textbf{98} & 0 & 1 & 1 & 0.033 & 0.098 & 0.030\\
\hline
NOT  &  & 0 & 0 & 27 & 0 & 6 & 18 & 49 & 0.035 & 0.571 & 0.163\\
TF  & & 0 & 0 & 0 & 0 & 0 & 0 & 100 & 606.523 & 0.432 & 0.117\\
\textbf{CPOP} &  & 0 & 0 & 0 & \textbf{90} & 6 & 2 & 2 & 0.010 & 0.097 & 0.078\\
MARS  & (W3) & 91 & 0 & 7 & 2 & 0 & 0 & 0 & 3.991 & 2.258 & 0.008\\
{\textbf{FKS}} &  & 0 & 0 & 0 & {\textbf{90}} & 9 & 1 & 0 & 0.010 & 0.097 & 67.582\\
\textbf{ID} &  & 0 & 0 & 0 & \textbf{99} & 1 & 0 & 0 & 0.013 & 0.101 & 0.017\\
\textbf{ID.SDLL} &  & 0 & 0 & 0 & \textbf{93} & 4 & 1 & 2 & 0.022 & 0.130 & 0.010\\
\hline
NOT  & & 0 & 1 & 14 & 20 & 16 & 20 & 29 & 0.109 & 0.998 & 0.958\\
TF  & & 0 & 0 & 0 & 0 & 0 & 0 & 100 & 660.399 & 0.465 & 1.349\\
\textbf{CPOP}  & (W4) & 0 & 0 & 0 & \textbf{92} & 8 & 0 & 0 & 0.015 & 0.084 & 1.627\\
MARS  &  & 100 & 0 & 0 & 0 & 0 & 0 & 0 & 22.058 & 1.609 & 0.019\\
\textbf{ID} &  & 0 & 0 & 0 & \textbf{92} & 8 & 0 & 0 & 0.038 & 0.123 & 0.045\\
\textbf{ID.SDLL} &  & 0 & 0 & 0 & \textbf{92} & 4 & 1 & 3 & 0.062 & 0.120 & 0.025\\
\hline
\end{tabular}}}
\label{models_linearW1W2}
\end{table}
\vspace{-0.11in}
\begin{table}[H]
\centering
\caption{Distribution of $\hat{N} - N$ over 100 simulated data sequences of the continuous piecewise-linear signal (W2). The average MSE, $d_H$ and computational time for each method are also given}
{\small{
\begin{tabular}{|l|l|l|l|l|l|l|l|l|l|l|}
\cline{1-11}
 & \multicolumn{7}{|c|}{•} & & & \\
 & \multicolumn{7}{|c|}{$\hat{N} - N$} & & & \\
Method & $\leq -90$ & $(-90,-1)$ & $-1$ & 0 & $1$ & $(1,60]$ & $> 60$ & MSE & $d_H$ & Time (s)\\
\hline
NOT  & 100 & 0 & 0 & 0 & 0 & 0 & 0 & 4.731 & 99 & 0.869\\
TF  & 0 & 0 & 0 & 0 & 0 & 0 & 100 & 212.547 & 0.387 & 0.863\\
\textbf{CPOP}  & 0 & 0 & 0 & \textbf{97} & 3 & 0 & 0 & 0.162 & 0.189 & 1.161\\
MARS  & 100 & 0 & 0 & 0 & 0 & 0 & 0 & 4.703 & 98.523 & 0.009\\
\textbf{ID} & 0 & 0 & 0 & \textbf{98} & 2 & 0 & 0 & 0.201 & 0.242 & 0.589\\
\textbf{ID.SDLL} & 0 & 0 & 0 & \textbf{98} & 2 & 0 & 0 & 0.256 & 0.287 & 0.097\\
\hline
\end{tabular}}}
\label{models_linearW3}
\end{table}
\vspace{-0.2in}
\begin{table}[H]
\centering
\caption{ID results for the distribution of $\hat{N} - N$ for the models (M2)-(M4) and (W1), over 100 simulations where the distribution of the noise is Student-$t_d$, for $d=3,5$. The average MSE, $d_H$ and computational time are also given}
{\small{
\begin{tabular}{|l|l|l|l|l|l|l|l|l|l|l|l|}
\cline{1-12}
 & & \multicolumn{7}{|c|}{•} &  &  & \\
 & & \multicolumn{7}{|c|}{$\hat{N} - N$} & & &\\
$d$ & Model & $\leq -3$ & $-2$ & $-1$ & 0 & 1 & 2 & $\geq 3$ & MSE & $d_H$ & Time (ms)\\
\hline
 & (M2) & 6 & 2 & 2 & 74 & 9 & 5 & 2 & $60\times 10^{-3}$ & 0.86 & 9.7\\
5 & (M3) & 0 & 0 & 0 & 75 & 16 & 5 & 4 & $21 \times 10^{-3}$ & 0.16 & 9.2\\
 & (W1) & 0 & 0 & 0 & 86 & 12 & 2 & 0 & $31 \times 10^{-3}$ & 0.23 & 32.8\\
\hline
 & (M2) & 7 & 1 & 2 & 52 & 21 & 8 & 9 & $71\times 10^{-3}$ & 1.18 & 8.7\\
3 & (M3) & 0 & 1 & 0 & 59 & 20 & 13 & 7 & $26 \times 10^{-3}$ & 0.22 & 9.8\\
 & (W1) & 0 & 0 & 0 & 62 & 28 & 4 & 6 & $32 \times 10^{-3}$ & 0.25 & 22.6\\
\hline
\end{tabular}}}
\label{models_Piotr_ht}
\end{table}

With regards to piecewise-constancy, ID is always in the top $10\%$ of the best methods when considering accuracy in any aspect (estimation of $N$, MSE, $d_H$); in most cases it is the best method overall. ID.SDLL is also, in most cases, in the top 10\% of the best performing methods; this provides evidence that the Isolate-Detect algorithm can be combined with various model selection criteria (thresholding, SIC, SDLL) and maintain a good practical behaviour. When the threshold constant, $C$, is equal to $\sqrt{3/2}$, the behaviour of ID remains good for signals that have a moderate number of change-points that are not near each other. As we can see from Table \ref{distribution_M2},
 ID$_{ \sqrt{3/2}}$ seems to struggle in scenarios with a large number of frequently occurring change-points. In continuous piecewise-linear signals, CPOP, ID, and ID.SDLL are in all cases in the top $10\%$ of the best methods in terms of the accurate estimation of $N$. In terms of the MSE and $d_H$, CPOP is by a narrow margin the overall best method, with ID and ID.SDLL coming second and third, respectively. We can deduce that our method exhibits uniformity in detecting with high accuracy the change-points for various different signal structures, a characteristic which is at least partly absent from the majority of its competitors. Furthermore, ID's behaviour is particularly impressive in extremely long signals with a large number of frequently occurring change-points; see Tables \ref{distribution_M2} and 
\ref{models_linearW3}. 
Compared to other well-behaved methods, such as NOT for piecewise-constancy and CPOP for continuous piecewise-linear signals, our methodology has by far the lowest computational cost.
To conclude, the simulation study provides evidence that Isolate-Detect is an accurate, reliable, and quick method for generalized change-point detection.

The results of Table \ref{models_Piotr_ht} are very good for $d=5$ and not too different from those under Gaussian noise. For $d=3$, there is a slight overestimation of the number of change-points. When the tails of the distribution of the noise are significantly heavier than those of the normal distribution, one can obtain better results by increasing the threshold constant. For example, the results in Table \ref{models_Piotr_ht} for $d=3$ were improved when the threshold constant was slightly increased. We highlight that more thorough simulations can be done using our \textsf{R} packages {\textbf{IDetect}} and {\textbf{breakfast}} and code available from  \url{https://github.com/Anastasiou-Andreas/IDetect/blob/master/R/Simulations_used.R}.
\vspace{-0.2in}
\section{Real data examples}
\label{sec:real_data}
\vspace{-0.05in}
\subsection{UK House Price Index}
\label{subsec:hpi}
We investigate the performance of ID on monthly percentage changes in the UK House price index from January 1995 to December 2020 in two London Boroughs: Tower Hamlets and Hackney. The data are available from \url{http://landregistry.data.gov.uk/app/ukhpi} and they were accessed in March 2021. Figure \ref{HPIdata} shows the fits of ID, ID.SDLL, NOT, and TGUH. In both data sets, ID behaves similarly to NOT whereas ID.SDLL's performance is closer to that of TGUH where we detect more change-points. This difference between the examined methods is, in our opinion, due to the fact that ID in this example and NOT detect change-points based on the Schwarz Information Criterion, so fewer estimated change-points can be expected. The detection of two change-points near March 2008 and September 2009 for both boroughs may be related to the financial crisis during that time, which led to a decrease in house prices. As explained in Section \ref{subsec:sSIC}, our methodology returns the solution path defined in \eqref{orderedcpts}, which can be used to obtain different fits; see Section \ref{sec:Supp_more_real_data} in the supplement for more details and for a real-data example where this is useful.
\begin{figure}[h]
\centering
\includegraphics[width=1\textwidth,height=0.3\textheight]{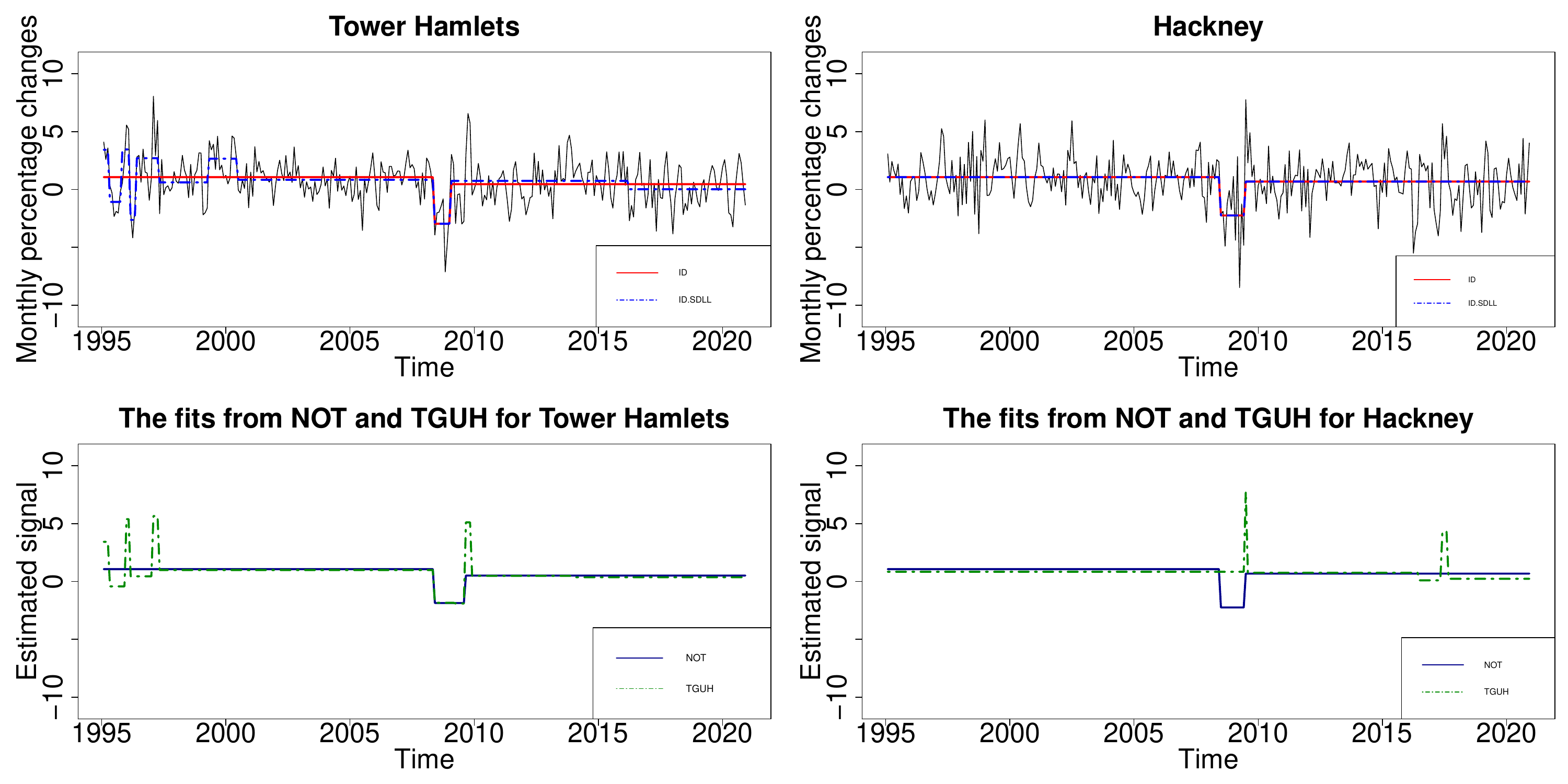}
\caption{{\textbf{Top row:}} The time series and the fitted piecewise-constant mean signals obtained by ID and ID.SDLL for both Tower Hamlets and Hackney. {\textbf{Bottom row:}} NOT (solid) and TGUH (dashed) estimates for Tower Hamlets and Hackney.} 
\label{HPIdata}
\end{figure}
Residual diagnostics have indicated that the behaviour of the raw residuals, $X_t - \hat{f}_t$, in relation to normality and independence is good for all methods. 

\subsection{The COVID-19 outbreak in the UK}
\label{subsec:covid}
The performance of ID is investigated on data from the recent COVID-19 pandemic; we employ a continuous piecewise-linear model on the daily number of lab-confirmed cases in England, as well as on the daily additional COVID-19 associated UK deaths. The data concern the period from the beginning of March 2020 until the end of February 2021 and they are available from \url{https://coronavirus.data.gov.uk}. The data were accessed on the $8^{th}$ of March 2021. Before applying the various methods to the data, we bring the distribution closer to Gaussian with constant variance. To achieve this we perform the Anscombe transform, $a:\mathbb{N} \rightarrow \mathbb{R}$, with $a(x) = 2\sqrt{x+3/8}$ as described in \cite{Anscombe}. We denote the transformed number of COVID-19 cases by $\tilde{X}_t$ and the transformed number of COVID-19 associated deaths by $\tilde{D}_t$. Figure \ref{figure_covid_cases_deaths} presents the results of ID, ID.SDLL, CPOP, and NOT for the transformed data.
\begin{figure}[h]
\centering\includegraphics[width=1\textwidth,height=0.3\textheight]{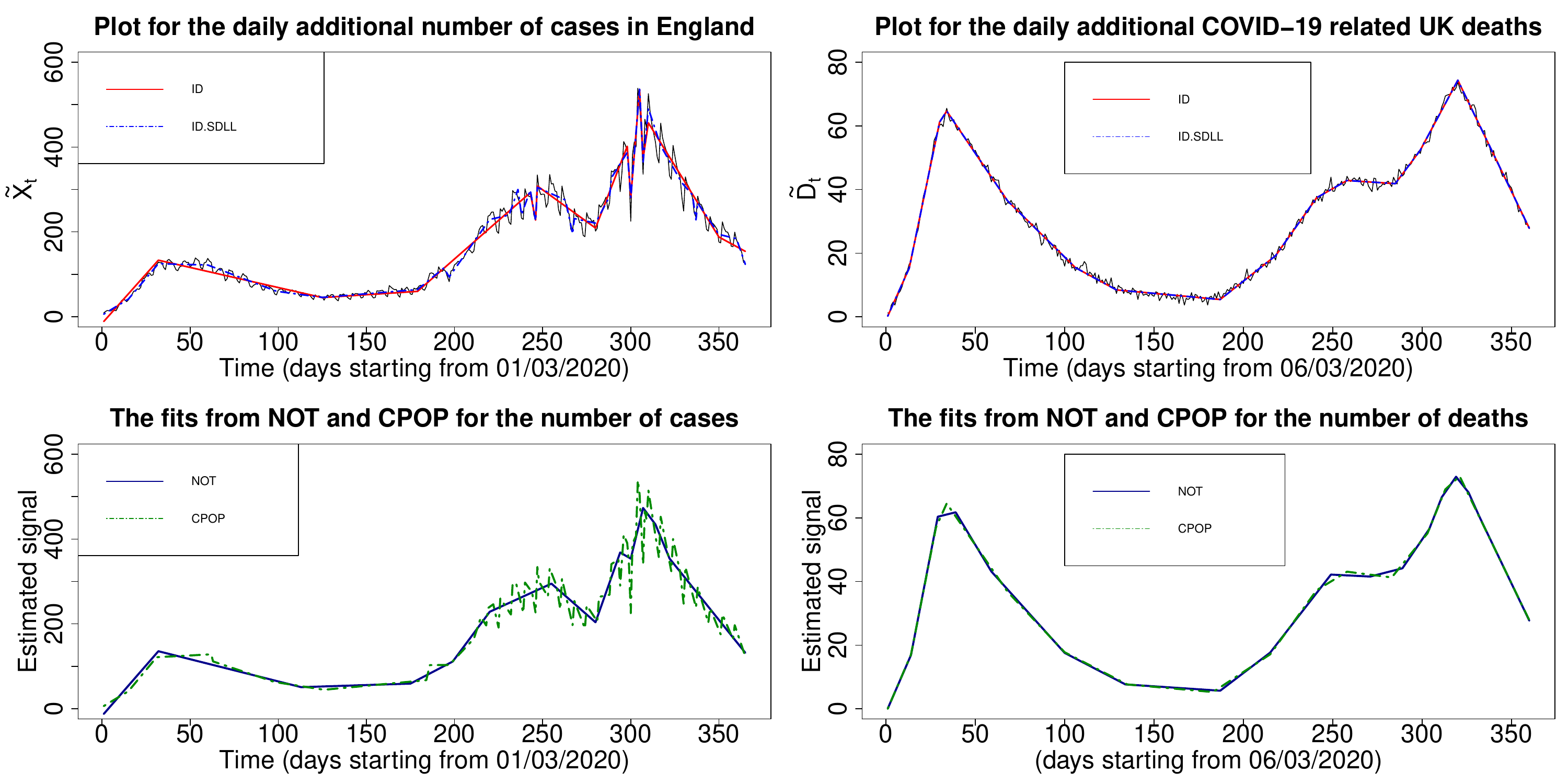}
\caption{{\textbf{Top row:}} The transformed data sequence and the fitted continuous and piecewise-linear mean signals obtained by ID and ID.SDLL for both the daily number of cases and the daily number of deaths. {\textbf{Bottom row:}} NOT (solid) and CPOP (dashed) estimates for the daily number of cases and the daily number of deaths.} 
\label{figure_covid_cases_deaths}
\end{figure}
We observe that ID, ID.SDLL, and NOT have a similar behaviour, while CPOP gives a higher estimated number of change-points. In an attempt to date the detected change-points by ID, we provide a possible explanation of their location with respect to the outbreak of the pandemic in the UK; this discussion is given in Section \ref{sec:supp_discussion_covid} of the supplementary material.

For another example related to the continuous, piecewise-linear case, see Section \ref{sec:Supp_more_real_data} of the supplement where we explore the behaviour of Isolate-Detect and two competitors, CPOP and NOT, on the daily closing stock prices of Samsung Electronics Co. from July 2012 until June 2020.
\section{Concluding reflections on ID}
\label{sec:Discussion}
In this paper, we have proposed Isolate-Detect which is a new, generic technique for multiple generalized change-point detection in noisy data sequences. The method is based on a change-point isolation approach which seems to provide an advantage in detection power, especially in complex structures where most state-of-the-art competitors seem to suffer (see the simulations in Section \ref{sec:simulation}) such as limited spacings between change-points. In addition, the aforementioned isolation aspect allows the extension of our method to the detection of knots in higher-order polynomial signals. As already mentioned in Section \ref{sec:intro}, NOT, WBS, and WBS2 also work on sub-intervals of the data, but the way the isolation is carried out in ID, where one of the end-points of the subintervals is kept fixed, provides predictable execution times for the analysis of a given data sequence, which are faster than the aforementioned competitors; see Sections \ref{subsec:Compcost} and \ref{sec:simulation}. Another advantage of our method over NOT, WBS and WBS2 is that, due to its pseudo-sequential interval expansion character, it can easily be applied for online change-point detection.

In Section \ref{subsec:hybrid}, a variant of ID was introduced that combines the threshold- and SIC-based versions of our proposed method with the aim to enhance its accuracy (both in terms of the estimated number and the estimated change-point locations) for signals of different structures with respect to the true number of change-points and the distance between them. In addition, due to the way that the relevant hybrid approach has been developed in Section \ref{subsec:hybrid}, we manage to offer, for ease of execution, minimal parameter choice. Apart from thresholding and SIC, we have also combined ID with the SDLL model selection criterion.

In the practical applications of Sections \ref{sec:simulation} and \ref{sec:real_data}, compared to the state-of-the-art competitors, ID lies in the top 10\% (in terms of the accurate estimation of the number and the location of the change-points) of the best methods. Furthermore, it exhibits a notable advantage over other techniques in long signals with many change-points that occur frequently. In addition, ID's pseudo-sequential character assists in attaining a low computational time; our method can accurately analyse signals of tens of thousands with thousands of change-points in less than a second; see for example Table \ref{distribution_M2}. In cases where the normality assumption for the error terms is violated, Section \ref{subsec:ht_variant} provides a practical solution where pre-processing allows us to use ID without altering the proposed parameter values. The results of simulations from a Student-$t$ distribution with two options for the degrees of freedom are in Table \ref{models_Piotr_ht}.

Since no method has a uniformly best behaviour, it is natural to also highlight the weaknesses of our method in terms of its practical behaviour. To start with, ID can be slow in long and constant signals in which change-points do not occur. This is because of the expanding intervals attribute, which in the case of no change-points will push the method to keep testing for change-points in growing, overlapping intervals. This is inevitably going to lead to high computational costs. We tried to eliminate this weakness by introducing a window-based variant, as explained in Section \ref{subsec:var_proposed}. Another drawback of the method is that, due to its left- and right-expanding feature, the change-points need to be detectable based on relatively unbalanced intervals. This could lead to accuracy issues in signals where the change-points are in the middle of the data sequence and offset each other. In practice, we have not encountered this type of behaviour in ID; in particular it accurately detects the change-points for the model (M4) in Table \ref{models_Piotr}, which is an example of the aforementioned structure with two nearby change-points in the middle of the data sequence.

\clearpage
\thispagestyle{empty}
\bibliographystyle{natbib}
\markboth{Andreas Anastasiou \and Piotr Fryzlewicz}{}
\begin{center}
{\Large{\textbf{Supplement to ``Detecting multiple generalized
\\
\vspace{0.05in}
change-points by isolating single ones''}
\\
\vspace{0.2in}}
\large{B{\small{Y}} A. ANASTASIOU$^{1, \dagger}$ and P. FRYZLEWICZ$^{2, \ddagger}$}}\\
\vspace{0.1in}
\textit{$^1$Department of Mathematics and Statistics, University of Cyprus,\\
P.O. Box: 20537, 1678, Nicosia, Cyprus}\\
\textit{$^2$Department of Statistics, The London School of Economics and Political Science,\\
Columbia House, Houghton Street, London, WC2A 2AE}\\
\hspace{-1.2in}\textit{E-mail: }$^{\dagger}$anastasiou.andreas@ucy.ac.cy; $^{\ddagger}$p.fryzlewicz@lse.ac.uk
\date{}
\end{center}
\setcounter{section}{0}
\setcounter{table}{0}
\setcounter{page}{1}
\section{The different parts in the removal process of the solution path algorithm}
\label{supp:solution_path_algorithm_parts}
The four different parts of the removal process applied in order to create the solution path explained in Section \ref{subsec:sSIC} of the paper are very similar and are based on the idea of removing change-points according to their contrast function values as well as their distance to neighbouring estimates. We mention that if the algorithm proceeds from Part 1 to Part 2 as below, then it is guaranteed that it will also proceed up to Part 4. All events below occur with probability tending to one with $T$.
\vspace{0.05in}
\\
{\textbf{Part 1:}} With $C^*$ being a positive constant, the aim is to prune the estimates in $\tilde{S}$, such that, for each true change-point, there are at most four and at least one estimated change-point within a distance of $C^*(\log T)^{\alpha}$. To achieve this, $\forall j \in \left\lbrace 1,2,\ldots,J\right\rbrace$ and with $\tilde{r}_0 = 1$, $\tilde{r}_{J+1} = T$,  we collect triplets $(\tilde{r}_{j-1},\tilde{r}_j,\tilde{r}_{j+1})$ and we calculate $CS(\tilde{r}_j) := C_{\tilde{r}_{j-1},\tilde{r}_{j+1}}^{\tilde{r}_j}(\boldsymbol{X})$, with $C_{s,e}^b(\boldsymbol{X})$ being the relevant contrast function. For $m = {\rm argmin}_{j}\left\lbrace CS(\tilde{r}_j) \right\rbrace$, firstly we check whether $CS(\tilde{r}_m) \leq \tilde{\tilde{C}}\sqrt{\log T}$, for $\tilde{\tilde{C}} > 0$; in the proofs of Theorems \ref{consistency_theorem2} and \ref{consistency_theorem_trend2}, $\tilde{\tilde{C}} = 2\sqrt{2}$, but smaller values could be sufficient. If $CS(\tilde{r}_m) \leq \tilde{\tilde{C}}\sqrt{\log T}$ and also $\tilde{r}_{j+1} - \tilde{r}_{j-1} \leq 2C^*(\log T)^{\alpha}$, we remove $\tilde{r}_m$ from $\tilde{S}$, reduce $J$ by 1, relabel the remaining estimates (in increasing order) in $\tilde{S}$, and repeat this estimate removal process. We proceed to Part 2 when $CS(\tilde{r}_m) > \tilde{\tilde{C}}\sqrt{\log T}$. If this is not satisfied at any point of this part, then we conclude that there are no change-points in the data sequence and we stop.
\vspace{0.05in}
\\
{\textbf{Part 2:}} The aim is to continue the pruning process of Part 1, in a way that at the end of Part 2 there is at least one estimate within a distance of $C^*(\log T)^{\alpha}$ from each true change-point, but also there are at most two estimates between any pair of consecutive true change-points. For the relabelled estimates in $\tilde{S}$ after the completion of Part 1, if $\tilde{r}_{j} - \tilde{r}_{j-1} \leq C^*(\log T)^{\alpha}$, then we remove $\tilde{r}_j$, relabel the remaining estimates, and keep removing the estimates until there is no pair $(\tilde{r}_{j-1},\tilde{r}_{j})$, such that $\tilde{r}_j - \tilde{r}_{j-1} \leq C^*(\log T)^{\alpha}$. We then calculate $CS(\tilde{r}_j)$ as in Part 1 and for $m = {\rm argmin}_{j}\left\lbrace CS(\tilde{r}_j) \right\rbrace$, if $CS(\tilde{r}_m) \leq \tilde{\tilde{C}}\sqrt{\log T}$, then we remove $\tilde{r}_m$ and relabel the remaining elements of $\tilde{S}$. This removal process is repeated and we proceed to Part 3 only when $CS(\tilde{r}_m) > \tilde{\tilde{C}}\sqrt{\log T}$.
\vspace{0.05in}
\\
{\textbf{Part 3:}} We need to ensure that once $\tilde{S}$ contains $N$ estimates, then for $j=1,2,\ldots, N$, each $\tilde{r}_j$ is within a distance of $C^*\left(\log T\right)^{\alpha}$ from $r_j$. To achieve this, for the remaining estimated change-points after Part 2, we use triplets $(\tilde{s}_{j},\tilde{r}_j,\tilde{e}_{j})$, with $\tilde{s}_j = \left\lfloor(\tilde{r}_{j-1} + \tilde{r}_j)/2\right\rfloor + 1$ and $\tilde{e}_j = \left\lceil(\tilde{r}_{j} + \tilde{r}_{j+1})/2\right\rceil$. For $m = {\rm argmin}_{j} C_{\tilde{s}_{j},\tilde{e}_{j}}^{\tilde{r}_j}(\boldsymbol{X})$, if $C_{\tilde{s}_m,\tilde{e}_m}^{\tilde{r}_m}(\boldsymbol{X}) \leq \tilde{\tilde{C}}\sqrt{\log T}$, then we remove $\tilde{r}_m$ and relabel the remaining estimates in $\tilde{S}$ in increasing order. We repeat this removal procedure until $C_{\tilde{s}_{m},\tilde{e}_{m}}^{\tilde{r}_m}(\boldsymbol{X}) > \tilde{\tilde{C}}\sqrt{\log T}$, which is when we proceed to Part 4.
\vspace{0.05in}
\\
{\textbf{Part 4:}} For the estimated change-points that are in $\tilde{S}$ after Part 3 is completed, we use again the triplets $(\tilde{r}_{j-1},\tilde{r}_j,\tilde{r}_{j+1})$ in order to find $m = {\rm argmin}_{j}\left\lbrace CS(\tilde{r}_j) \right\rbrace$ and then remove $\tilde{r}_m$ from $\tilde{S}$. This estimates removal approach is repeated until $\tilde{S} = \emptyset$.
\section{Models used in the simulation study of the paper}
\label{supp:signals}
The characteristics of the test signals $f_t$ as well as the standard deviations $\sigma$ of the noise $\epsilon_t$, which were used in the simulation study are given in the list below.
\begin{itemize}
\item[(M1)] \textit{blocks}: length 2048 with change-points at 205, 267, 308, 472, 512, 820, 902, 1332, 1557, 1598, 1659 with values between change-points 0, 14.64, $-$3.66, 7.32, $-$7.32, 10.98, $-$4.39, 3.29, 19.03, 7.68, 15.37, 0. The standard deviation is $\sigma = 10$.



\item[(M2)] \textit{teeth}: length 140 with change-points at 11, 21, 31, 41, 51, 61, 71, 81, 91, 101, 111, 121, 131 with values between change-points 0, 1, 0, 1, 0, 1, 0, 1, 0, 1, 0, 1, 0, 1. The standard deviation of the noise is $\sigma = 0.4$.

\item[(M3)] \textit{stairs}: length 150 with change-points at 11, 21, 31, 41, 51, 61, 71, 81, 91, 101, 111, 121, 131, 141 with values between change-points 1, 2, 3, 4, 5, 6, 7, 8, 9, 10, 11, 12, 13, 14, 15. The standard deviation of the noise is $\sigma = 0.3$.

\item[(M4)] \textit{middle-points}: length 2000 with change-points at 1000 and 1020 with values between change-points 0, 1.5, 0. The standard deviation of the noise is $\sigma = 1$.

\item[(M5)] {\textit{long teeth}}: length 20000 with 1999 change-points at $10,20,\ldots,19990$ with values between change-points $0,3,0,3,\ldots,0,3$. The standard deviation is $\sigma = 0.8$.




\item[(NC)] {\textit{constant signal}}: length 3000 with no change-points. The standard deviation is $\sigma = 1$.

\item[(W1)] {\textit{wave 1}}: piecewise-linear signal without jumps in the intercept, $T = 1408$, with 7 change-points at $256, 512, 768, 1024, 1152, 1280, 1344$ with the corresponding changes in slopes $-1/64,2/64,-3/64,4/64,-5/64,6/64,-7/64$, starting intercept $f_1 = 1$ and slope $f_2-f_1$ = $1/256$. The standard deviation of the noise is $\sigma=1$.

\item[(W2)] {\textit{wave 2}}: piecewise-linear signal without jumps in the intercept, $T = 1500$, with 99 change-points at $15,30,\ldots,1485$. The corresponding changes in the slope are $-1,1, -1, \ldots, -1$, while the starting intercept is $f_1 = -1/2$ and the starting slope is $f_2-f_1$ = $1/40$. The standard deviation is $\sigma=1$.

\item[(W3)] {\textit{smoother signal 1}}: piecewise-linear signal without jumps in the intercept, $T = 200$, with 9 change-points at $20,40,\ldots,180$ with the corresponding changes in slopes $1/6,1/2, -3/4, -1/3, -2/3, 1, 1/4, 3/4, -5/4$. The starting intercept is $f_1 = 1$ and slope $f_2-f_1$ = $1/32$. The standard deviation of the noise is $\sigma=0.3$.

\item[(W4)] {\textit{smoother signal 2}}: piecewise-linear signal without jumps in the intercept, $T = 1000$, with 19 change-points at $50,100,\ldots,950$ with the corresponding changes in slopes $-1/16, -5/16, -5/8, 1, 5/16, 15/32, -5/8, -7/32, -3/4, 13/16, \linebreak 5/16, 19/32, -1, -5/8, 23/32, 1/2, 15/16,-25/16, -5/4$, starting intercept $f_1 = 1$ and slope $1/32$. The standard deviation of the noise is $\sigma=0.6$.
\end{itemize}
\section{Improvement of ID in the case of big data}
\label{sec:Supp_wind_impr}
In this section, we show through simulations that applying ID on a fixed window grid improves its speed in large data sets, without affecting its accuracy. We compare the classic ID method as explained in Section \ref{subsec:methodology} of the main paper with the new window-grid-based version (WID) of Section \ref{subsec:var_proposed} in the case of three data sequences of length $10^5$, each with standard Gaussian noise. We work under the scenario of piecewise-constant mean. The three signals are
\begin{itemize}
\item[(D1)] No change-points;
\item[(D2)] three change-points at 25000, 55000, 85000 and the values between change-points are 0,3,-3,2;
\item[(D3)] seven change-points at 16000, 22000, 28000, 46000, 62000, 74000, 86000 and the values between change-points are 0,4,-4,4,-4,4,-4,4.
\end{itemize}
We took the expansion parameter $\lambda_T$ to be equal to 10 and the results are shown in Table \ref{table_comparison_windows_basic}. As a measure of the accuracy of the detected locations, we provide Monte-Carlo estimates of the mean squared error, ${\rm MSE} = T^{-1}\sum_{t=1}^{T}\E\left(\hat{f}_t - f_t\right)^2$. The scaled Hausdorff distance, $$d_H = n_s^{-1}\max\left\lbrace \max_j\min_k\left|r_j-\hat{r}_k\right|,\max_k\min_j\left|r_j-\hat{r}_k\right|\right\rbrace,$$ where $n_s$ is the length of the largest segment, is also given for (D2) and (D3); for (D1), $d_H$ is not informative. In terms of accuracy, both methods exhibit excellent behaviour. However, in terms of speed, the advantage of the windows-based approach is obvious. Note the decrease in the computational time of ID when the number of change-points gets larger. This is expected because the worst case in terms of computational complexity is when there are no change-points because ID will then be forced to calculate the contrast function on quite large intervals, even on $[1,T]$, which is computationally more expensive.
\begin{table}[H]
\centering
\caption{A comparison on the performance of WID and ID over 10 simulated time series of three different models of length $10^5$ each. The distribution of $\hat{N} - N$, as well as the average MSE, Hausdorff distance and computational time for each method are provided}{
\begin{tabular}{|l|l|l|l|l|l|}
\cline{1-6}
 & & & & & \\
 & & & & & \\
Method & Model & $\hat{N} - N = 0$ & MSE & $d_H$ & Time (s)\\
\hline
WID &  & 10 & $1.65 \times 10^{-5}$ & - & 2.39 \\
ID & (D1) & 10 & $1.65 \times 10^{-5}$ & - & 79.40\\
\hline
WID &  & 10 & $8.8 \times 10^{-5}$ & $3.8 \times 10^{-5}$ & 2.28\\
ID & (D2) & 10 & $11 \times 10^{-5}$ & $1.6 \times 10^{-5}$ & 13.06\\
\hline
WID &  & 10 & $8.9 \times 10^{-5}$ & 0 & 2.17\\
ID & (D3) & 10 & $8.9 \times 10^{-5}$  & 0 & 6.34\\
\hline
\end{tabular}}
\label{table_comparison_windows_basic}
\end{table}





\section{A justification of the estimations in the example of Section \ref{subsec:covid}}
\label{sec:supp_discussion_covid}
Here, we provide a possible explanation of the three most important (based on the solution path) detected change-points by ID in the real data related to the COVID-19 outbreak in the UK. We focus on the signal obtained with respect to the daily number of deaths (ID and ID.SDLL give exactly the same outcome), but similar conclusions can be extracted for the daily number of cases. The three most important changes in the behaviour of the data sequence have been detected on the $8^{th}$ of April, the $12^{th}$ of July and the $19^{th}$ of January. Some justification of these results is as follows:
\begin{itemize}
\item The change-point on the $8^{th}$ of April, shows that the upward trend vanishes and a negative slope takes its place leading to a vast decrease in the number of reported deaths. There is an obvious connection of this change-point with the British Government’s important decision in late March to impose a general lockdown. On 9 April, the Secretary of State for Foreign Affairs stated that the UK was {\textit{starting to see the impact of the restrictions}}; our second detection is in full agreement with the aforementioned date and statement.
\item With respect to the change-point on the $12^{th}$ of July, it indicates a stabilisation, at a low level, for the number of deaths. This is, of course, expected because the downward trend could not continue indefinitely. 
\item The last change-point on the $19^{th}$ of January indicates that an upward trend on the daily number of deaths, which was apparent for a period of almost 3.5 months, vanishes and its place takes a significantly negative slope. There is an obvious connection of this change-point with the vaccination programme in the country. More specifically, on 8 January 2021, MRNA-1273 (commonly known as the Moderna vaccine) was the third COVID-19 vaccine approved for use in the UK.
\end{itemize}
\section{Additional simulation results}
\label{sec:Supp_more_simul}
Here we present the results of simulations for various signals other than those presented in Section \ref{sec:simulation} of the main article. Tables \ref{distribution_LT2}-\ref{sim_supp5} summarize the results for the following models.
\begin{itemize}
\item[(LS)] {\textit{long stairs}}: length 10000 with 499 change-points at $20,40,\ldots,9980$ with values between change-points $0,2,4,6,\ldots,996,998$. The standard deviation is $\sigma = 1$.

\item[(LT2)] {\textit{long teeth 2}}: length 10000 with 249 change-points at $40,80,\ldots,9960$ with values between change-points $0,1.5,0,1.5,\ldots,0,1.5$. The standard deviation is $\sigma = 1$.

\item[(ELT)] {\textit{extremely long teeth}}: length 100000 with 19999 change-points at $5,10,\ldots,99995$ with values between change-points $0,2,0,2,\ldots,0,2$. The standard deviation is $\sigma = 0.3$.

\item[(NC2)] {\textit{constant signal 2}}: length 300 with no change-points. The standard deviation is $\sigma = 1$.

\item[(SW1)] {\textit{wave 5}}: piecewise-linear signal without jumps in the intercept, $T = 2400$, with 119 change-points at $20,40,\ldots,2380$ with the corresponding changes in slopes $2.5,-2.5, 2.5, \ldots, 2.5$, starting intercept $f_1 = 1$, slope $f_2-f_1$ = $1.25$ and $\sigma=3$.

\item[(SW2)] {\textit{wave 6}}: piecewise-linear signal without jumps in the intercept, $T = 1500$, with 29 change-points at $50,100,\ldots,1450$ with the corresponding changes in slopes $-1/7,1/7, -1/7, \ldots, -1/7$, starting intercept $f_1 = -1/2$, slope $f_2-f_1$ = $1/24$ and $\sigma=1$.

\item[(SW3)] {\textit{wave 7}}: piecewise-linear signal without jumps in the intercept, $T = 840$, with 119 change-points at $7,14,\ldots,833$ with the corresponding changes in slopes $-1, 1, -1, \ldots, -1$, starting intercept $f_1 = -1/2$ and slope $f_2-f_1$ = $1/32$. The standard deviation is $\sigma=0.3$.
\end{itemize}
The signals (LS), (LT2), (ELT), and (NC2) are treated under piecewise-constancy, while (SW1), (SW2) and (SW3) under the continuous and piecewise-linear case. FDR, WBSIC and S3IB are excluded from the comparative study for the extremely long signal (ELT). For FDR, we had to interrupt the execution after 10 hours, while for WBSIC and S3IB, in order to have a fair comparison of the methods with the rest, we had to increase the default value of the maximum number of change-points allowed to be detected to be greater than 20000. For a single iteration we had to stop the execution for WBSIC and S3IB after 30 minutes.
\begin{table}[H]
\centering
\caption{Distribution of $\hat{N} - N$ over 100 simulated data sequences from the piecewise-constant signal (LS). The average MSE, $d_H$ and computational time are also given}
{\small{
\begin{tabular}{|l|l|l|l|l|l|l|l|l|}
\cline{1-9}
 & \multicolumn{5}{|c|}{•} & & & \\
 & \multicolumn{5}{|c|}{$\hat{N} - N$} & & & \\
Method  & $\leq -300$ & $(-300,-100]$ & $(-100,-10)$ & $[-10,10]$&  $>10$ & MSE & $d_H$ & Time (s)\\
\hline
PELT  & 0 & 100 & 0 & 0 & 0 & 0.67 & 1.02 & 0.024\\
NP.PELT & 100 & 0 & 0 & 0 & 0 & 9655.28 & 113.76 & 11.302\\
S3IB  & 0 & 13 & 87 & 0 & 0 & 0.44 & 1.01 & 111.143\\
CumSeg  & 100 & 0 & 0 & 0 & 0 & 87.11 & 15.08 & 0.323\\
\textbf{CPM.$l.$500}  & 0 & 0 & 0 & \textbf{100} & 0 & 0.19 & 0.75 & 0.002\\
CPM.$l.$10000  & 0 & 0 & 74 & 25 & 0 & 0.22 & 1 & 0.002\\
WBSC1  & 0 & 0 & 80 & 20 & 0 & 0.24 & 1 & 1.293\\
WBSIC & 0 & 0 & 22 & 78 & 0& 0.22 & 0.99 & 1.293\\
WBS2 & 0 & 0 & 29 & 58 & 13 & 0.22 & 0.99 & 1.293\\
NOT & 100 & 0 & 0 & 0 & 0 & 9.31 & 5.27 & 4.330\\
\textbf{FDR} & 0 & 0 & 2 & \textbf{98} & 0 & 0.19 & 0.83 & -\\
TGUH & 0 & 0 & 100 & 0 & 0 & 0.29 & 1 & 0.484\\
\textbf{ID} & 0 & 0 & 14 & \textbf{86} & 0 & 0.21 & 1.00 & 0.460\\
\hline
\end{tabular}}}
\label{distribution_M4}
\end{table}
\begin{table}[H]
\centering
\caption{Distribution of $\hat{N} - N$ over 100 simulated data sequences from the piecewise-constant signal (LT2). The average MSE, $d_H$ and computational time are also given}
{\small{
\begin{tabular}{|l|l|l|l|l|l|l|l|l|}
\cline{1-9}
 & \multicolumn{5}{|c|}{•} &  & & \\
 & \multicolumn{5}{|c|}{$\hat{N} - N$} & & & \\
Method  & $\leq -150$ & $(-150,-50]$ & $(-50,-10)$ & $[-10,10]$ & $> 10$ & MSE & $d_H$ & Time (s)\\
\hline
PELT  & 100 & 0 & 0 & 0 & 0 & 0.55 & 130.64 & 0.016\\
NP.PELT  & 0 & 12 & 88 & 0 & 0 & 0.21 & 4.95 & 0.496\\
S3IB  & 0 & 1 & 54 & 45 & 0 & 0.13 & 2.58 & 33.410\\
CumSeg  & 100 & 0 & 0 & 0 & 0 & 0.56 & 249 & 0.428\\
CPM.$l.$500  & 0 & 0 & 0 & 9 & 91 & 0.12 & 0.59 & 0.008\\
\textbf{CPM.$l.$10000}  & 0 & 0 & 0 & \textbf{100} & 0 & 0.12 & 1.21 & 0.008\\
WBSC1  & 0 & 84 & 16 & 0 & 0 & 0.27 & 4.22 & 0.631\\
WBSIC & 7 & 0 & 0 & 88 & 5 & 0.16 & 15.92 & 1.051\\
NOT  & 100 & 0 & 0 & 0 & 0 & 0.56  & 240.08 & 0.610\\
\textbf{FDR} & 0 & 0 & 0 & \textbf{99} & 1 & 0.11 & 0.72 & -\\
\textbf{TGUH} & 0 & 0 & 2 & \textbf{98} & 0 & 0.14 & 1.43 & 0.580\\
\textbf{ID} & 0 & 0 & 0 & \textbf{100} & 0 & 0.11 & 0.76 & 0.139\\
\hline
\end{tabular}}}
\label{distribution_LT2}
\end{table}
\begin{table}[H]
\centering
\caption{Distribution of $\hat{N} - N$ over 100 simulated time series from the signal (ELT). Also the average MSE and computational times for each method are given}{
\begin{tabular}{|l|l|l|l|l|l|}
\cline{1-6}
 & \multicolumn{3}{|c|}{•} & & \\
 & \multicolumn{3}{|c|}{$\hat{N} - N$} & & \\
Method  & $\leq -17000$ & $(-17000,-10)$ & $[-10,10]$ & MSE  & Time (s)\\
\hline
PELT  & 100 & 0 & 0 & 0.94 & 0.086\\
NP.PELT & 100 & 0 & 0 & 1 & 136.154\\
CumSeg  & 100 & 0 & 0 & 1 & 4.831\\
CPM.$l.$500  & 100 & 0 & 0 & 1 & 59.936\\
WBSC1  & 100 & 0 & 0 & 0.92 & 6.346\\
NOT & 100 & 0 & 0 & 1 & 2.873\\
\textbf{TGUH} & 0 & 0 & \textbf{100} & 0.02 & 4.751\\
\textbf{ID} & 0 & 10 & \textbf{90} & 0.02 & 3.693\\
\hline
\end{tabular}}
\label{sim_supp1}
\end{table}
\begin{table}[H]
\centering
\caption{Distribution of $\hat{N} - N$ over 100 simulated time series from (NC2). Also the average MSE and computational times for each method are given}{
\begin{tabular}{|l|l|l|l|l|l|l|}
\cline{1-7}
 & \multicolumn{4}{|c|}{•} & & \\
 & \multicolumn{4}{|c|}{$\hat{N} - N$} & & \\
Method  &  0 & 1 & 2 & $\geq 3$ & MSE & Time (ms)\\
\hline
\textbf{PELT}  & \textbf{100} & 0 & 0 & 0 & 28 $\times 10^{-4}$  & 12.3\\
NP.PELT & 56 & 13 & 23 & 8 & 269 $\times 10^{-4}$  & 30.4\\
\textbf{S3IB} & \textbf{95} & 3 & 2 & 0 & 57 $\times 10^{-4}$  & 35.7\\
\textbf{CumSeg}  & \textbf{100} & 0 & 0 & 0 & 28 $\times 10^{-4}$ & 19.3\\
CPM.$l.$500 & 54 & 11 & 15 & 20 & 294 $\times 10^{-4}$ & 1.5\\
WBSC1  & 17 & 16 & 19 & 48 & 578 $\times 10^{-4}$ & 62.7\\
\textbf{WBSIC}  & \textbf{95} & 3 & 1 & 1 & 59 $\times 10^{-4}$ & 61.3\\
\textbf{NOT}  & \textbf{99} & 0 & 0 & 1 & 39 $\times 10^{-4}$  & 37.6\\
\textbf{FDR}  & \textbf{90} & 7 & 2 & 1 & 66 $\times 10^{-4}$  & -\\
TGUH & 83 & 0 & 12 & 5 & 147 $\times 10^{-4}$  & 49.6\\
\textbf{ID}  & \textbf{95} & 4 & 1 & 0 & 60 $\times 10^{-4}$ & 1.1\\
\hline
\end{tabular}}
\label{sim_supp3}
\end{table}
\begin{table}[H]
\centering
\caption{Distribution of $\hat{N} - N$ over 100 simulated time series of the signal (SW1). Also, the average MSE, Hausdorff distance and computational time for each method are given}{
\begin{tabular}{|l|l|l|l|l|l|l|l|l|l|}
\cline{1-10}
 & \multicolumn{6}{|c|}{•} & & & \\
 & \multicolumn{6}{|c|}{$\hat{N} - N$} & & & \\
Method & $\leq -110$ & $(-110,-1]$ & 0 & $1$ & $(1,90)$ & $\geq 90$ & MSE & $d_H$ & Time (s)\\
\hline
NOT & 100 & 0 & 0 & 0 & 0 & 0 & 52.352 & 119 & 6.061\\
TF  & 0 & 0 & 0 &0 & 0 & 100 & 107.463 & 0.381 & 2.309\\
\textbf{CPOP}  & 0 & 0 & \textbf{96} & 4 & 0 & 0 & 1.063 & 0.146 & 3.327\\
\textbf{ID} & 0 & 0 & \textbf{90} & 10 & 0 & 0 & 1.781 & 0.254 & 0.041\\
\hline
\end{tabular}}
\label{sim_supp4}
\end{table}
\begin{table}[H]
\centering
\caption{Distribution of $\hat{N} - N$ over 100 simulated time series of the signal (SW2). Also, the average MSE, Hausdorff distance and computational time for each method are given}{
\begin{tabular}{|l|l|l|l|l|l|l|l|l|l|l|l|}
\cline{1-12}
 & \multicolumn{8}{|c|}{•} & & & \\
 & \multicolumn{8}{|c|}{$\hat{N} - N$} & & & \\
Method & $\leq -10$ & $(-10,-1)$ & $-1$ & 0 & 1 & 2 & $(2,10)$ & $\geq 10$ & MSE & $d_H$ & Time (s)\\
\hline
NOT  & 0 & 5 & 1 & 1 & 1 & 3 & 25 & 64 & 0.24 & 0.95 & 0.412\\
TF  & 0 & 0 & 0 & 0 & 0 & 0 & 0 & 100 & 0.76 & 0.33 & 1.417\\
\textbf{CPOP}  & 0 & 0 & 0 & \textbf{97} & 3 & 0 & 0 & 0 & 0.05 & 0.17 & 8.728\\
\textbf{ID} & 0 & 0 & 0 & \textbf{97} & 2 & 1 & 0 & 0 & 0.07 & 0.27 & 0.049\\
\hline
\end{tabular}}
\label{sim_supp5}
\end{table}
\begin{table}[H]
\centering
\caption{Distribution of $\hat{N} - N$ over 100 simulated data sequences of the continuous piecewise-linear signal (SW3). The average MSE, $d_H$ and computational time for each method are also given}
{\small{
\begin{tabular}{|l|l|l|l|l|l|l|l|l|l|l|}
\cline{1-11}
 & \multicolumn{7}{|c|}{•} & & & \\
 & \multicolumn{7}{|c|}{$\hat{N} - N$} & & & \\
Method & $\leq -100$ & $(-100,-1)$ & $-1$ & 0 & $1$ & $(1,10]$ & $> 10$ & MSE & $d_H$ & Time (s)\\
\hline
NOT & 100 & 0 & 0 & 0 & 0 & 0 & 0 & 1.063 & 119 & 0.485\\
TF & 0 & 0 & 0 & 0 & 0 & 0 & 100 & 217868.2 & 0.324 & 0.632\\
\textbf{CPOP} & 0 & 0 & 0 & \textbf{98} & 2 & 0 & 0 & 0.027 & 0.154 & 0.438\\
\textbf{ID} & 0 & 0 & 0 & \textbf{100} & 0 & 0 & 0 & 0.039 & 0.210 & 0.055\\
\hline
\end{tabular}}}
\label{models_linearW4}
\end{table}

In all examples, the ID methodology is within $10\%$ of the best method. Once again, it exhibits remarkable behaviour when it comes to very long signals with a large number of frequently appearing change-points; see Tables \ref{distribution_LT2}, \ref{sim_supp1} and \ref{sim_supp4}.
\section{Investigation of the impact of different values of $\lambda_T$ on accuracy}
\label{sec:investigation_lambda}
In this section, we provide a small-scale simulation study in order to examine the behaviour of ID on different values of the expansion parameter $\lambda_T$. The simulation set up is as in Section \ref{sec:simulation} of the main paper and the signals used are those explained in Section \ref{supp:signals} of the supplement. For the expansion parameter, we take $\lambda_T \in \left\lbrace 5,20,80\right\rbrace$. The results are in Tables \ref{supplement_table_lambda1} - \ref{supplement_table_lambda9} below.
\begin{table}[h!]
\centering
\caption{Distribution of $\hat{N} - N$ over 100 simulated data sequences from (NC). The average MSE and computational times are also given}
{\small{
\begin{tabular}{|l|l|l|l|l|l|l|}
\cline{1-7}
 & \multicolumn{4}{|c|}{•} & & \\
 & \multicolumn{4}{|c|}{$\hat{N} - N$} & & \\
Method  &  0 & 1 & 2 & $\geq 3$ & MSE & Time (s)\\
\hline
ID$_{\lambda_T = 5}$  & 100 & 0 & 0 & 0 & 31 $\times 10^{-5}$ & 0.116\\
ID$_{\lambda_T = 20}$ & 100 & 0 & 0 & 0 & 31 $\times 10^{-5}$ & 0.033\\
ID$_{\lambda_T = 80}$ & 100 & 0 & 0 & 0 & 31 $\times 10^{-5}$ & 0.012\\
\hline
\end{tabular}}}
\label{supplement_table_lambda1}
\end{table}
\begin{table}[H]
\centering
\caption{Distribution of $\hat{N} - N$ over 100 simulated data sequences of the piecewise-constant signals (M1)-(M3). The average MSE, $d_H$ and computational time are also given}
{\small{
\begin{tabular}{|l|l|l|l|l|l|l|l|l|l|l|l|}
\cline{1-12}
 & & \multicolumn{7}{|c|}{•} & & & \\
 & & \multicolumn{7}{|c|}{$\hat{N} - N$} & & & \\
Method & Model & $\leq -3$ & $-2$ & $-1$ & 0 & 1 & 2 & $\geq 3$ & MSE & $d_H$ & Time (ms)\\
\hline
ID$_{\lambda_T = 5}$ &  & 0 & 3 & 40 & 56 & 1 & 0 & 0 & 2.49 & 0.06 & 26.8\\
ID$_{\lambda_T = 20}$ & (M1) & 0 & 3 & 40 & 57 & 0 & 0 & 0 & 2.44 & 0.06 & 17.4\\
ID$_{\lambda_T = 80}$ &  & 2 & 8 & 47 & 43 & 0 & 0 & 0 & 2.92 & 0.08 & 11.6\\
\hline
ID$_{\lambda_T = 5}$ &  & 7 & 8 & 3 & 72 & 9 & 1 & 0 & 67 $\times 10^{-3}$ & 0.88 & 10.9\\
ID$_{\lambda_T = 20}$ & (M2) & 83 & 7 & 8 & 4 & 0 & 0 & 0 & 185 $\times 10^{-3}$ & 6.80 & 8.1\\
ID$_{\lambda_T = 80}$ &  & 91 & 2 & 0 & 7 & 0 & 0 & 0 & 212 $\times 10^{-3}$ & 8.75 & 5.7\\
\hline
ID$_{\lambda_T = 5}$ &  & 0 & 0 & 0 & 82 & 14 & 4 & 0 & 23 $\times 10^{-3}$ & 0.16 & 9.9\\
ID$_{\lambda_T = 20}$ & (M3) & 20 & 38 & 30 & 10 & 2 & 0 & 0 & 86 $\times 10^{-3}$ & 0.82 & 9.2\\
ID$_{\lambda_T = 80}$  &  & 100 & 0 & 0 & 0 & 0 & 0 & 0 & 784 $\times 10^{-3}$ & 2.73 & 5\\
\hline
\end{tabular}}}
\label{supplement_table_lambda2}
\end{table}
\begin{table}[H]
\centering
\caption{Distribution of $\hat{N} - N$ over 100 simulated data sequences from the piecewise-constant signal (M4). The average MSE, $d_H$ and computational time are also given}
{\small{
\begin{tabular}{|l|l|l|l|l|l|l|l|l|}
\cline{1-9}
 & \multicolumn{5}{|c|}{•} &  & & \\
 & \multicolumn{5}{|c|}{$\hat{N} - N$} & & & \\
Method  & $-2$ & $-1$ & $0$ & $1$ & $\geq 2$ & MSE & $d_H$ & Time (ms)\\
\hline
ID$_{\lambda_T = 5}$ & 12 & 0 & 87 & 1 & 0 & 7 $\times 10^{-3}$ & 0.13 & 31.5\\
ID$_{\lambda_T = 20}$ & 11 & 0 & 89 & 0 & 0 & 6 $\times 10^{-3}$ & 0.11 & 11.7\\
ID$_{\lambda_T = 80}$ & 28 & 0 & 72 & 0 & 0 & 11 $\times 10^{-3}$ & 0.29 & 5.8\\
\hline
\end{tabular}}}
\label{supplement_table_lambda3}
\end{table}
\begin{table}[H]
\centering
\caption{Distribution of $\hat{N} - N$ over 100 simulated data sequences from the piecewise-constant signal (M5). The average MSE, $d_H$ and computational time are also given}
{\small{
\begin{tabular}{|l|l|l|l|l|l|l|l|l|}
\cline{1-9}
 & \multicolumn{5}{|c|}{•} & & & \\
 & \multicolumn{5}{|c|}{$\hat{N} - N$} & & & \\
Method  & $\leq -500$ & $(-500,-50]$ & $(-50,-10)$ & $[-10,10]$ & $> 10$ & MSE & $d_H$ & Time (s)\\
\hline
ID$_{\lambda_T = 5}$ & 0  & 0 & 0 & 100 & 0 & 0.14 & 0.99 & 0.772\\
ID$_{\lambda_T = 20}$ & 100 & 0 & 0 & 0 & 0 & 1.28 & 3.45 & 0.395\\
ID$_{\lambda_T = 80}$ & 100 & 0 & 0 & 0 & 0 & 1.53 & 10.34 & 0.308\\
\hline
\end{tabular}}}
\label{supplement_table_lambda4}
\end{table}
\begin{table}[H]
\centering
\caption{Distribution of $\hat{N} - N$ over 100 simulated data sequences from the continuous piecewise-linear signal (W1). The average MSE, $d_H$ and computational time for each method are also given}
{\small{
\begin{tabular}{|l|l|l|l|l|l|l|l|l|l|l|}
\cline{1-11}
 & \multicolumn{7}{|c|}{•} &  &  & \\
 & \multicolumn{7}{|c|}{$\hat{N} - N$} &  &  & \\
Method & $\leq -3$ & $-2$ & $-1$ & 0 & 1 & 2 & $\geq 3$ & MSE & $d_H$ & Time (s)\\
\hline
ID$_{\lambda_T = 5}$ & 0 & 0 & 0 & 91 & 9 & 0 & 0 & 0.031 & 0.104 & 0.036\\
ID$_{\lambda_T = 20}$ & 0 & 0 & 0 & 92 & 8 & 0 & 0 & 0.027 & 0.099 & 0.020\\
ID$_{\lambda_T = 80}$ & 0 & 0 & 0 & 99 & 1 & 0 & 0 & 0.020 & 0.067 & 0.017\\
\hline
\end{tabular}}}
\label{supplement_table_lambda6}
\end{table}
\begin{table}[H]
\centering
\caption{Distribution of $\hat{N} - N$ over 100 simulated data sequences of the continuous piecewise-linear signal (W2). The average MSE, $d_H$ and computational time for each method are also given}
{\small{
\begin{tabular}{|l|l|l|l|l|l|l|l|l|l|l|}
\cline{1-11}
 & \multicolumn{7}{|c|}{•} & & & \\
 & \multicolumn{7}{|c|}{$\hat{N} - N$} & & & \\
Method & $\leq -90$ & $(-90,-1)$ & $-1$ & 0 & $1$ & $(1,60]$ & $> 60$ & MSE & $d_H$ & Time (s)\\
\hline
ID$_{\lambda_T = 5}$ & 0 & 0 & 0 & 97 & 3 & 0 & 0 & 0.227 & 0.272 & 0.690\\
ID$_{\lambda_T = 20}$ & 0 & 0 & 0 & 100 & 0 & 0 & 0 & 0.364 & 0.328 & 0.722\\
ID$_{\lambda_T = 80}$ & 100 & 0 & 0 & 0 & 0 & 0 & 0 & 4.730 & 98.955 & 0.102\\
\hline
\end{tabular}}}
\label{suppement_table_lambda7}
\end{table}
\begin{table}[H]
\centering
\caption{Distribution of $\hat{N} - N$ over 100 simulated time series of the continuous piecewise-linear signals (W3) and (W4). The average MSE, $d_H$ and computational time are also given}
{\small{
\begin{tabular}{|l|l|l|l|l|l|l|l|l|l|l|l|}
\cline{1-12}
 & & \multicolumn{7}{|c|}{•} &  &  & \\
 & & \multicolumn{7}{|c|}{$\hat{N} - N$} & & & \\
Method & Model & $\leq -3$ & $-2$ & $-1$ & 0 & 1 & 2 & $\geq 3$ & MSE & $d_H$ & Time (s)\\
\hline
ID$_{\lambda_T = 5}$  &  & 0 & 0 & 0 & 94 & 6 & 0 & 0 & 0.020 & 0.122 & 0.014\\
ID$_{\lambda_T = 20}$ & (W3) & 0 & 0 & 0 & 99 & 1 & 0 & 0 & 0.009 & 0.067 & 0.012\\
ID$_{\lambda_T = 80}$ &  & 100 & 0 & 0 & 0 & 0 & 0 & 0 & 3.917 & 1.875 & 0.009\\
\hline
ID$_{\lambda_T = 5}$ &  & 0 & 0 & 0 & 77 & 22 & 1 & 0 & 0.055 & 0.134 & 0.040\\
ID$_{\lambda_T = 20}$ & (W4) & 0 & 0 & 0 & 98 & 2 & 0 & 0 & 0.029 & 0.106 & 0.032\\
ID$_{\lambda_T = 80}$ &  & 0 & 36 & 50 & 14 & 0 & 0 & 0 & 1.603 & 0.950 & 0.030\\
\hline
\end{tabular}}}
\label{supplement_table_lambda9}
\end{table}
In terms of accuracy, we notice that in most cases, the best method is ID$_{\lambda_T = 5}$. In addition, as long as the different values of $\lambda_T$ are all less than the minimum distance, $\delta_T$, between two successive change-points, then the results are extremely similar (if not identical); see for example Tables \ref{supplement_table_lambda1} and \ref{supplement_table_lambda6}. On the other hand, when $\lambda_T > \delta_T$ the accuracy is getting worse; compare for example the results for the three different values of $\lambda_T$ in Models (M2), (M3), (M5), (W2), and (W3). In terms of speed, as expected, the larger the value of $\lambda_T$, the quicker the method in general.
\section{Additional real-data example}
\label{sec:Supp_more_real_data}
In this section we explore the behaviour of our method and two competitors, CPOP and NOT, to the daily closing stock prices of Samsung Electronics Co. from July 2012 until June 2020. The data are available from \url{https://finance.yahoo.com/quote/005930.KS/history?p=005930.KS} and they were accessed in July 2020. We look for changes in a continuous piecewise-linear mean signal. Figure \ref{Samdata} shows the results for the ID, ID.SDLL, NOT and CPOP methods, which detect $165,134,22,239$ change-points, respectively. From both the fit and the residuals given in Figure \ref{Samdata}, it is not easy to say which of the three methods gives the ``best'' number of change-points.
\begin{figure}[H]
\centering
\includegraphics[width=1\textwidth,height=0.35\textheight]{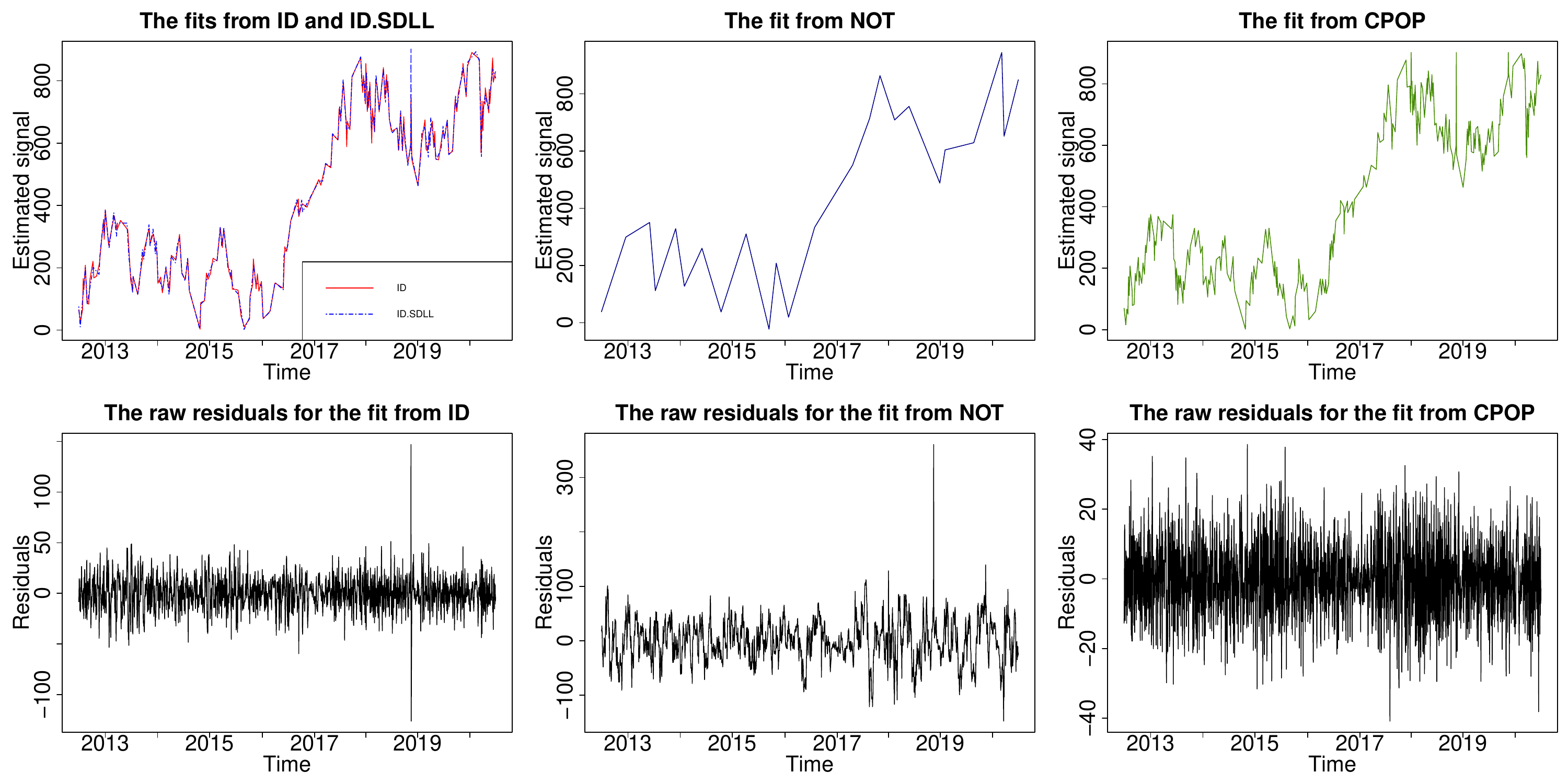}
\caption{{\textbf{Top row:}} From left to right, the fits for ID and ID.SDLL (on the same plot), NOT and CPOP, respectively. {\textbf{Bottom row:}} The raw residuals $\epsilon_t = Y_t - \hat{f}_t$ for each method.} 
\label{Samdata}
\end{figure}

ID can return a range of different fits providing users with the flexibility to choose according to their preference. In Figure \ref{CPOP_comp}, we use the solution path and we obtain the estimated signal and the raw residuals of ID for $\hat{N} = 22$ and $\hat{N} = 239$, which are the estimated change-point numbers through NOT and CPOP, respectively. The fit is similar to those obtained by the aforementioned methods as presented in Figure \ref{Samdata}. However, we note that both those competitors are significantly slower than ID; see Tables \ref{models_linearW1W2} and \ref{models_linearW3} in the main paper and Tables \ref{sim_supp4} - \ref{models_linearW4} in the supplement for a comparison. To conclude, apart from returning the estimated fit, the ID methodology can directly, and without any extra effort, produce a series of estimated signals based on the solution path defined in \eqref{orderedcpts}.
\begin{figure}[H]
\centering\includegraphics[width=1\textwidth,height=0.25\textheight]{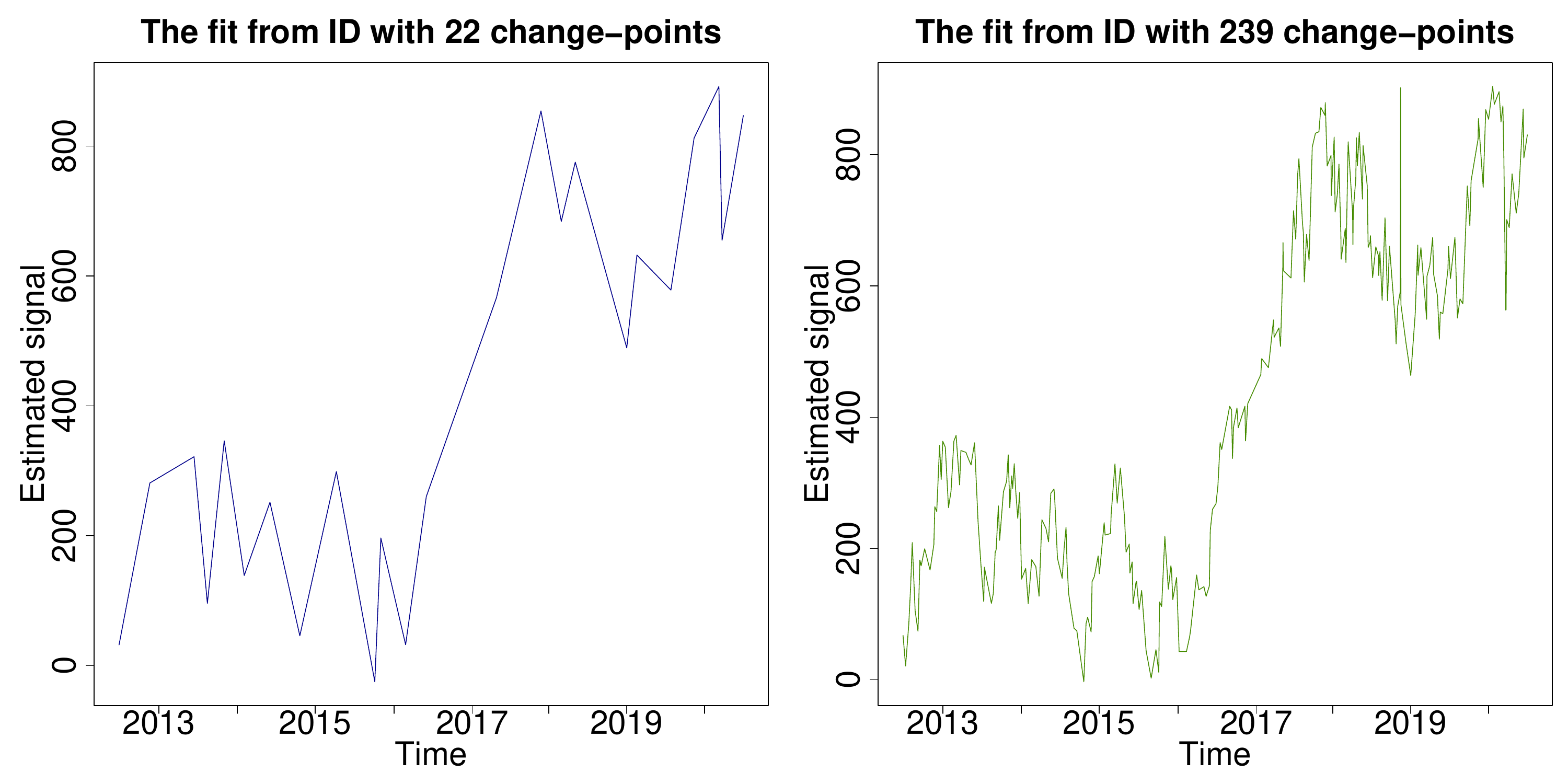}
\caption{The estimated signals obtained by ID with 22 and 239 change-points. The solution path was employed in order to obtain these fits.} 
\label{CPOP_comp}
\end{figure}
\section{Proof of the theorems in the paper and of Corollary \ref{cor:bound}}
\label{sec:supp_proofs}
From now on, the contrast vector $\boldsymbol{\psi_{s,e}^b} = (\psi_{s,e}^b(1), \psi_{s,e}^b(2),\ldots, \psi_{s,e}^b(T))$ is defined through the contrast function
\begin{equation}
\nonumber \psi_{s,e}^b(t) = \begin{cases}
    \sqrt{\frac{e-b}{n(b-s+1)}}, & t = s,s+1,\ldots,b,\\
    -\sqrt{\frac{b-s+1}{n(e-b)}}, & t=b+1,b+2,\ldots,e,\\
    0, & \text{otherwise}.
  \end{cases},
\end{equation}
where $s\leq b < e$. Notice that for any vector $\boldsymbol{v} = (v_1,v_2,\ldots,v_T)$, we have that $\langle \boldsymbol{v}, \boldsymbol{\psi_{s,e}^b}\rangle = \tilde{v}_{s,e}^b.$ We first present Lemma \ref{lemma_from_NOT}, which is partly used for the proof of Theorem \ref{consistency_theorem},
\begin{lemma}
\label{lemma_from_NOT}
Suppose $\boldsymbol{f} = \left(f_1,f_2,\ldots,f_T\right)^{\intercal}$ is a piecewise-constant vector. Pick any interval $[s,e] \subset [1,T]$ such that $[s,e-1]$ contains exactly one change-point $r_j$. Let $\rho = |r_j-b|$, $\Delta_j^f = \left|f_{r_j+1}-f_{r_j}\right|$, $\eta_L= r_j-s+1$ and $\eta_R=e-r_j$. Then,
\begin{equation}
\nonumber \|\psi_{s,e}^b\langle\boldsymbol{f},\boldsymbol{\psi_{s,e}^b}\rangle - \boldsymbol{\psi_{s,e}^{r_j}}\langle\boldsymbol{f},\boldsymbol{\psi_{s,e}^{r_j}}\rangle\|_2^2 = \left(\tilde{f}_{s,e}^{r_j}\right)^2 - \left(\tilde{f}_{s,e}^{b}\right)^2.
\end{equation}
In addition,
\begin{enumerate}
\item for any $r_j \leq b < e$, $\left(\tilde{f}_{s,e}^{r_j}\right)^2 - \left(\tilde{f}_{s,e}^{b}\right)^2 = \left(\rho\eta_L/(\rho+\eta_L)\right)\left(\Delta_j^f\right)^2$;
\item for any $s \leq b < r_j$, $\left(\tilde{f}_{s,e}^{r_j}\right)^2 - \left(\tilde{f}_{s,e}^{b}\right)^2 = \left(\rho\eta_R/(\rho+\eta_R)\right)\left(\Delta_j^f\right)^2$;
\end{enumerate}
\end{lemma}
\begin{proof}
See Lemma 4 from \cite{NOT_paper}.
\end{proof}
{\raggedleft{\textbf{Brief discussion of the steps of the proof of Theorem \ref{consistency_theorem}}}}\\
Before proceeding with the thorough mathematical proof, we give an informal explanation of the main steps. In the main part of the proof, we derive results for the signal $f_t$. However, the consistency is concerned with the estimated number and locations of the change-points in the observed process $\left\lbrace X_t\right\rbrace_{t=1,2,\ldots,T}$. Therefore, in order to be able to deduce consistency related to $X_t$ from our $f_t$-reliant proof, we need first to show that for all $1\leq s \leq b < e \leq T$, the observed quantity $\left|\tilde{X}_{s,e}^b\right|$ is uniformly close to the unobserved $\left|\tilde{f}_{s,e}^b\right|$; this is achieved in Step 1. In Step 2, for $b_1,b_2 \in [s,e)$, we control the distance between the noised $\left|\tilde{X}_{s,e}^{b_1}\right| - \left|\tilde{X}_{s,e}^{b_2}\right|$ and its noiseless equivalent $\left|\tilde{f}_{s,e}^{b_1}\right| - \left|\tilde{f}_{s,e}^{b_2}\right|$ for all possible combinations of $s,e,b_1,b_2$. This allows us to transfer the decision on whether $b_1$ or $b_2$ is more suitable as a change-point, from $\left|\tilde{f}_{s,e}^{b_1}\right| - \left|\tilde{f}_{s,e}^{b_2}\right|$ to the calculable $\left|\tilde{X}_{s,e}^{b_1}\right| - \left|\tilde{X}_{s,e}^{b_2}\right|$. Step 3 is the main part of our proof, where we first show that as the ID algorithm proceeds, each change-point will get isolated in an interval where its detection will occur with high probability. Therefore, it suffices to restrict our proof to a single change-point detection framework, and the convergence rate is proved to hold for each estimated location. Because upon detection ID proceeds from the end-point (or start-point) of the interval where the detection occurred, we also show that with probability one there is no change-point in those bypassed points (between the detection and the new start- or end-point). Furthermore, in Step 3 it is shown that, the new start- and end-points are at places that allow the detection of the next change-point. In Step 4, we conclude the proof by showing that after detecting all change-points, then ID, with high probability, will terminate after scanning all the remaining data. We mention that for our proof, we employ Lemma \ref{lemma_from_NOT} given in the online supplement.
\vspace{0.05in}
\\
{\raggedright{\textbf{Proof of Theorem \ref{consistency_theorem}}. We will prove the more specific result}}
\begin{equation}
\label{mainresult_theorem2}
\Prob\left(\hat{N} = N, \max_{j=1,2,\ldots,N}\left(\left|\hat{r}_j - r_j\right|\left(\Delta_j^f\right)^2\right) \leq C_3\log T\right) \geq 1 - \frac{1}{6\sqrt{\pi}T},
\end{equation}
which implies the result in \eqref{mainresult_theorem}.
\vspace{0.05in}
\\
{\textbf{Step 1}:} Allow us to denote by
\begin{equation}
\label{A_T}
A_T = \left\lbrace \max_{s,b,e: 1\leq s \leq b < e \leq T}\left|\tilde{X}_{s,e}^b - \tilde{f}_{s,e}^b\right|\leq \sqrt{8\log T} \right\rbrace.
\end{equation}
We will show that $\Prob\left(A_T\right) \geq 1-1/(12\sqrt{\pi}T)$. From \eqref{new_model} and \eqref{CUSUM}, simple steps yield $\tilde{X}_{s,e}^b - \tilde{f}_{s,e}^b = \tilde{\epsilon}_{s,e}^b$, where $\tilde{\epsilon}_{s,e}^b \sim \mathcal{N}(0,1)$. Thus, for $Z \sim \mathcal{N}(0,1)$, using the Bonferroni inequality we get that 
\begin{align}
\nonumber & \Prob\left((A_T)^{c}\right) = \Prob\left(\max_{s,b,e: 1\leq s \leq b < e \leq T}\left|\tilde{X}_{s,e}^b - \tilde{f}_{s,e}^b\right| > \sqrt{8\log T}\right)\\
\nonumber & \leq \sum_{1\leq s\leq b <e \leq T}\Prob\left(\left|\tilde{\epsilon}_{s,e}^b\right|>\sqrt{8\log T}\right) \leq \frac{T^3}{6}\Prob(|Z|>\sqrt{8\log T})\\
\nonumber & = \frac{T^3}{3}\Prob\left(Z>\sqrt{8 \log T}\right)\leq \frac{T^3}{3}\frac{\phi(\sqrt{8\log T})}{\sqrt{8\log T}} \leq \frac{1}{12\sqrt{\pi}T},
\end{align}
where $\phi(\cdot)$ is the probability density function of the standard normal distribution.
\vspace{0.1in}
\\
{\textbf{Step 2:}} For intervals $[s,e)$ that contain only one true change-point $r_j$, we denote by
\begin{equation}
\label{B_T}
B_T = \left\lbrace \max_{j=1,2,\ldots,N}\max_{\substack{r_{j-1}<s\leq r_j\\r_j < e \leq r_{j+1}\\s\leq b < e}}\frac{\left|\left\langle\boldsymbol{\psi_{s,e}^b}\langle\boldsymbol{f},\boldsymbol{\psi_{s,e}^b}\rangle - \boldsymbol{\psi_{s,e}^{r_j}}\langle\boldsymbol{f},\boldsymbol{\psi_{s,e}^{r_j}}\rangle,\boldsymbol{\epsilon}\right\rangle\right|}{\|\boldsymbol{\psi_{s,e}^b}\langle\boldsymbol{f},\boldsymbol{\psi_{s,e}^b}\rangle - \boldsymbol{\psi_{s,e}^{r_j}}\langle\boldsymbol{f},\boldsymbol{\psi_{s,e}^{r_j}}\rangle\|_{2}}\leq \sqrt{8 \log T}\right\rbrace.
\end{equation}
Because $\left|\left\langle\boldsymbol{\psi_{s,e}^b}\langle\boldsymbol{f},\boldsymbol{\psi_{s,e}^b}\rangle - \boldsymbol{\psi_{s,e}^{r_j}}\langle\boldsymbol{f},\boldsymbol{\psi_{s,e}^{r_j}}\rangle\right\rangle\right|/\left(\|\boldsymbol{\psi_{s,e}^b}\langle\boldsymbol{f},\boldsymbol{\psi_{s,e}^b}\rangle - \boldsymbol{\psi_{s,e}^{r_j}}\langle\boldsymbol{f},\boldsymbol{\psi_{s,e}^{r_j}}\rangle\|_{2}\right)$ follows the standard normal distribution, then we use a similar approach as in Step 1, to show that $\Prob\left(\left(B_T\right)^{c}\right) \leq \frac{1}{12\sqrt{\pi}T}$. Therefore, Steps 1 and 2 lead to
\begin{equation}
\nonumber \Prob\left(A_T \cap B_T\right) \geq 1 - \frac{1}{6\sqrt{\pi}T}.
\end{equation}
{\textbf{Step 3:}} This is the main part of our proof, where we explain in detail how to get the result in \eqref{mainresult_theorem2}. For ease of understanding, we split this step into two smaller parts. From now on, we assume that $A_T$ and $B_T$ both hold. The constants we use are
\begin{equation}
\label{constantsPCM}
C_1 = \sqrt{C_3} + \sqrt{8}, C_2 = \frac{1}{\sqrt{6}} - \frac{2\sqrt{2}}{\underline{C}}, C_3 = 2(2\sqrt{2}+4)^2,
\end{equation}
where $\underline{C}$ is as in condition (A1).
\vspace{0.1in}
\\
{\textbf{Step 3.1:}} For ease of presentation, we take $\lambda_T \leq \delta_T/3$; see Remark \ref{remark_lambda_T} for comments in regards to the general case of $\lambda_T \leq \delta_T/m$, for an $m>1$. Allow us now $\forall j \in \left\lbrace 1,2, \ldots, N\right\rbrace$, to define the intervals 
\begin{equation}
\label{isolating_intervals}
I_{j}^R = \left[r_j + \frac{\delta_T}{3}, r_j + 2\frac{\delta_T}{3}\right), \qquad I_{j}^L = \left(r_j - 2\frac{\delta_T}{3}, r_j - \frac{\delta_T}{3}\right].
\end{equation}
In order for $I_j^R$ and $I_j^L$ to have at least one point, we actually implicitly require that $\delta_T > 3$, which is the case for sufficiently large $T$; see assumption (A1). Since the length of the intervals in \eqref{isolating_intervals} is equal to $\delta_T/3$ and $\lambda_T \leq \delta_T/3$, then ID ensures that for $K=\left\lceil T/\lambda_T  \right\rceil$ and $k,m \in \left\lbrace 1,2,\ldots,K\right\rbrace$, there exists at least one $c_k^r = k\lambda_T$ and at least one $c_m^l = T-m\lambda_T+1$ that are in $I_j^R$ and $I_j^L$, $\forall j =1,2,\ldots,N$. At the beginning of our algorithm, $s=1$, $e=T$ and depending on whether $r_1 \leq T - r_N$ then $r_1$ or $r_N$ will get isolated in a right- or left-expanding interval, respectively. W.l.o.g., assume that $r_1 \leq T - r_N$. As already mentioned, ID naturally ensures that $\exists k \in \left\lbrace 1,2,\ldots, K\right\rbrace$ such that $c_k^r \in I_1^R$. There is no other change-point in $[1,c_k^r]$ apart from $r_1$. We will show that for $\tilde{b} = {\rm argmax}_{1\leq t < c_k^r}\left|\tilde{X}_{1,c_k^r}^t\right|$, then $\left|\tilde{X}_{1,c_k^r}^{\tilde{b}}\right| > \zeta_T$. Using \eqref{A_T}, we have that
\begin{equation}
\label{thresholdpassing_firststep}
\left|\tilde{X}_{1,c_k^r}^{\tilde{b}}\right| \geq \left|\tilde{X}_{1,c_k^r}^{r_1}\right| \geq \left|\tilde{f}_{1,c_k^r}^{r_1}\right| - \sqrt{8 \log T}.
\end{equation} 
But,
\begin{align}
\label{midstep1}
\nonumber & \left|\tilde{f}_{1,c_k^r}^{r_1}\right|  = \left|\sqrt{\frac{c_k^r - r_1}{r_1 c_k^r}}r_1f_{r_1} - \sqrt{\frac{r_1}{c_k^r(c_k^r - r_1)}}(c_k^r - r_1)f_{r_1+1}\right| = \sqrt{\frac{(c_k^r - r_1)r_1}{c_k^r}}\Delta_1^f\\
& = \sqrt{\frac{(c_k^r - r_1)r_1}{(c_k^r-r_1)+r_1}}\Delta_1^f \geq \sqrt{\frac{(c_k^r - r_1)r_1}{2\max\left\lbrace c_k^r - r_1, r_1 \right\rbrace}}\Delta_1^f = \sqrt{\frac{\min\left\lbrace c_k^r - r_1, r_1 \right\rbrace}{2}}\Delta_1^f.
\end{align}
By the definition of $\delta_T$ and from our notation of $r_0 = 0$, we know that $r_1 \geq \delta_T$. In addition, since $c_k^r \in I_1^R$, then $\delta_T/3 \leq c_k^r - r_1 < 2\delta_T/3$, meaning that
\begin{equation}
\label{mindistance1}
\min\left\lbrace c_k^r - r_1, r_1 \right\rbrace \geq \frac{\delta_T}{3}.
\end{equation}
The result in \eqref{thresholdpassing_firststep}, the assumption (A1) and the application of \eqref{mindistance1} in \eqref{midstep1} yield
\begin{align}
\nonumber \left|\tilde{X}_{1,c_k^r}^{\tilde{b}}\right| & \geq \sqrt{\frac{\delta_T}{6}}\Delta_1^f - \sqrt{8\log T} \geq \sqrt{\frac{\delta_T}{6}}\underline{f}_T - \sqrt{8\log T}\\
\nonumber & = \left(\frac{1}{\sqrt{6}} - \frac{2\sqrt{2 \log T}}{\sqrt{\delta_T}\underline{f}_T}\right)\sqrt{\delta_T}\underline{f}_T \geq \left(\frac{1}{\sqrt{6}} - \frac{2\sqrt{2}}{\underline{C}}\right)\sqrt{\delta_T}\underline{f}_T\\
\nonumber & = C_2\sqrt{\delta_T}\underline{f}_T > \zeta_T.
\end{align}
Therefore, there will be an interval of the form $[1,c_{\tilde{k}}^r]$, with $c_{\tilde{k}}^r > r_1$, such that $[1,c_{\tilde{k}}^r]$ contains only $r_1$ and $\max_{1\leq b < c_{\tilde{k}}^r}\left|\tilde{X}_{1,c_{\tilde{k}}^r}^{b}\right|> \zeta_T$. Let us, for $k^* \in \left\lbrace 1,2,\ldots, K \right\rbrace$, to denote by $c_{k^*}^r \leq c_{\tilde{k}}^r$ the first right-expanding point where this happens and let $b_1 = {\rm argmax}_{1\leq t < c_{k^*}^r}\left|\tilde{X}_{1,c_{k^*}^r}^t\right|$ with $\left|\tilde{X}_{1,c_{k^*}^r}^{b_1}\right| > \zeta_T$. Our aim now is to find $\gamma_T>0$ such that for any $b^* \in \left\lbrace 1,2,\ldots,c_{k^*}^r - 1 \right\rbrace$ with $\left|b^*-r_1\right|\left(\Delta_1^f\right)^2 > \gamma_T$, we have that
\begin{equation}
\label{relationship_contradiction}
\left(\tilde{X}_{1,c_{k^*}^r}^{r_1}\right)^2 > \left(\tilde{X}_{1,c_{k^*}^r}^{b^*}\right)^2.
\end{equation}
Proving \eqref{relationship_contradiction} and using the definition of $b_1$ we can conclude that $|b_1-r_1|\left(\Delta_1^f\right)^2 \leq \gamma_T$. Now, since $X_t = f_t + \epsilon_t$, then \eqref{relationship_contradiction} can be expressed as
\begin{align}
\label{relationship_contradiction2}
\nonumber \left(\tilde{f}_{1,c_{k^*}^r}^{r_1}\right)^2 - \left(\tilde{f}_{1,c_{k^*}^r}^{b^*}\right)^2 > &\left(\tilde{\epsilon}_{1,c_{k^*}^r}^{b^*}\right)^2 - \left(\tilde{\epsilon}_{1,c_{k^*}^r}^{r_1}\right)^2\\
& + 2\left\langle\boldsymbol{\psi_{1,c_{k^*}^r}^{b^*}}\langle\boldsymbol{f},\boldsymbol{\psi_{1,c_{k^*}^r}^{b^*}}\rangle - \boldsymbol{\psi_{1,c_{k^*}^r}^{r_1}}\langle\boldsymbol{f},\boldsymbol{\psi_{1,c_{k^*}^r}^{r_1}}\rangle,\boldsymbol{\epsilon}\right\rangle.
\end{align}
W.l.o.g. assume that $b^* \geq r_1$ and a similar approach as below holds when $b^*<r_1$. Lemma \ref{lemma_from_NOT}, gives for the left-hand side of the inequality in \eqref{relationship_contradiction2} that
\begin{equation}
\label{Lambda_mean}
\left(\tilde{f}_{1,c_{k^*}^r}^{r_1}\right)^2 - \left(\tilde{f}_{1,c_{k^*}^r}^{b^*}\right)^2 = \frac{\left|b^* - r_1\right|r_1}{\left|b^* - r_1\right| + r_1}\left(\Delta_1^f\right)^2 := \Lambda.
\end{equation}
For the terms on the right-hand side of \eqref{relationship_contradiction2}, using \eqref{A_T} we obtain that
\begin{equation}
\nonumber \left(\tilde{\epsilon}_{1,c_{k^*}^r}^{b^*}\right)^2 - \left(\tilde{\epsilon}_{1,c_{k^*}^r}^{r_1}\right)^2 \leq \max_{s,e,b:s\leq b <e} \left(\tilde{\epsilon}_{s,e}^b\right)^2 - \left(\tilde{\epsilon}_{1,c_{k^*}^r}^{r_1}\right)^2 \leq \max_{s,e,b:s\leq b <e} \left(\tilde{\epsilon}_{s,e}^b\right)^2 \leq 8\log T,
\end{equation}
while from \eqref{B_T} and Lemma \ref{lemma_from_NOT},
\begin{align}
\nonumber & 2\left\langle\boldsymbol{\psi_{1,c_{k^*}^r}^{b^*}}\langle\boldsymbol{f},\boldsymbol{\psi_{1,c_{k^*}^r}^{b^*}}\rangle - \boldsymbol{\psi_{1,c_{k^*}^r}^{r_1}}\langle\boldsymbol{f},\boldsymbol{\psi_{1,c_{k^*}^r}^{r_1}}\rangle,\boldsymbol{\epsilon}\right\rangle\\
\nonumber & \leq 2\|\boldsymbol{\psi_{1,c_{k^*}^r}^{b^*}}<\boldsymbol{f},\boldsymbol{\psi_{1,c_{k^*}^r}^{b^*}}> - \boldsymbol{\psi_{1,c_{k^*}^r}^{r_1}}<\boldsymbol{f},\boldsymbol{\psi_{1,c_{k^*}^r}^{r_1}}>\|_{2}\sqrt{8\log T} = 2\sqrt{\Lambda}\sqrt{8\log T}.
\end{align}
Therefore \eqref{relationship_contradiction2} is satisfied if the stronger inequality $\nonumber \Lambda > 8\log T + 2\sqrt{\Lambda}\sqrt{8\log T}$ is satisfied, which has solution
\begin{equation}
\nonumber \Lambda > (2\sqrt{2}+4)^2\log T.
\end{equation}
From \eqref{Lambda_mean} and since $(\left|b^*-r_1\right|r_1)/(\left|b^*-r_1\right| + r_1) \geq \min\left\lbrace\left|b^*-r_1\right|,r_1\right\rbrace/2$, we deduce that \eqref{relationship_contradiction} is implied by
\begin{equation}
\label{relationship_contradiction3}
\min\left\lbrace\left|b^*-r_1\right|,r_1\right\rbrace > \frac{2(2\sqrt{2}+4)^2\log T}{\left(\Delta_1^f\right)^2} = \frac{C_3\log T}{\left(\Delta_1^f\right)^2}.
\end{equation}
However,
\begin{align}
\label{relationship_contradiction4}
\min\left\lbrace r_1,c_{k^*}^r-r_1 \right\rbrace > C_3\frac{\log T}{\left(\Delta_1^f\right)^2}
\end{align}
and this is because if we assume that $\min \left\lbrace r_1,c_{k^*}^r-r_1 \right\rbrace \leq C_3\log T/\left(\Delta_1^f\right)^2$, then
\begin{align}
\nonumber \left|\tilde{X}_{1,c_{k^*}^r}^{b_1}\right| & \leq \left|\tilde{f}_{1,c_{k^*}^r}^{r_1}\right| + \sqrt{8\log T} = \sqrt{\frac{(c_{k^*}^r - r_1)r_1}{c_{k^*}^r}}\Delta_1^f + \sqrt{8\log T}\\
\nonumber & \leq \sqrt{\min \left\lbrace c_{k^*}^r - r_1,r_1 \right\rbrace}\Delta_1^f + \sqrt{8\log T}\leq \left(\sqrt{C_3} + \sqrt{8}\right)\sqrt{\log T}\\
\nonumber & = C_1\sqrt{\log T} \leq \zeta_T.
\end{align}
This comes to a contradiction to $\left|\tilde{X}_{1,c_{k^*}^r}^{b_1}\right| > \zeta_T$. Therefore, \eqref{relationship_contradiction4} holds and \eqref{relationship_contradiction3} is restricted to $\left|b^* - r_1\right|\left(\Delta_1^f\right)^2 > C_3\log T$, which implies \eqref{relationship_contradiction}. Thus, we conclude that necessarily,
\begin{equation}
\label{distance1}
\left|b_1 - r_1\right|\left(\Delta_1^f\right)^2 \leq C_3 \log T.
\end{equation}
So far, for $\lambda_T \leq \delta_T/3$ we have proven that working under the assumption that $A_T$ and $B_T$ hold, there will be an interval $[1,c_{k^*}^r]$, with $\left|\tilde{X}_{1,c_{k^*}^r}^{b_1}\right|>\zeta_T$, where $b_1=\underset{1\leq t < c_{k^*}^r}{{\rm argmax}}\left|\tilde{X}_{1,c_{k^*}^r}^{t}\right|$ is an estimation of $r_1$ that satisfies \eqref{distance1}.
\vspace{0.1in}
\\
{\textbf{Step 3.2:}}  After detecting the first change-point, ID follows the same process as in Step 3.1 but in the set $[c_{k^*}^r,T]$, which contains $r_2, r_3, \ldots, r_N$. This means that we bypass, without checking for possible change-points, the interval $[b_1+1, c_{k^*}^r)$ and we need to prove that:
\begin{itemize}
\item[(S.1)] There is no change-point in $[b_1 + 1,c_{k^*}^r)$, apart from maybe the already detected $r_1$;
\item[(S.2)] $c_{k^*}^r$ is at a location which allows for detection of $r_2$.
\end{itemize}
{\textbf{For (S.1):}} We will split the explanation into two cases with respect to the location of $b_1$.\\
{\textbf{Case 1:}} $b_1 < r_1 < c_{k^*}^r$. Using \eqref{distance1} and imposing the condition
\begin{equation}
\label{condition_delta1}
\delta_T > 3C_3\frac{\log T}{\left(\Delta_1^f\right)^2},
\end{equation}
then since $c_{\tilde{k}}^r \in I_1^R$, we have that
\begin{equation}
\nonumber c_{k^*}^r - b_1 \leq c_{\tilde{k}}^r - b_1 = c_{\tilde{k}}^r - r_1 + r_1 - b_1 < 2\frac{\delta_T}{3} + r_1 - b_1 \leq 2\frac{\delta_T}{3} + \frac{C_3\log T}{\left(\Delta_1^f\right)^2} < \delta_T.
\end{equation}   
Since $r_2 - r_1 \geq \delta_T$ and $r_1$ is already in $[b_1+1,c_{k^*}^r)$, then there is no other change-point in $[b_1+1,c_{k^*}^r)$ apart from $r_1$. Actually, the result in \eqref{condition_delta1} is not an extra assumption and we will briefly explain the reason at the end of our proof.\\
{\textbf{Case 2:}} $r_1 \leq b_1 < c_{k^*}^r$. Since $c_{\tilde{k}}^r \in I_1^R$, then $c_{k^*}^r - r_1 \leq c_{\tilde{k}}^r - r_1 < 2\delta_T/3$, which means that apart from $r_1$ there is no other change-point in $[r_1,c_{k^*}^r)$. With $r_1 \leq b_1$, then $[b_1+1,c_{k^*}^r)$ does not have any change-point.

Cases 1 and 2 above show that no matter the location of $b_1$, there is no change-point in $[b_1+1,c_{k^*}^r)$ other than possibly the previously detected $r_1$. Similarly to the approach in Step 3.1, our method applied now in $[c_{k^*}^r,T]$,  will first isolate $r_2$ or $r_N$ depending on whether $r_2 - c_{k^*}^r$ is smaller or larger than $T-r_N$. If $T-r_N < r_2 - c_{k^*}^r$ then $r_N$ will get isolated first in a left-expanding interval and the procedure to show its detection is exactly the same as for the detection of $r_1$ in Step 3.1. Therefore, for the sake of showing (S.2) let us assume that $r_2 - c_{k^*}^r \leq T-r_N$.
\vspace{0.1in}
\\
{\textbf{For (S.2):}} With ${\rm R}_{s,e}$ as in \eqref{expanding_points_s_e}, there exists $c_{k_2}^r \in {\rm R}_{c_{k^*}^r,T}$ such that $c_{k_2}^r \in I_2^R$, with $I_{j}^R$ defined in \eqref{isolating_intervals}. We will show that $r_2$ gets detected in $[c_{k^*}^r,c_{k^*_2}^r]$, for $k_2^* \leq k_2$ and its detection is $b_2 = {\rm argmax}_{c_{k^*}^r \leq t < c_{k^*_2}^r}\left|\tilde{X}_{c_{k^*}^r,c_{k^*_2}^r}^t\right|$, which satisfies $\left|b_2 - r_2\right|\left(\Delta_2^f\right)^2 \leq C_3\log T$. Following similar steps as in \eqref{midstep1}, we have that for $\tilde{b}_2 = \underset{c_{k^*}^r \leq t < c_{k_2}^r}{\rm argmax}\left|\tilde{X}_{c_{k^*}^r,c_{k_2}^r}^t\right|$,
\begin{equation}
\label{midstepforr2}
\left|\tilde{X}_{c_{k^*}^r,c_{k_2}^r}^{\tilde{b}_2}\right| \geq \left|\tilde{f}_{c_{k^*}^r,c_{k_2}^r}^{r_2}\right| - \sqrt{8\log T} \geq \sqrt{\frac{\min\left\lbrace c_{k_2}^r-r_2,r_2-c_{k^*}^r+1 \right\rbrace}{2}}\Delta_2^f - \sqrt{8\log T}.
\end{equation}
By construction, $c_{k_2}^r - r_2 \geq \delta_T/3$ and
\begin{align}
\nonumber & r_2 - c_{k^*}^r + 1 \geq r_2 - c_{\tilde{k}}^r + 1 = r_2 - r_1 - (c_{\tilde{k}}^r - r_1) + 1 \geq \delta_T - (c_{\tilde{k}}^r - r_1) + 1\\
\nonumber & \qquad\qquad\;\quad > \delta_T - 2\frac{\delta_T}{3} + 1 > \frac{\delta_T}{3},
\end{align}
which means that $\min\left\lbrace c_{k_2}^r-r_2,r_2-c_{k^*}^r+1 \right\rbrace \geq (\delta_T/3)$ and therefore continuing from \eqref{midstepforr2}, 
\begin{align}
\nonumber \left|\tilde{X}_{c_{k^*}^r,c_{k_2}^r}^{\tilde{b}_2}\right| & \geq \sqrt{\frac{\delta_T}{6}}\Delta_2^f - \sqrt{8\log T} \geq \left(\frac{1}{\sqrt{6}} - \frac{2\sqrt{2}\sqrt{\log T}}{\sqrt{\delta_T}\underline{f}_T}\right)\sqrt{\delta_T}\underline{f}_T\\
\nonumber & \geq \left(\frac{1}{\sqrt{6}}-\frac{2\sqrt{2}}{\underline{C}}\right)\sqrt{\delta_T}\underline{f}_T = C_2\sqrt{\delta_T}\underline{f}_T >\zeta_T.
\end{align}
Therefore, for a $c_{\tilde{k}_2}^r \in {\rm{R}}_{c_{k^*}^r,T}$ we have shown that there exists an interval of the form $[c_{k^*}^r,c_{\tilde{k}_2}^r]$, with $\max_{c_{k^*}^r\leq b <c_{\tilde{k}_2}^r}\left|\tilde{X}_{c_{k^*}^r,c_{\tilde{k}_2}^r}^b\right| > \zeta_T$. Let us denote by $c_{k^*_2}^r \in {\rm{R}}_{c_{k^*}^r,T}$ the first right-expanding point where this occurs  and let $b_2 = {\rm argmax}_{c_{k^*}^r\leq t < c_{k^*_2}^r}\left|\tilde{X}_{c_{k^*}^r,c_{k^*_2}^r}^t\right|$ with $\left|\tilde{X}_{c_{k^*}^r,c_{k^*_2}^r}^{b_2}\right| > \zeta_T$.

We will now show that $\left|b_2 - r_2\right|\left(\Delta_2^f\right)^2 \leq C_3\log T$. Following exactly the same process as in Step 3.1 and assuming now w.l.o.g. that $b_2 < r_2$, we have that for $b^* \in \left\lbrace c_{k^*}^r,\ldots,c_{k^*_2}^r - 1\right\rbrace$, 
\begin{equation}
\label{r2_relationship_contradiction}
\left(\tilde{X}_{c_{k^*}^r,c_{k^*_2}^r}^{r_2}\right)^2 > \left(\tilde{X}_{c_{k^*}^r,c_{k^*_2}^r}^{b^*}\right)^2
\end{equation}
is implied by $\min \left\lbrace \left|b^* - r_2\right|,c_{k^*_2}^r - r_2 \right\rbrace > C_3\log T/\left(\Delta_2^f\right)^2$. In the same way as in Step 3.1 and by contradiction we can show that $\min \left\lbrace c_{k^*_2}^r - r_2, r_2 - c_{k^*}^r + 1\right\rbrace > C_3\log T/\left(\Delta_2^f\right)^2$
and \eqref{r2_relationship_contradiction} is implied by $\left|b^* - r_2\right|\left(\Delta_2^f\right)^2 > C_3\log T$. Therefore $\left|b_2 - r_2\right|\left(\Delta_2^f\right)^2 > C_3 \log T$ would mean that $\left|\tilde{X}_{c_{k^*}^r,c_{k^*_2}^r}^{r_2}\right| > \left|\tilde{X}_{c_{k^*}^r,c_{k^*_2}^r}^{b_2}\right|$, which is not true by the definition of $b_2$. Having said this, we conclude that $\left|b_2 - r_2\right|\left(\Delta_2^f\right)^2 \leq C_3\log T$.
Having detected $r_2$, then our algorithm will proceed in the interval $[s,e]=[c_{k^*_2}^r, T]$ and all the change-points will get detected one by one since Step 3.2 will be applicable as long as there are undetected change-points in $[s,e]$.

Denoting by $\hat{r}_j$ the estimation of $r_j$ as we did in the statement of the theorem, then we conclude that all change-points will get detected one by one and $\left|\hat{r}_j - r_j\right|\left(\Delta_j^f\right)^2 \leq C_3\log T, \quad \forall j \in \left\lbrace 1,2,\ldots,N \right\rbrace$. In addition, as one can see from \eqref{condition_delta1}, our process imposes that $\delta_T > 3C_3 \log T/\left(\Delta_j^f\right)^2, \forall j \in \left\lbrace 1,2,\ldots,N \right\rbrace$, which, by the definition of $\underline{f}_T$, is implied by
\begin{equation}
\label{condition_deltageneral}
\delta_T > 3\frac{C_3 \log T}{\underline{f}_T^2}.
\end{equation}
We will now explain why \eqref{condition_deltageneral} is not actually an extra assumption but it is implied by our assumption (A1), which requires $\delta_T \geq \underline{C}^2\log T/\underline{f}_T^2$. Proving that $\underline{C} > \sqrt{3C_3}$ would mean that indeed (A1) implies \eqref{condition_deltageneral}. Due to $C_1\sqrt{\log T} \leq \zeta_T < C_2\sqrt{\delta_T}\underline{f}_T$, we require $\underline{C}$ to be such that $\underline{C}C_2 > C_1$. Simple steps yield
\begin{equation}
\label{condition_underline_C}
\underline{C}C_2 > C_1 \Leftrightarrow \underline{C}\left(\frac{1}{\sqrt{6}}-\frac{2\sqrt{2}}{\underline{C}}\right) > \sqrt{C_3} + \sqrt{8} \Leftrightarrow \underline{C} > \sqrt{6}\left(\sqrt{C_3} + 4\sqrt{2}\right).
\end{equation}
We conclude that $\underline{C} > \sqrt{3C_3}$, meaning that \eqref{condition_deltageneral} is something already satisfied due to (A1).
\vspace{0.1in}
\\
{\textbf{Step 4:}} The arguments given in Steps 1-3 hold in $A_T \cap B_T$. At the beginning of the algorithm, $s=1, e=T$ and for $N\geq 1$, there exist $k_1\in \left\lbrace 1,2,\ldots, K \right\rbrace$ such that $s_{k_1} = s, e_{k_1} \in I_1^R$ and $k_2\in \left\lbrace 1,2,\ldots, K \right\rbrace$ such that $s_{k_2} \in I_N^L, e_{k_2} = e$. As in our previous steps, w.l.o.g. assume that $r_1 \leq T - r_N$ and $r_1$ gets isolated and detected first in an interval $[s,c_{k^*}^r]$, where $c_{k^*}^r \in {\rm{R}}_{1,T}$ and it is less than or equal to $e_{k_1}$. Then, $\hat{r}_1 = {\rm argmax}_{s \leq t < c_{k^*}^r}|\tilde{X}_{s,c_{k^*}^r}^t|$ is the estimated location for $r_1$ and $\left|r_1-\hat{r}_1\right|\left(\Delta_1^f\right)^2\leq C_3\log T$. After this, the method continues in $[c_{k^*}^r,T]$ and keeps detecting all the change-points as explained in Step 3. There will not be any double detection issues because naturally, at each step of the algorithm, the new interval $[s,e]$ does not include any previously detected change-points. Once all the change-points have been detected one by one, then $[s,e]$ will contain no other change-points. ID will keep checking for possible change-points in intervals of the form $\left[s,c_{\tilde{k}_1}^r\right]$ and $\left[c_{\tilde{k}_2}^l,e\right]$ for $c_{\tilde{k}_1}^r \in {\rm{R}}_{s,e}$ and $c_{\tilde{k}_2}^l \in {\rm{L}}_{s,e}$. We denote by $[s^*,e^*]$ any of these intervals. ID will not detect anything in $[s^*,e^*]$ since $\forall b \in [s^*,e^*)$,
\begin{equation}
\nonumber \left|\tilde{X}_{s^*,e^*}^{b}\right| \leq \left|\tilde{f}_{s^*,e^*}^{b}\right| + \sqrt{8\log T} = \sqrt{8\log T} < C_1\sqrt{\log T}\leq \zeta_T.
\end{equation}
After not detecting anything in all intervals of the above form, then the algorithm concludes that there are not any change-points in $[s,e]$ and stops.$\qquad\qquad\qquad\;\;\;\;\blacksquare$
\begin{remark}
\label{remark_lambda_T}
It is interesting to explore what happens when instead of $\lambda_T \leq \delta_T/3$, we use the more general case of $\lambda_T \leq \delta_T/m$, for $m > 1$. The adjustments need to be made are
\begin{itemize}[leftmargin=0.4in]
\item[(Adj.1)] Instead of the definition in \eqref{isolating_intervals}, we now have
\begin{align}
\nonumber & I_j^R = \left[r_j + \frac{(m-1)\delta_T}{2m}, r_j + \frac{(m+1)\delta_T}{2m}\right)\\
\nonumber & I_j^L = \left(r_j - \frac{(m+1)\delta_T}{2m}, r_j - \frac{(m-1)\delta_T}{2m}\right].
\end{align}
Note that the length of the above intervals is $\delta_T/m$, meaning that with probability one there will be at least one left and one right expanding point in each of them because the distance between two consecutive right (left) expanding points is $\lambda_T \leq \delta_T/m$.
\item[(Adj.2)] Instead of $C_2$ as in \eqref{constantsPCM}, we should now use that $C_2 = \sqrt{(m-1)/4m} - 2\sqrt{2}/\underline{C}$. This is easy to prove and it will not be shown here.
\item[(Adj.3)] In \eqref{condition_underline_C}, we give a lower bound for $\underline{C}$. Following similar steps, this now becomes
\begin{equation}
\nonumber \underline{C} > \sqrt{\frac{4m}{m-1}}\left(\sqrt{C_3} + 4\sqrt{2}\right).
\end{equation}  
\end{itemize}
We see from (Adj.3) that the higher the value of $m$, the smaller the lower bound will be, meaning that the assumption on $\underline{C}$ gets possibly relaxed for larger values of $m$. On the other hand, the results above hold for an expanding level of $\lambda_T \leq \delta_T/m$ and thus, we notice that the smaller the value of $m$, the larger the upper bound for the acceptable $\lambda_T$-values. Our choice of $m=3$ gives a more symmetric aspect to our approach as the length of the intervals $I_j^R$ and $I_j^L$ is the same as the minimum distance of their start- and end-points from possible change-points, which is $\delta_T/3$.
\end{remark}

We now proceed to prove the result in Theorem \ref{consistency_theorem_trend}. For the continuous piecewise-linear case, the contrast function values at $b$ for the observed data, the signal, and the noise are denoted by $C_{s,e}^b(\boldsymbol{X})$, $C_{s,e}^b(\boldsymbol{f})$ and $C_{s,e}^b(\boldsymbol{\epsilon})$, respectively. We have $\Delta_{j}^f = \left|2f_{r_{j}} - f_{r_j-1}- f_{r_j+1}\right|$ and as in the case of piecewise-constancy, $\underline{f}_T = \underset{j=1,2,\ldots,N}{\min}\Delta_{j}^f$. The contrast vector $\boldsymbol{\phi_{s,e}^b} = (\phi_{s,e}^b(1), \phi_{s,e}^b(2),\ldots, \phi_{s,e}^b(T))$ is defined through the contrast function
\begin{equation}
\nonumber \phi_{s,e}^b(t) = \begin{cases}
    \alpha_{s,e}^b\beta_{s,e}^b\left[(e+2b-3s+2)t - (be +bs - 2s^2+2s)\right], & t=s,\ldots,b,\\
    -\frac{\alpha_{s,e}^b}{\beta_{s,e}^b}\left[(3e-2b-s+2)t - (2e^2+2e-be-bs)\right], & t=b+1,\ldots,e,\\
    0, & \text{otherwise},
  \end{cases}
\end{equation}
where $\alpha_{s,e}^b = (6/[n(n^2-1)(1+(e-b+1)(b-s+1)+(e-b)(b-s))])^{1/2}$ and $\beta_{s,e}^b = ([(e-b+1)(e-b)]/[(b-s+1)(b-s)])^{1/2}$, with $n=e-s+1$. For any vector $\boldsymbol{v} = (v_1,v_2,\ldots,v_T)$, we have that $$\left|\langle \boldsymbol{v}, \boldsymbol{\phi_{s,e}^b}\rangle \right| = C_{s,e}^b(\boldsymbol{v}).$$
Towards the proof of Theorem \ref{consistency_theorem_trend}, we use Lemmas \ref{lemma_from_NOT2} and \ref{lemma_from_NOT3} given below.
\begin{lemma}
\label{lemma_from_NOT2}
Suppose $\boldsymbol{f} = \left(f_1,f_2,\ldots,f_T\right)^{\intercal}$ is piecewise-linear vector and $r_1,\ldots,r_N$ are the locations of the change-points. Suppose $1 \leq s < e \leq T$, such that $r_{j-1} \leq s < r_j < e \leq r_{j+1}$, for some $j = 1,2,\ldots, N$. Let $\eta = \min\left\lbrace r_j - s, e-r_j \right\rbrace$. Then,
\begin{equation}
\nonumber C_{s,e}^{r_j}(\boldsymbol{f}) = \max_{s<b<e}C_{s,e}^b(\boldsymbol{f}) \begin{cases}
    \geq \frac{1}{\sqrt{24}}\eta^{\frac{3}{2}}\Delta_j^f,\\
    \leq \frac{1}{\sqrt{3}}(\eta+1)^{\frac{3}{2}}\Delta_j^f, & .
  \end{cases}
\end{equation}
\end{lemma}
\begin{proof}
See Lemma 5 from \cite{NOT_paper}.
\end{proof}
\begin{lemma}
\label{lemma_from_NOT3}
Suppose $\boldsymbol{f} = \left(f_1,f_2,\ldots,f_T\right)^{\intercal}$ is piecewise-linear vector. With $r_1,r_2,\ldots,r_N$ the locations of the change-points, suppose that $1 \leq s < e \leq T$, such that $r_{j-1}\leq s < r_j < e \leq r_{j+1}$ for some $j=1,2,\ldots,N$. Let $\rho = |r_j-b|$, $\Delta_j^f = \left|2f_{r_j}-f_{r_j-1} - f_{r_j+1}\right|$, $\eta_L= r_j-s$ and $\eta_R=e-r_j$. Then,
\begin{equation}
\nonumber \|\phi_{s,e}^b\langle\boldsymbol{f},\boldsymbol{\phi_{s,e}^b}\rangle - \boldsymbol{\phi_{s,e}^{r_j}}\langle\boldsymbol{f},\boldsymbol{\phi_{s,e}^{r_j}}\rangle\|_2^2 = \left(C_{s,e}^{r_j}(\boldsymbol{f})\right)^2 - \left(C_{s,e}^{b}(\boldsymbol{f})\right)^2.
\end{equation}
Furthermore,
\begin{enumerate}
\item for any $r_j \leq b < e$, $\left(C_{s,e}^{r_j}(\boldsymbol{f})\right)^2 - \left(C_{s,e}^{b}(\boldsymbol{f})\right)^2 \geq (1/63)\min(\rho,\eta_L)^3\left(\Delta_j^f\right)^2$;
\item for any $s < b \leq r_j$, $\left(C_{s,e}^{r_j}(\boldsymbol{f})\right)^2 - \left(C_{s,e}^{b}(\boldsymbol{f})\right)^2 \geq (1/63)\min(\rho,\eta_R)^3\left(\Delta_j^f\right)^2$.
\end{enumerate}
\end{lemma}
\begin{proof}
See Lemma 7 from \cite{NOT_paper}, where the approach is similar as to the one for Lemma \ref{lemma_from_NOT}.
\end{proof}
The steps we follow for the proof of Theorem \ref{consistency_theorem_trend} are the same as those explained for the proof of Theorem \ref{consistency_theorem}.
\vspace{0.05in}
\\
\textbf{Proof of Theorem \ref{consistency_theorem_trend}}. We will prove the more specific result
\begin{equation}
\label{mainresult_theorem3}
\Prob\left(\hat{N} = N, \max_{j=1,2,\ldots,N}\left(\left|\hat{r}_j - r_j\right|\left(\Delta_j^f\right)^{\frac{2}{3}}\right) \leq C_3(\log T)^\frac{1}{3}\right) \leq 1 - \frac{1}{6\sqrt{\pi}T},
\end{equation}
which implies the result in \eqref{mainresult_theorem_trend}.\\
{\textbf{Steps 1 and 2}:} As in Theorem \ref{consistency_theorem}, let
\begin{align}
\nonumber & A^*_T = \left\lbrace \max_{s,b,e: 1\leq s \leq b < e \leq T}\left|C_{s,e}^b(\boldsymbol{X}) - C_{s,e}^b(\boldsymbol{f})\right|\leq \sqrt{8\log T} \right\rbrace.\\
\nonumber & B^*_T = \left\lbrace \max_{j=1,2,\ldots,N}\max_{\substack{r_{j-1}<s\leq r_j\\r_j < e \leq r_{j+1}\\s\leq b < e}}\frac{\left|\left\langle\boldsymbol{\phi_{s,e}^b}\langle\boldsymbol{f},\boldsymbol{\phi_{s,e}^b}\rangle - \boldsymbol{\phi_{s,e}^{r_j}}\langle\boldsymbol{f},\boldsymbol{\phi_{s,e}^{r_j}}\rangle,\boldsymbol{\epsilon}\right\rangle\right|}{\|\boldsymbol{\phi_{s,e}^b}\langle\boldsymbol{f},\boldsymbol{\phi_{s,e}^b}\rangle - \boldsymbol{\phi_{s,e}^{r_j}}\langle\boldsymbol{f},\boldsymbol{\phi_{s,e}^{r_j}}\rangle\|_{2}}\leq \sqrt{8 \log T}\right\rbrace.
\end{align}
The same reasoning as in the proof of Theorem \ref{consistency_theorem} leads to $\Prob\left(A^*_T\right) \geq 1-1/(12\sqrt{\pi}T)$ and $\Prob\left(B^*_T\right) \geq 1- 1/(12\sqrt{\pi}T)$. Therefore, Steps 1 and 2 lead to
\begin{equation}
\nonumber \Prob\left(A^*_T \cap B^*_T\right) \geq 1 - \frac{1}{6\sqrt{\pi}T}.
\end{equation}
{\textbf{Step 3:}} This is the main part of our proof, where we explain in detail how to get the result in \eqref{mainresult_theorem3}. From now on, we assume that $A_T^*$ and $B_T^*$ both hold. The constants we use are
\begin{equation}
\nonumber C_1 = \sqrt{\frac{2}{3}}C_3^{\frac{3}{2}} + \sqrt{8},\;\; C_2 = \frac{1}{3\sqrt{72}} - \frac{2\sqrt{2}}{C^*},\;\; C_3 = 63^{\frac{1}{3}}(2\sqrt{2}+4)^{\frac{2}{3}},
\end{equation}
where $C^*$ is as in assumption (A2).
\vspace{0.1in}
\\
{\textbf{Step 3.1:}} First, $\forall j \in \left\lbrace 1,2, \ldots, N\right\rbrace$, we define $I_{j}^R$ and $I_{j}^L$ as in \eqref{isolating_intervals}. At the beginning of our algorithm, $s=1$, $e=T$ and depending on whether $r_1 \leq T - r_N$ then $r_1$ or $r_N$ will get isolated first, respectively. W.l.o.g., assume that $r_1 \leq T - r_N$. Our aim is to first show that there will be at least an interval of the form $[1,c_{\tilde{k}}^r]$, for $\tilde{k} \in \left\lbrace 1,2,\ldots,K \right\rbrace$, which contains only $r_1$ and no other change-point, such that $\max_{1\leq b <c_{\tilde{k}}^r} C_{1, c_{\tilde{k}}^r}^b > \zeta_T$. Due to ID's nature, for $K = \left\lceil T/\lambda_T  \right\rceil$, then $\exists \tilde{k} \in \left\lbrace 1,2,\ldots, K\right\rbrace$ such that $c_{\tilde{k}}^r = \tilde{k}\lambda_T \in I_1^R$ and there is no other change-point in $[1,c_{\tilde{k}}^r]$ apart from $r_1$. We will now show that for $\tilde{b}_1 = {\rm argmax}_{1 < t < c_{\tilde{k}}^r}C_{1,c_{\tilde{k}}^r}^t(\boldsymbol{X})$, then $C_{1,c_{\tilde{k}}^r}^{\tilde{b}_1}(\boldsymbol{X}) > \zeta_T$. Firstly, we have that
\begin{equation}
\label{thresholdpassing_firststep_trend}
C_{1,c_{\tilde{k}}^r}^{\tilde{b}_1}(\boldsymbol{X}) \geq C_{1,c_{\tilde{k}}^r}^{r_1}(\boldsymbol{X}) \geq C_{1,c_{\tilde{k}}^r}^{r_1}(\boldsymbol{f}) - \sqrt{8 \log T}.
\end{equation} 
From Lemma \ref{lemma_from_NOT2}, we know that $C_{1,c_{\tilde{k}}^r}^{r_1}(\boldsymbol{f}) \geq 1/(\sqrt{24})\left(\min\left\lbrace r_1-1, c_{\tilde{k}}^r - r_1\right\rbrace \right)^{3/2}\Delta_1^f$.
Now, $r_1 -1 = r_1 - r_0 - 1 \geq \delta_T -1 > \delta_T/3$, because for continuous piecewise-linear signals we have that $\delta_T \geq 2$ for identifiability purposes. In addition, since $c_{\tilde{k}}^r \in I_1^R$, then $c_{\tilde{k}}^r - r_1 \geq \delta_T/3$, meaning that
\begin{equation}
\label{mindistance1_trend}
\min\left\lbrace c_{\tilde{k}}^r - r_1, r_1 -1 \right\rbrace \geq \frac{\delta_T}{3}.
\end{equation}
The result in \eqref{thresholdpassing_firststep_trend}, the assumption (A2) and \eqref{mindistance1_trend} yield
\begin{align}
\label{thresholdpassing_trend}
\nonumber C_{1,c_{\tilde{k}}^r}^{\tilde{b}_1}(\boldsymbol{X}) & \geq \frac{1}{\sqrt{24}}\left(\frac{\delta_T}{3}\right)^{3/2}\Delta_1^f - \sqrt{8\log T} \geq \frac{1}{\sqrt{24}}\left(\frac{\delta_T}{3}\right)^{3/2}\underline{f}_T - \sqrt{8\log T}\\
\nonumber & = \delta_T^{3/2}\underline{f}_T\left(\frac{1}{3\sqrt{72}} - \frac{2\sqrt{2\log T}}{\delta_T^{3/2}\underline{f}_T}\right) \geq \left(\frac{1}{3\sqrt{72}}- \frac{2\sqrt{2}}{C^*}\right)\delta_T^{3/2}\underline{f}_T\\
& = C_2\delta_T^{3/2}\underline{f}_T > \zeta_T.
\end{align}
Therefore, there will be an interval of the form $[1,c_{\tilde{k}}^r]$, with $c_{\tilde{k}}^r > r_1$, such that $[1,c_{\tilde{k}}^r]$ contains only $r_1$ and $\max_{1\leq b < c_{\tilde{k}}^r}C_{1,c_{\tilde{k}}^r}^{b}> \zeta_T$. Let us, for $k^* \in \left\lbrace 1,2,\ldots, K \right\rbrace$, denote by $c_{k^*}^r \leq c_{\tilde{k}}^r$ the first right-expanding point where this happens and let $b_1 = {\rm argmax}_{1\leq t < c_{k^*}^r}C_{1,c_{k^*}^r}^t$ with $C_{1,c_{k^*}^r}^{b_1} > \zeta_T$. Note that $b_1$ can not be an estimation of any other change-point as $[1,c_{k^*}^r]$ includes only $r_1$.

Our aim now is to find $\tilde{\gamma}_T > 0$ such that for any $b^* \in \left\lbrace 1,2,\ldots,c_{k^*}^r - 1 \right\rbrace$ with $\left|b^*-r_1\right|\left(\Delta_1^f\right)^{2/3} > \tilde{\gamma}_T$, we have
\begin{equation}
\label{relationship_contradiction_trend}
\left(C_{1,c_{k^*}^r}^{r_1}(\boldsymbol{X})\right)^2 > \left(C_{1,c_{k^*}^r}^{b^*}(\boldsymbol{X})\right)^2.
\end{equation}
Proving \eqref{relationship_contradiction_trend} and using the definition of $b_1$ we can conclude that $|b_1-r_1|\left(\Delta_1^f\right)^{2/3} \leq \tilde{\gamma}_T$. Since $X_t = f_t + \epsilon_t$, then \eqref{relationship_contradiction_trend} can be expressed as
\begin{align}
\label{relationship_contradiction2_trend}
\nonumber \left(C_{1,c_{k^*}^r}^{r_1}(\boldsymbol{f})\right)^2 - \left(C_{1,c_{k^*}^r}^{b^*}(\boldsymbol{f})\right)^2 > &\left(C_{1,c_{k^*}^r}^{b^*}(\boldsymbol{\epsilon})\right)^2 - \left(C_{1,c_{k^*}^r}^{r_1}(\boldsymbol{\epsilon})\right)^2\\
& + 2\left\langle\boldsymbol{\phi_{1,c_{k^*}^r}^{b^*}}\langle\boldsymbol{f},\boldsymbol{\phi_{1,c_{k^*}^r}^{b^*}}\rangle - \boldsymbol{\phi_{1,c_{k^*}^r}^{r_1}}\langle\boldsymbol{f},\boldsymbol{\phi_{1,c_{k^*}^r}^{r_1}}\rangle,\boldsymbol{\epsilon}\right\rangle.
\end{align}
W.l.o.g. assume that $b^* \geq r_1$ and a similar approach as below holds when $b^*<r_1$. We denote by
\begin{equation}
\nonumber \Lambda := \left(C_{1,c_{k^*}^r}^{r_1}(\boldsymbol{f})\right)^2 - \left(C_{1,c_{k^*}^r}^{b^*}(\boldsymbol{f})\right)^2
\end{equation}
and for the terms in the right-hand side of \eqref{relationship_contradiction2_trend}, we get that
\begin{align}
\nonumber \left(C_{1,c_{k^*}^r}^{b^*}(\boldsymbol{\epsilon})\right)^2 - \left(C_{1,c_{k^*}^r}^{r_1}(\boldsymbol{\epsilon})\right)^2 & \leq  \max_{s,e,b:s\leq b <e} \left(C_{s,e}^{b}(\boldsymbol{\epsilon})\right)^2 - \left(C_{1,c_k^r}^{r_1}(\boldsymbol{\epsilon})\right)^2 \leq 8\log T,
\end{align}
while from Lemma \ref{lemma_from_NOT3},
\begin{align}
\nonumber & 2\left\langle\boldsymbol{\phi_{1,c_{k^*}^r}^{b^*}}\langle\boldsymbol{f},\boldsymbol{\phi_{1,c_{k^*}^r}^{b^*}}\rangle - \boldsymbol{\phi_{1,c_{k^*}^r}^{r_1}}\langle\boldsymbol{f},\boldsymbol{\phi_{1,c_{k^*}^r}^{r_1}}\rangle,\boldsymbol{\epsilon}\right\rangle \\
\nonumber & \leq 2\|\boldsymbol{\phi_{1,c_{k^*}^r}^{b^*}}<\boldsymbol{f},\boldsymbol{\phi_{1,c_{k^*}^r}^{b^*}}> - \boldsymbol{\phi_{1,c_{k^*}^r}^{r_1}}<\boldsymbol{f},\boldsymbol{\phi_{1,c_{k^*}^r}^{r_1}}>\|_{2}\sqrt{8\log T}\\
\nonumber & = 2\sqrt{\Lambda}\sqrt{8\log T}.
\end{align}
Therefore \eqref{relationship_contradiction2_trend} is satisfied if the stronger inequality
\begin{equation}
\nonumber \Lambda > 8\log T + 2\sqrt{\Lambda}\sqrt{8\log T}
\end{equation}
is satisfied, which has solution
\begin{equation}
\label{step_middle_trend}
\Lambda > (2\sqrt{2}+4)^2\log T.
\end{equation}
Using Lemma \ref{lemma_from_NOT3}, we have that \eqref{step_middle_trend} is implied by
\begin{align}
\label{relationship_contradiction3_trend}
\nonumber & \frac{1}{63}\left(\min\left\lbrace \left|r_1 - b^*\right|,r_1 - 1 \right\rbrace\right)^3\left(\Delta_1^f\right)^2 > \left(2\sqrt{2} + 4\right)^2\log T\\
& \Leftrightarrow \min\left\lbrace \left|r_1 - b^*\right|,r_1 - 1 \right\rbrace >  \frac{\left(63\log T\right)^{1/3}(2\sqrt{2}+4)^{2/3}}{\left(\Delta_1^f\right)^{2/3}} = \frac{C_3\left(\log T\right)^{1/3}}{\left(\Delta_1^f\right)^{2/3}}.
\end{align}
However,
\begin{align}
\label{relationship_contradiction4_trend}
\min\left\lbrace r_1-1,c_{k^*}^r-r_1 \right\rbrace > 2^{1/3}C_3\frac{\left(\log T\right)^{1/3}}{\left(\Delta_1^f\right)^{2/3}} - 1
\end{align}
and this is because if we assume that $$\min \left\lbrace r_1-1,c_{k^*}^r-r_1 \right\rbrace \leq 2^{1/3}C_3\left(\log T\right)^{1/3}/\left(\Delta_1^f\right)^{2/3} - 1$$ yields
\begin{align}
\nonumber C_{1,c_{k^*}^r}^{b_1}(\boldsymbol{X}) & \leq C_{1,c_{k^*}^r}^{r_1}(\boldsymbol{f}) + \sqrt{8\log T} \leq \frac{1}{\sqrt{3}}\left(\min\left\lbrace r_1 -1, c_{k^*}^r - r_1 \right\rbrace + 1\right)^{3/2}\Delta_1^f + \sqrt{8\log T}\\
\nonumber & \leq \frac{1}{\sqrt{3}}\left(2^{1/3}C_3\frac{\left(\log T\right)^{1/3}}{\left(\Delta_1^f\right)^{2/3}}\right)^{3/2}\Delta_1^f + \sqrt{8\log T} = \sqrt{\frac{2}{3}}C_3^{3/2}\sqrt{\log T} + \sqrt{8\log T}\\
\nonumber & = \left(\sqrt{\frac{2}{3}}C_3^{3/2} + \sqrt{8}\right)\sqrt{\log T} = C_1\sqrt{\log T} \leq \zeta_T.
\end{align}
This comes to a contradiction to $C_{1,c_{k^*}^r}^{b_1}(\boldsymbol{X}) > \zeta_T$. Therefore, \eqref{relationship_contradiction4_trend} holds and for sufficiently large $T$,
\begin{equation}
\label{relationship_contradiction3_2_trend}
\min\left\lbrace r_1 -1,c_{k^*}^r - r_1 \right\rbrace > 2^{1/3}C_3\frac{\left(\log T\right)^{1/3}}{\left(\Delta_1^f\right)^{2/3}} - 1> C_3\frac{\left(\log T\right)^{1/3}}{\left(\Delta_1^f\right)^{2/3}}.
\end{equation}
From \eqref{relationship_contradiction3_2_trend} we deduce that \eqref{relationship_contradiction3_trend} is restricted to
\begin{equation}
\nonumber \left|r_1 - b^*\right| > C_3\frac{\left(\log T\right)^{1/3}}{\left(\Delta_1^f\right)^{2/3}}, 
\end{equation}
which implies \eqref{relationship_contradiction_trend}. Therefore, necessarily,
\begin{equation}
\label{distance1_trend}
\left|b_1 - r_1\right|\left(\Delta_1^f\right)^{2/3} \leq C_3 (\log T)^{1/3}.
\end{equation}
So far, for $\lambda_T \leq \delta_T/3$ we have proven that working under the sets $A_T^*$ and $B_T^*$, there will be an interval of the form $[1,c_{k^*}^r]$, with $C_{1,c_{k^*}^r}^{b_1} > \zeta_T$, where $b_1=\underset{1\leq t < c_{k^*}^r}{\rm argmax}\;C_{1,c_{k^*}^r}^{t}$ is an estimation of $r_1$ that satisfies \eqref{distance1_trend}.
\vspace{0.1in}
\\
{\textbf{Step 3.2:}}  After detecting the first change-point, ID follows the same process as in Step 3.1 in the set $[c_{k^*}^r,T]$, which contains $r_2, r_3, \ldots, r_N$. This means that we do not check for possible change-points in the interval $[b_1+1, c_{k^*}^r)$. Therefore, we need to prove that:
\begin{itemize}
\item[(S.1)] There is no other change-point in $[b_1 + 1,c_{k^*}^r)$, apart from possibly the already detected $r_1$;
\item[(S.2)] $c_{k^*}^r$ is at a location which allows for detection of $r_2$.
\end{itemize}
{\textbf{For (S.1):}} The approach is the same as the one in Step 3.2 in the proof of Theorem \ref{consistency_theorem} and will not be repeated here.

Similarly to the approach in Step 3.1, our method applied now to $[c_{k^*}^r,T]$,  will first detect $r_2$ or $r_N$ depending on whether $r_2 - c_{k^*}^r$ is smaller or larger than $T-r_N$. If $T-r_N < r_2 - c_{k^*}^r$ then $r_N$ will get isolated first and the procedure to show its detection is exactly the same as in Step 3.1 where we explained the detection of $r_1$. Therefore, w.l.o.g. and also for the sake of showing (S.2) let us assume that $r_2 - c_{k^*}^r \leq T-r_N$.
\vspace{0.1in}
\\
{\textbf{For (S.2):}}  With ${\rm{R}}_{s,e}$ as in \eqref{expanding_points_s_e}, there exists $c_{k_2}^r \in {\rm{R}}_{c_{k^*}^r,T}$ such that $c_{k_2}^r \in I_2^R$. We will show that $r_2$ gets detected in $[c_{k^*}^r,c_{k^*_2}^r]$, for $k_2^* \leq k_2$ and its detection is $b_2 = {\rm argmax}_{c_{k^*}^r \leq t < c_{k^*_2}^r}C_{c_{k^*}^r,c_{k^*_2}^r}^t(\boldsymbol{X})$, which satisfies $\left|b_2 - r_2\right|\left(\Delta_2^f\right)^{2/3} \leq C_3(\log T)^{1/3}$. Using again Lemma \ref{lemma_from_NOT3} and for $\tilde{b}_2 = {\rm argmax}_{c_{k^*}^r \leq t < c_{k_2}^r}C_{c_{k^*}^r,c_{k_2}^r}^t$, we have that
\begin{align}
\label{midstepforr2_trend}
\nonumber C_{c_{k^*}^r,c_{k_2}^r}^{\tilde{b}_2}(\boldsymbol{X}) & \geq C_{c_{k^*}^r,c_{k_2}^r}^{r_2}(\boldsymbol{X}) \geq C_{c_{k^*}^r,c_{k_2}^r}^{r_2}(\boldsymbol{f}) - \sqrt{8 \log T}\\
& \geq \frac{1}{\sqrt{24}}\left(\min \left\lbrace r_2 - c_{k^*}^r, c_{k_2}^r-r_2 \right\rbrace\right)^{3/2}\Delta_2^f - \sqrt{8\log T}.
\end{align}
By construction,
\begin{align}
\nonumber & c_{k_2}^r - r_2 \geq \frac{\delta_T}{3}\\
\nonumber & r_2 - c_{k^*}^r \geq r_2 - c_{k}^r = r_2 - r_1 - (c_k^r - r_1) \geq \delta_T - (c_k^r - r_1) > \delta_T - 2\frac{\delta_T}{3}  = \frac{\delta_T}{3},
\end{align}
which means that $\min\left\lbrace c_{k_2}^r-r_2,r_2-c_{k^*}^r \right\rbrace \geq \delta_T/3$. Therefore, continuing from \eqref{midstepforr2_trend} and using the exact same calculations as in \eqref{thresholdpassing_trend}, we have that
\begin{align}
\nonumber C_{c_{k^*}^r,c_{k_2}^r}^{\tilde{b}_2}(\boldsymbol{X}) \geq C_2\delta_T^{3/2}\underline{f}_T > \zeta_T.
\end{align}
Therefore, for a $c_{\tilde{k}_2}^r \in {\rm{R}}_{c_{k^*}^r,T}$ we have shown that there exists an interval of the form $[c_{k^*}^r,c_{\tilde{k}_2}^r]$, with $\max_{c_{k^*}^r\leq b <c_{\tilde{k}_2}^r}C_{c_{k^*}^r,c_{\tilde{k}_2}^r}^b > \zeta_T$. Let us denote by $c_{k^*_2}^r \in {\rm{R}}_{c_{k^*}^r,T}$ the first right-expanding point where this occurs  and let $b_2 = {\rm argmax}_{c_{k^*}^r\leq t < c_{k^*_2}^r}C_{c_{k^*}^r,c_{k^*_2}^r}^t$ with $C_{c_{k^*}^r,c_{k^*_2}^r}^{b_2} > \zeta_T$. We will now show that $\left|b_2 - r_2\right|\left(\Delta_2^f\right)^{2/3} \leq C_3\left(\log T\right)^{1/3}$. Following the same process as in Step 3.1 and assuming now that $b_2 < r_2$, we have that for $b^* \in \left\lbrace c_{k^*}^r, c_{k^*}^r+1,\ldots,c_{k^*_2}^r - 1\right\rbrace$, 
\begin{equation}
\label{r2_relationship_contradiction_trend}
\left(C_{c_{k^*}^r,c_{k^*_2}^r}^{r_2}(\boldsymbol{X})\right)^2 > \left(C_{c_{k^*}^r,c_{k^*_2}^r}^{b^*}(\boldsymbol{X})\right)^2
\end{equation}
is implied by $\min \left\lbrace \left|b^* - r_2\right|,c_{k_2}^r - r_2 \right\rbrace > C_3\left(\log T\right)^{1/3}/\left(\Delta_2^f\right)^{2/3}$. However, following the same procedure as in Step 3.1 we can show that for sufficiently large $T$, $$\min \left\lbrace c_{k_2}^r - r_2, r_2 - c_{k}^r \right\rbrace > C_3\frac{\left(\log T\right)^{1/3}}{\left(\Delta_2^f\right)^{2/3}}.$$ Thus, \eqref{r2_relationship_contradiction_trend} is implied by $\left|b^* - r_2\right|\left(\Delta_2^f\right)^{2/3} > C_3(\log T)^{1/3}$. Therefore, \linebreak $\left|b_2 - r_2\right|\left(\Delta_2^f\right)^{2/3} > C_3 (\log T)^{1/3}$ would necessarily mean that $C_{c_{k^*}^r,c_{k^*_2}^r}^{r_2}(\boldsymbol{X}) > C_{c_{k^*}^r,c_{k^*_2}^r}^{b_2}(\boldsymbol{X})$, which is not true by the definition of $b_2$. Having said this, we conclude that $\left|b_2 - r_2\right|\left(\Delta_2^f\right)^{2/3} \leq C_3\left(\log T\right)^{1/3}$.

Having detected $r_2$, then our algorithm will proceed in the interval $[s,e]=[c_{k^*_2}^r, T]$ and all the change-points will get detected one by one since Step 3.2 will be applicable as long as there are previously undetected change-points in $[s,e]$. Denoting by $\hat{r}_j$ the estimation of $r_j$ as we did in the statement of the theorem, then we conclude that all change-points will first get isolated and then detected one by one and
\begin{equation}
\nonumber \left|\hat{r}_j - r_j\right|\left(\Delta_j^f\right)^{2/3} \leq C_3\left(\log T\right)^{1/3}, \quad \forall j \in \left\lbrace 1,2,\ldots,N \right\rbrace.
\end{equation}
{\textbf{Step 4:}} The arguments given in Steps 1-3 hold in $A_T^* \cap B_T^*$.  At the beginning of the algorithm, $s=1, e=T$ and for $N\geq 1$, there exist $k_1\in \left\lbrace 1,2,\ldots, K \right\rbrace$ such that $s_{k_1} = s, e_{k_1} \in I_1^R$ and $k_2\in \left\lbrace 1,2,\ldots, K \right\rbrace$ such that $s_{k_2} \in I_N^L, e_{k_2} = e$. As in our previous steps, w.l.o.g. assume that $r_1 \leq T- r_N + 1$, meaning that $r_1$ gets isolated and detected first in an interval $[s,c_{k^*}^r]$, where $c_{k^*}^r \in {\rm{R}}_{1,T}$ and it is less than or equal to $e_{k_1}$. Then, $\hat{r}_1 = {\rm argmax}_{s \leq t < c_{k^*}^r}C_{s,c_{k^*}^r}^t(\boldsymbol{X})$ is the estimated location for $r_1$ and $\left|r_1-\hat{r}_1\right|\left(\Delta_1^f\right)^{2/3}\leq C_3\left(\log T\right)^{1/3}$. After this, the algorithm continues in $[c_{k^*}^r,T]$ and keeps detecting all the change-points as explained in Step 3. It is important to note that there will not be any double detection issues because naturally, at each step of the algorithm, the new interval $[s,e]$ does not include any previously detected change-points.

Once all the change-points have been detected one by one, then $[s,e]$ will have no other change-points in it. Our method will keep interchangeably checking for possible change-points in intervals of the form $\left[s,c_{\tilde{k}_1}^r\right]$ and $\left[c_{\tilde{k}_2}^l,e\right]$ for $c_{\tilde{k}_1}^r \in {\rm{R}}_{s,e}$ and $c_{\tilde{k}_2}^l \in {\rm{L}}_{s,e}$. Allow us to denote by $[s^*,e^*]$ any of these intervals. Our algorithm will not detect anything in $[s^*,e^*]$ since $\forall b \in [s^*,e^*)$,
\begin{equation}
\nonumber C_{s^*,e^*}^{b}(\boldsymbol{X}) \leq C_{s^*,e^*}^{b}(\boldsymbol{f}) + \sqrt{8\log T} = \sqrt{8\log T} < C_1\sqrt{\log T} \leq \zeta_T.
\end{equation}
After not detecting anything in all intervals of the above form, then the algorithm concludes that there are not any change-points in $[s,e]$ and stops.$\hfill\blacksquare$
\vspace{0.1in}
\\
{\raggedleft{\textbf{Brief discussion of the steps of the proof of Theorem \ref{consistency_theorem2}}}}\\
Before the thorough mathematical proof of Theorem \ref{consistency_theorem2}, we provide an informal explanation of the three main steps in our proof. The notation is as in the main paper with $\tilde{S}$ denoting the ordered set with the remaining and relabelled estimated change-points, $\tilde{r}_k$, after each estimation is removed. At the beginning of the change-point removal approach, $\tilde{S} = \left[\tilde{r}_1, \tilde{r}_2, \ldots, \tilde{r}_{J}\right]$. In Step 1 of the proof, we show that for each true change-point $r_j, j \in \left\lbrace 1,2,\ldots, N \right\rbrace$, there is at least one and at most four estimated change-points, $\tilde{r}_k, k \in \left\lbrace 1,2,\ldots, J\right\rbrace$ within a distance equal to $\tilde{C}(\log T)^{\alpha}/\left(\Delta_j^f\right)^2$, where $\tilde{C}>0$. In Step 2, we show that there are at most two estimated change-points between two consecutive true change-points. In Step 3, we prove that as the algorithm proceeds, then $\forall j \in \left\lbrace 1,2,\ldots, N\right\rbrace$, the only remaining change-point for $r_j$ is within a distance of $C_1(\log T)^{\alpha}/\left(\Delta_j^f\right)^2$ from $r_j$ and it cannot be removed whilst there are still more than $N$ estimated change-points in $\tilde{S}$. Step 4 shows that the sSIC penalty as defined in \eqref{sSIC}, proposes a solution with $\hat{N} = N$ estimated change-points.
\vspace{0.1in}
\\
\textbf{Proof of Theorem \ref{consistency_theorem2}}\\
Allow us first to denote by
\begin{equation}
\label{D_T}
D_T = \left\lbrace \max_{s,b,e: 1 \leq b \leq e \leq T}\left|\frac{1}{\sqrt{e-b+1}}\sum_{t=b}^e\epsilon_t\right|\leq \sqrt{6\log T} \right\rbrace.
\end{equation}
We will show that $\Prob\left(D_T\right) \geq 1-\sqrt{2}/(\sqrt{\pi}T)$. For $Z \sim \mathcal{N}(0,1)$, using the Bonferroni inequality we get that 
\begin{align}
\nonumber \Prob\left((D_T)^{c}\right) & = \Prob\left(\max_{s,b,e: 1 \leq b \leq e \leq T}\left|\frac{1}{\sqrt{e-b+1}}\sum_{t=b}^e\epsilon_t \right| > \sqrt{6\log T}\right)\\
\nonumber & \leq \sum_{1 \leq b \leq e \leq T}\Prob\left(\left|Z\right|>\sqrt{6\log T}\right) \leq T^2\Prob(|Z|>\sqrt{6\log T})\\
\nonumber & = 2T^2\Prob\left(Z>\sqrt{6 \log T}\right) \leq 2T^2\frac{\phi(\sqrt{6\log T})}{\sqrt{6\log T}} \leq \frac{\sqrt{2}}{\sqrt{\pi}T},
\end{align}
where $\phi(\cdot)$ is the probability density function of the standard normal distribution. Therefore, $\Prob (D_T) \geq 1 - \sqrt{2}/(\sqrt{\pi}T)$. With $A_T$ as in \eqref{A_T}, the work that follows is valid on the set $A_T \cap D_T$ with $\Prob\left(A_T \cap D_T\right) \geq 1 - 1/(12\sqrt{\pi}T) - \sqrt{2}/(\sqrt{\pi}T)$. We take $\tilde{\tilde{C}}$ from the main paper to be equal to $2\sqrt{2}$.\\
{\textbf{Step 1}:} 
When the algorithm moves from Part 1 to Part 2, as described in Subsection \ref{subsec:sSIC}, then we are under a structure described by the following three characteristics:\\
\textbf{(P1)} For $\tilde{C} > 0$, there is at least one estimation within a distance of $\tilde{C}(\log T)^{\alpha}/\left(\Delta_j^f\right)^2$ from $r_j$, $\forall j \in \left\lbrace 1,2,\ldots,N \right\rbrace$. We know that this is true at the beginning of Part 1 due to calling the ID algorithm with threshold $\zeta_T$. This continues to be the case when the algorithm proceeds to Part 2, because if $\tilde{r}_k$ is the last estimation within $\tilde{C}(\log T)^{\alpha}/\left(\Delta_j^f\right)^2$ from $r_j$, then $\tilde{r}_{k+1} - \tilde{r}_{k-1} > 2\tilde{C}(\log T)^{\alpha}/\left(\Delta_j^f\right)^2 = 2C^*(\log T)^{\alpha}$ and $\tilde{r}_k$ cannot be removed in Part 1.\\
\textbf{(P2)} For each $j=1,2,\ldots,N$, there are at most four estimated change-points within a distance of $\tilde{C}(\log T)^{\alpha}/\left(\Delta_j^f\right)^2$ from $r_j$. We can not have more than four estimations as if this was the case then at least three of them, let's denote them by $p_1, p_2, p_3$, would be either on the right or the left of $r_j$, which would then mean that both $CS(p_2) \leq 2\sqrt{2\log T}$ and $p_3 - p_1 \leq \tilde{C}(\log T)^{\alpha}/\left(\Delta_j^f\right)^2 = C^*(\log T)^{\alpha}$ are satisfied and therefore $p_2$ would have been removed in Part 1 of the algorithm as explained in Subsection \ref{subsec:sSIC} of the paper.\\
\textbf{(P3)} There is an unknown, possibly large, number of estimated change-points (which tends to infinity as $T$ goes to infinity) between any two true change-points, namely $r_j$ and $r_{j+1}$. This issue is solved in Part 2 of the algorithm.
\vspace{0.05in}
\\
{\textbf{Step 2:}} We are now in Part 2 of the algorithm as explained in Subsection \ref{subsec:sSIC}, which guarantees that the minimum distance between two estimated change-points is \linebreak $C^*(\log T)^{\alpha}$, and also that there exists $C \geq \tilde{C}$, such that there is at least one estimation within a distance of $C(\log T)^{\alpha}/\left(\Delta_j^f\right)^2$ from $r_j$. After that, in Part 2 we collect the triplets $(\tilde{r}_{j-1},\tilde{r}_j,\tilde{r}_{j+1})$ and we calculate ${\rm CS}(\tilde{r}_j)$ and for $m = {\rm argmin}_{j}\left\lbrace {\rm CS}(\tilde{r}_j) \right\rbrace$, if $CS(\tilde{r}_m) \leq 2\sqrt{2\log T}$, then $\tilde{r}_m$ is removed and the process is repeated for the remaining estimated change-points. By doing this it is easy to see that, first of all, between $r_j$ and $r_{j+1}$, $j=0,1,\ldots,N$ there will be at most two estimated change-points, since if there were more, then we would have triplets $(\tilde{r}_{j-1},\tilde{r}_j,\tilde{r}_{j+1})$ with ${\rm CS}(\tilde{r}_j) \leq 2\sqrt{2\log T}$ and $\tilde{r}_j$ would have been removed. Secondly, for each $j = 1,2,\ldots,N$ there is still at least an estimation within a distance of $C(\log T)^{\alpha}/\left(\Delta_j^f\right)^2$ from $r_j$. If there is exactly one estimated change-point in the area $r_j \pm C(\log T)^{\alpha}/\left(\Delta_j^f\right)^2$, namely $\tilde{r}_k$, then this cannot be removed in Part 2 of the algorithm because \linebreak $\min \left\lbrace \tilde{r}_{k} - \tilde{r}_{k-1}, \tilde{r}_{k+1} - \tilde{r}_k\right\rbrace > C^*(\log T)^{\alpha}$ and therefore for $C_4 > 0$,
\begin{align}
\nonumber \left|\tilde{X}_{\tilde{r}_{k-1},\tilde{r}_{k+1}}^{\tilde{r}_k}\right| & \geq \left|\tilde{f}_{\tilde{r}_{k-1},\tilde{r}_{k+1}}^{\tilde{r}_k}\right| - 2\sqrt{2\log T}\\
\nonumber & \geq C_4(\log T)^{\alpha/2} - 2\sqrt{2\log T},
\end{align}
which for sufficiently large $T$ is greater than $2\sqrt{2\log T}$. Therefore, $\tilde{r}_k$ will not be removed and there is at least one estimation within a distance of $C(\log T)^{\alpha}/\left(\Delta_j^f\right)^2$ from $r_j$, $\forall j = 1,2,\ldots, N$. This estimation will be in the set $\tilde{S}$ of estimated change-points that continue in Part 3 of the algorithm. 
\vspace{0.1in}
\\
{\textbf{Step 3:}} We are in Part 3 of the algorithm as explained in Subsection \ref{subsec:sSIC}. In this step, we will show that with $m = {\rm argmin}_{\tilde{r}_k \in \tilde{S}}\;CS(\tilde{r}_k)$, then once $CS(\tilde{r}_m) > 2\sqrt{2\log T}$ we will be at the stage where $\tilde{S}$ contains $N$ estimated change-points; one estimated change-point within a distance of $C_1(\log T)^{\alpha}/\left(\Delta_j^f\right)^2$ from each $r_j$, $\forall j \in \left\lbrace 1,2,\ldots,N \right\rbrace$, where $C_1 > 0$. W.l.o.g. let $\tilde{r}_m$ be between $r_j$ and $r_{j+1}$. From Step 2 we know that there is a finite number (no more than four) of estimated change-points in $[r_{j-1},r_{j+1}]$. It is straightforward that at the beginning of Part 3 we have at most $2N$ estimations and either of the following two cases is possible:\\
\textbf{Case 1:} $\tilde{r}_m$ is not the closest change-point to either the true change-point on its left ($r_j$ for a $j\in\left\lbrace 1,2,\ldots,N\right\rbrace$) or the true change-point on its right ($r_{j+1}$). Since, $\tilde{s}_m = \left\lfloor \left(\tilde{r}_{m-1} + \tilde{r}_m\right)/2 \right\rfloor + 1$ and $\tilde{e}_m = \left\lceil \left(\tilde{r}_{m} + \tilde{r}_{m+1}\right)/2 \right\rceil$, it is straightforward to see that necessarily $\tilde{s}_m$ is on the right of $r_j$ and $\tilde{e}_m$ is on the left of $r_{j+1}$. This means that there are not any true change-points in $\left[\tilde{s}_m, \tilde{e}_m\right]$ and because we are working in the set $A_T$, we have that
\begin{equation}
\nonumber CS\left(\tilde{r}_m\right) = \left|X_{\tilde{s}_m,\tilde{e}_m}^{\tilde{r}_m}\right| = \left|X_{\tilde{s}_m,\tilde{e}_m}^{\tilde{r}_m}  - f_{\tilde{s}_m,\tilde{e}_m}^{\tilde{r}_m}\right| \leq 2\sqrt{2\log T}.
\end{equation}
Therefore, $\tilde{r}_m$ will be removed from the set $\tilde{S}$.\\
\textbf{Case 2:} $\tilde{r}_m$ is the closest change-point to a true change-point, namely $r_j$ for a $j \in \left\lbrace 1,2,\ldots,N \right\rbrace$, and from what has been discussed in Steps 1 and 2, $\tilde{r}_m$ is within a distance of $C(\log T)^{\alpha}/\left(\Delta_j^f\right)^2$. If $CS(\tilde{r}_m) \leq 2\sqrt{2\log T}$, then there is at least another estimated change-point within a distance of $C_{m}(\log T)^{\alpha}/\left(\Delta_j^f\right)^2$ from $\tilde{r}_m$ (and therefore from $r_j$ too), where $C_{m}$ is a constant that does not depend on $T$. If this was not the case, then since we are working under $A_T$,
\begin{align}
\label{sic_proof1}
\nonumber \left|\tilde{X}_{\tilde{s}_m,\tilde{e}_m}^{\tilde{r}_m}\right| & \geq \left|\tilde{f}_{\tilde{s}_m,\tilde{e}_m}^{\tilde{r}_m}\right| - 2\sqrt{2 \log T}\\
\nonumber & \geq C^*_m(\log T)^{\alpha} - 2\sqrt{2 \log T} > 2\sqrt{2 \log T}
\end{align}
for a constant $C^*_m$ and for sufficiently large $T$. Therefore, $\tilde{r}_m$ would not get removed.
Since there are at most $2N$ estimations, then the constants $C_m$ are upper bounded by a general constant $C_1$ and therefore, when Part 3 terminates, each true change-point will have only one estimated change-point within the distance of $C_1(\log T)^{\alpha}/\left(\Delta_j^f\right)^2$. This is exactly the stage in our algorithm, where $J = N$, with $J$ denoting the number of estimated change-points in the set $\tilde{S}$. There are not any change-point identifiability issues, because $(\log T)^{\alpha} = o\left(\delta_T\underline{f}_T^2\right)$ and for $T$ large enough
\begin{equation}
\nonumber 2C_1(\log T)^{\alpha} < \delta_T\underline{f}_T^2.
\end{equation}
The algorithm will then proceed to Part 4 explained in Subsection \ref{subsec:sSIC} and each of the remaining estimated change-points will be removed one by one until $\tilde{S}$ is the empty set.  
\vspace{0.1in}
\\
{\textbf{Step 4:}} In this last step we will prove that the sSIC penalty indicates the solution obtained when Part 3 terminates, where the number of estimated change-points is equal to $N$. We have already explained in Section \ref{subsec:sSIC} of the paper that in the scenario of piecewise-constant mean signals,
\begin{equation}
\label{sSIC_proof}
{\rm sSIC}(j) = \frac{T}{2}\log \hat{\sigma}_j^2 + (j+1)(\log T)^{\alpha},
\end{equation} 
where for any candidate model $\mathcal{M}_{j}, j=0,1,\ldots,J$, we have \linebreak $\hat{f}_t^j = (\hat{r}_{j+1} - \hat{r}_j)^{-1}\sum_{k=r_{j}+1}^{\hat{r}_{j+1}}X_j$, for $\hat{r}_{j}+1 \leq t \leq \hat{r}_{j+1}$, and $\hat{\sigma}_j^2 = T^{-1}\sum_{t=1}^{T}(X_t - \hat{f}_t^j)^2$ is the maximum likelihood estimator of the residual variance associated with model $\mathcal{M}_j$. It has also been proven in \cite{Fryzlewicz_WBS} that in an interval $[s,e]$,
\begin{equation}
\nonumber \hat{\sigma}_{j-1}^2 - \hat{\sigma}_j^2 = \frac{\left(\tilde{X}_{s,e}^d\right)^2}{T},
\end{equation}
where $d \in \left\lbrace s,s+1,\ldots,e-1 \right\rbrace$. With $j=0,1,\ldots,J$ the number of estimated change-points related to $\mathcal{M}_j$, if $j > N$, it means that all the change-points have been detected (see explanation in Step 2) and therefore $\left|\tilde{X}_{s,e}^d\right| \leq 2\sqrt{2\log T}$ since we are working under $A_T$. Therefore, $\hat{\sigma}_{j-1}^2 - \hat{\sigma}_j^2 \leq 8\log T/T$. In addition, since we are working in the set $D_T$, we have that $\left|\hat{\sigma}_N^2 - \sigma^2\right|\leq C^*\log T/T$, for a positive constant $C^*$. Using these results, the definition of ${\rm sSIC}(j)$ in \eqref{sSIC_proof}, and a first order Taylor expansion, we conclude that
\begin{align}
\label{ineqsSIC1}
\nonumber {\rm sSIC}(j) - {\rm sSIC}(N) & = \frac{T}{2}\log\frac{\hat{\sigma}_j^2}{\hat{\sigma}_N^2} + (j-N)(\log T)^{\alpha} \\
\nonumber & = \frac{T}{2}\log\left(1 - \frac{\hat{\sigma}_N^2 - \hat{\sigma}_j^2}{\hat{\sigma}_N^2}\right) + (j-N)(\log T)^{\alpha}\\
\nonumber & \geq - \frac{T}{2}(1+w)\frac{\hat{\sigma}_N^2 - \hat{\sigma}_j^2}{\hat{\sigma}_N^2} + (j-N)(\log T)^{\alpha}\\
& \geq -K_1\log T + (j-N)(\log T)^{\alpha},
\end{align}
where $K_1$ and $w$ are positive constants. The lower bound in \eqref{ineqsSIC1} is positive for $T$ large enough. Now, if $j < N$, then from the proof of Theorem \ref{consistency_theorem}, we know that in the interval $[s,e]$, we have that $\left|\tilde{X}_{s,e}^d\right| \geq \tilde{C}_2\sqrt{\delta_T}\underline{f}_T$, leading to $\hat{\sigma}_{j-1}^2 - \hat{\sigma}_j^2 \geq \tilde{C}_2^2\delta_t\underline{f}_T^2/T$. Therefore, 
\begin{align}
\label{ineqsSIC2}
\nonumber {\rm sSIC}(j) - {\rm sSIC}(N) & = \frac{T}{2}\log\frac{\hat{\sigma}_j^2}{\hat{\sigma}_N^2} + (j-N)(\log T)^{\alpha} \\
\nonumber & = \frac{T}{2}\log\left(1 + \frac{\hat{\sigma}_j^2 - \hat{\sigma}_N^2}{\hat{\sigma}_N^2}\right) + (j-N)(\log T)^{\alpha}\\
\nonumber & \geq - \frac{T}{2}(1-w_2)\frac{\hat{\sigma}_j^2 - \hat{\sigma}_N^2}{\hat{\sigma}_N^2} + (j-N)(\log T)^{\alpha}\\
& \geq K_2\delta_T\underline{f}_T^2 + (j-N)(\log T)^{\alpha},
\end{align}
where $K_2$ and $w_2$ are positive constants. The lower bound in \eqref{ineqsSIC2} is positive for $T$ large enough because $(\log T)^{\alpha} = o\left(\delta_T\underline{f}_T^2\right)$. The results in \eqref{ineqsSIC1} and \eqref{ineqsSIC2} show that for $T$ large enough and on the set $A_T \cap D_T$, we have that ${\rm sSIC}(j) > {\rm sSIC}(N)$, for $j\neq N$. Therefore, ${\rm sSIC}(j)$ is minimized for $j=N$, showing that $\hat{N} = N$.$\hfill\blacksquare$
\vspace{0.1in}
\\
{\bf Proof of Corollary \ref{cor:bound}.} From now on, we denote by
\begin{equation}
\nonumber \tilde{\epsilon}_{s,e}^b = \sqrt{\frac{e-b}{n(b-s+1)}}\sum_{t=s}^{b}\epsilon_t - \sqrt{\frac{b-s+1}{n(e-b)}}\sum_{t=b+1}^{e}\epsilon_t,
\end{equation}
where $1\leq s \leq b < e\leq T$ and $n=e-s+1$. Using also the notation in \eqref{notation_Corollary}, we have
\begin{align}
\nonumber &\exists_{s,b,e}\,\, \left\lbrace\left(\tilde{\epsilon}_{s,e}^b\right)^2 > 3 \gamma\right\rbrace\\
\nonumber & \iff \exists_{s,b,e} \left\lbrace\vphantom{\frac{e-b}{n} \left(\tilde{\tilde{\epsilon}}_{s,b}\right)^2}\frac{e-b}{n} \left(\tilde{\tilde{\epsilon}}_{s,b}\right)^2 + \frac{b-s+1}{n} \left(\tilde{\tilde{\epsilon}}_{b+1,e}\right)^2\right.\\
\nonumber & \left. \qquad\qquad\qquad\qquad- 2 \frac{\{(e-b)(b-s+1)\}^{1/2}}{n}  \tilde{\tilde{\epsilon}}_{s,b} \tilde{\tilde{\epsilon}}_{b+1,e} > 3\gamma \vphantom{\frac{e-b}{n} \left(\tilde{\tilde{\epsilon}}_{s,b}\right)^2}\right\rbrace\\
\nonumber & \implies \exists_{s,b,e} \left\{\vphantom{\frac{e-b}{n} \left(\tilde{\tilde{\epsilon}}_{s,b}\right)^2}\frac{e-b}{n} \left(\tilde{\tilde{\epsilon}}_{s,b}\right)^2 + \frac{b-s+1}{n} \left(\tilde{\tilde{\epsilon}}_{b+1,e}\right)^2\right.\\
\nonumber & \left. \qquad\qquad\qquad\qquad + 2 \frac{\{(e-b)(b-s+1)\}^{1/2}}{n}  |\tilde{\tilde{\epsilon}}_{s,b}| |\tilde{\tilde{\epsilon}}_{b+1,e}| > 3\gamma \vphantom{\frac{e-b}{n} \left(\tilde{\tilde{\epsilon}}_{s,b}\right)^2}\right\}\\
\nonumber & \implies \exists_{s,b,e} \left\{ \frac{e-b}{n} \left(\tilde{\tilde{\epsilon}}_{s,b}\right)^2 > \frac{e-b}{n} 2\gamma \right\} \\
\nonumber & \hspace{20pt} \lor \,\, \exists_{s,b,e} \left\{ \frac{b-s+1}{n} \left(\tilde{\tilde{\epsilon}}_{b+1,e}\right)^2 > \frac{b-s+1}{n} 2\gamma \right\}\\
\nonumber & \hspace{20pt} \lor\,\, \exists_{s,b,e} \left\{ 2 \frac{\{(e-b)(b-s+1)\}^{1/2}}{n}  |\tilde{\tilde{\epsilon}}_{s,b}| |\tilde{\tilde{\epsilon}}_{b+1,e}| > \gamma\right\}\\
\nonumber & \implies \exists_{s,b} \left\{ \left(\tilde{\tilde{\epsilon}}_{s,b}\right)^2 > 2\gamma \right\} \lor  \exists_{b,e} \left\{ \left(\tilde{\tilde{\epsilon}}_{b+1,e}\right)^2 > 2\gamma  \right\}\\
\nonumber & \hspace{20pt} \lor\,\, \exists_{s,b,e} \left\{ 2 \frac{\{(e-b)(b-s+1)\}^{1/2}}{n}  \tilde{\tilde{\epsilon}}_{s,b} \tilde{\tilde{\epsilon}}_{b+1,e} > \gamma\right\}\\
\label{eq:main}
& \hspace{20pt} \lor\,\, \exists_{s,b,e} \left\{ - 2 \frac{\{(e-b)(b-s+1)\}^{1/2}}{n}  \tilde{\tilde{\epsilon}}_{s,b} \tilde{\tilde{\epsilon}}_{b+1,e} > \gamma\right\}.
\end{align}
Define
\begin{equation*}
\gamma' = \frac{\gamma n}{2 \{(e-b)(b-s+1)\}^{1/2}},
\end{equation*}
and consider any straight line in the $(\tilde{\tilde{\epsilon}}_{s,b}, \tilde{\tilde{\epsilon}}_{b+1,e})$-plane, defined by the equation $\alpha \tilde{\tilde{\epsilon}}_{s,b} + \beta
\tilde{\tilde{\epsilon}}_{b+1,e} = \lambda$, that is tangent to the curve $\tilde{\tilde{\epsilon}}_{s,b} \tilde{\tilde{\epsilon}}_{b+1,e} = \gamma'$ in the quadrant $(\tilde{\tilde{\epsilon}}_{s,b}, \tilde{\tilde{\epsilon}}_{b+1,e})
> (0,0)$. By elementary calculus, which we do not repeat here, the tangency is attained when $\lambda = 2\{  \gamma' \alpha \beta   \}^{1/2}$. By elementary geometry,
\begin{equation}
\label{eq:tangent}
\tilde{\tilde{\epsilon}}_{s,b} \tilde{\tilde{\epsilon}}_{b+1,e} > \gamma' \implies |\alpha \tilde{\tilde{\epsilon}}_{s,b} + \beta \tilde{\tilde{\epsilon}}_{b+1,e}| > 2\{  \gamma' \alpha \beta   \}^{1/2}.
\end{equation}
Consider the particular values of $(\alpha, \beta)$ given by
\[
\alpha = \left\{ \frac{b-s+1}{n}  \right\}^{1/2} \qquad \beta = \left\{  \frac{e-b}{n}  \right\}^{1/2},
\]
and note that by the definition of $\tilde{\tilde{\epsilon}}_{s,e}$ we have
\begin{equation}
\label{eq:ident}
\alpha \tilde{\tilde{\epsilon}}_{s,b} + \beta \tilde{\tilde{\epsilon}}_{b+1,e} = \left\{ \frac{b-s+1}{n}  \right\}^{1/2} \tilde{\tilde{\epsilon}}_{s,b} + \left\{  \frac{e-b}{n}  \right\}^{1/2} \tilde{\tilde{\epsilon}}_{b+1,e} = \tilde{\tilde{\epsilon}}_{s,e}.
\end{equation}
Therefore, from the definitions of $\gamma'$, $\alpha$ and $\beta$, the implication in (\ref{eq:tangent}) and the identity in (\ref{eq:ident}), we have
\begin{align}
\nonumber & 2 \frac{\{(e-b)(b-s+1)\}^{1/2}}{n}  \tilde{\tilde{\epsilon}}_{s,b} \tilde{\tilde{\epsilon}}_{b+1,e} > \gamma  \iff  \tilde{\tilde{\epsilon}}_{s,b} \tilde{\tilde{\epsilon}}_{b+1,e} > \gamma'\\
\label{eq:impl2}
& \implies |\tilde{\tilde{\epsilon}}_{s,e}| > 2\{  \gamma' \alpha \beta   \}^{1/2} \iff |\tilde{\tilde{\epsilon}}_{s,e}| > \{   2\gamma  \}^{1/2} \iff \left(\tilde{\tilde{\epsilon}}_{s,e}\right)^2 > 2\gamma.
\end{align}
Considering again (\ref{eq:main}), this implies 
\begin{align}
\nonumber & \exists_{s,b,e}\,\, \left\{\left(\tilde{\epsilon}_{s,e}^b\right)^2 > 3 \gamma\right\}\\
\label{eq:main:simpl}
& \implies \exists_{s,b} \left\{ \left(\tilde{\tilde{\epsilon}}_{s,b}\right)^2 > 2\gamma \right\}\,\,\lor\,\, \exists_{s,b,e} \left\{ - 2 \frac{\{(e-b)(b-s+1)\}^{1/2}}{n}  \tilde{\tilde{\epsilon}}_{s,b} \tilde{\tilde{\epsilon}}_{b+1,e} > \gamma\right\}.
\end{align}
By \eqref{assumption_corollary}, we have
\begin{align}
\nonumber & P\left(  \exists_{s,b,e} \left\{ - 2 \frac{\{(e-b)(b-s+1)\}^{1/2}}{n}  \tilde{\tilde{\epsilon}}_{s,b} \tilde{\tilde{\epsilon}}_{b+1,e} > \gamma\right\}    \right)\\
\label{eq:rotsym}
& \le P\left(  \exists_{s,b,e} \left\{ 2 \frac{\{(e-b)(b-s+1)\}^{1/2}}{n}  \tilde{\tilde{\epsilon}}_{s,b} \tilde{\tilde{\epsilon}}_{b+1,e} > \gamma\right\}    \right).
\end{align}
From (\ref{eq:impl2}), we have
\begin{equation}
\label{eq:bound2}
P\left(  \exists_{s,b,e} \left\{ 2 \frac{\{(e-b)(b-s+1)\}^{1/2}}{n}  \tilde{\tilde{\epsilon}}_{s,b} \tilde{\tilde{\epsilon}}_{b+1,e} > \gamma\right\}    \right) \le
P\left( \exists_{s,e}  \left(\tilde{\tilde{\epsilon}}_{s,e}\right)^2 > 2\gamma         \right)
\end{equation}
Therefore, from (\ref{eq:main:simpl}), and using (\ref{eq:rotsym}) and (\ref{eq:bound2}) in turn, we have
\begin{align}
\nonumber P(\exists_{s,b,e}\,\, \left(\tilde{\epsilon}_{s,e}^b\right)^2  > 3 \gamma) & \le P(\exists_{s,e} \left(\tilde{\tilde{\epsilon}}_{s,e}\right)^2 > 2\gamma)\\
\nonumber & \quad +  P\left(  \exists_{s,b,e} \left\{ 2 \frac{\{(e-b)(b-s+1)\}^{1/2}}{n}  \tilde{\tilde{\epsilon}}_{s,b} \tilde{\tilde{\epsilon}}_{b+1,e} > \gamma\right\}    \right)\\
& \le 2 P(\exists_{s,e} \left(\tilde{\tilde{\epsilon}}_{s,e}\right)^2 > 2\gamma),
\end{align}
Now, for any $\delta > 0$, taking $\gamma = \sigma^2(1+\delta)\log\,T$, we have that
\begin{equation}
\nonumber P\left( \exists_{s,b,e}\,\, \left(\tilde{\epsilon}_{s,e}^b\right)^2 > 3 \sigma^2 (1 + \delta) \log\,T   \right) \le
2 P\left(\exists_{s,e}\,\, \left(\tilde{\tilde{\epsilon}}_{s,e}\right)^2 > 2 \sigma^2 (1 + \delta) \log\,T   \right),
\end{equation}
and the statement is a consequence of Lemma 1 in \cite{Yao}. \hfill$\blacksquare$



\begin{thebibliography}{}
\expandafter\ifx\csname natexlab\endcsname\relax\def\natexlab#1{#1}\fi


\bibitem[{Anscombe(1948)}]{Anscombe}
Anscombe, F.~J. (1948). The transformation of {P}oisson, binomial and negative-binomial data.
\newblock \textit{Biometrika}, \textbf{{\bf {35}}}, 246--254.
  

\bibitem[{Auger and Lawrence(1989)}]{Auger_Lawrence}
Auger, I.~E. and Lawrence, C.~E. (1989). Algorithms for the {O}ptimal
  {I}dentification of {S}egment {N}eighborhoods.
\newblock \textit{Bulletin of Mathematical Biology}, \textbf{{\bf {51}}},
  39--54.


\bibitem[{Bai and Perron(1998)}]{BaiPerron}
Bai, J. and Perron, P. (1998). Estimating and testing linear models with
  multiple structural changes.
\newblock \textit{Econometrica}, \textbf{{\bf {66}}}, 47--78.


\bibitem[{Baranowski et~al.(2019)Baranowski, Chen and Fryzlewicz}]{NOT_paper}
Baranowski, R., Chen, Y. and Fryzlewicz, P. (2019). Narrowest-over-threshold detection of multiple change points and change-point-like features.
\newblock \textit{Journal of the Royal Statistical Society Series B}, \textbf{{\bf {81}}}, 649--672.



\bibitem[{Chan and Walther(2013)}]{Chan_Walther}
Chan, H.~P. and Walther, G. (2013). Detection with the scan and the average
  likelihood ratio.
\newblock \textit{Statistica Sinica}, \textbf{{\bf 23}}, 409--428.

\bibitem[{Cho and Kirch (2020)Cho and Kirch}]{Cho_Kirch}
Cho, H., and Kirch, C. (2020). Data segmentation algorithms: {U}nivariate mean change and beyond.
\newblock \url{https://arxiv.org/pdf/2012.12814.pdf}.



\bibitem[{Dette et~al.(2018)Dette, Sch\"uler and
  Vetter}]{DepSMUCE}
Dette, H., Eckle, T., and Vetter, M. (2020). Multiscale change point detection for dependent data.
\newblock \textit{Scandinavian Journal of Statistics}, \textbf{{\bf 47}}, 1243--1274.

\bibitem[\protect\citeauthoryear{Eichinger and Kirch}{Eichinger and
  Kirch}{2018}]{Kirch}
Eichinger, B. and C.~Kirch (2018).
\newblock A {MOSUM} procedure for the estimation of multiple random change
  points.
\newblock {\em Bernoulli\/}~{\em {\bf {24}}}, 526--564.

\bibitem[{Fang et~al.(2020)Fang, Li and Siegmund}]{Fang_Li_Siegmund}
Fang, X., Li, J. and Siegmund, D. (2020). Segmentation and estimation of change-point models: false positive control and confidence regions.
\newblock \textit{Annals of Statistics},
  \textbf{{\bf 48}}, 1615--1647.

\bibitem[{Fang and Siegmund (2020)Fang and Siegmund}]{Fang_Siegmund2020}
Fang, X., and Siegmund, D. (2020). Detection and Estimation of Local Signals.
\newblock \url{https://arxiv.org/pdf/2004.08159.pdf}.

\bibitem[{Maidstone et~al.(2019)Fearnhead, Maidstone and
  Letchford}]{FearnheadCPOP}
Fearnhead, P., Maidstone, R. and Letchford, A. (2019). Detecting
  {C}hanges in {S}lope {W}ith an ${L}_0$ {P}enalty.
\newblock \textit{Journal of Computational and Graphical Statistics}, \textbf{{\bf {28}}}, 265--275.

\bibitem[\protect\citeauthoryear{Fearnhead and Rigaill}{Fearnhead and Rigaill}{2020}]{Fearnhead2020}
Fearnhead, P. and Rigaill, G. (2020).
\newblock Relating and comparing methods for detecting changes in mean.
\newblock {\em Stat\/}~{\em {\bf {9}}}, e291.

\bibitem[{Frick et~al.(2014)Frick, Munk and Sieling}]{SMUCE}
Frick, K., Munk, A. and Sieling, H. (2014). Multiscale change point inference.
\newblock \textit{Journal of the Royal Statistical Society Series B},
  \textbf{{\bf 76}}, 495--580.

\bibitem[\protect\citeauthoryear{Friedman}{Friedman}{1991}]{Friedman}
Friedman, J.~H. (1991).
\newblock Multivariate {A}daptive {R}egression {S}plines.
\newblock {\em Annals of Statistics\/}~{\em {\bf {19}}}, 1--141.

\bibitem[{Fryzlewicz(2014)}]{Fryzlewicz_WBS}
Fryzlewicz, P. (2014). Wild binary segmentation for multiple change-point
  detection.
\newblock \textit{Annals of Statistics}, \textbf{{\bf 42}}, 2243--2281.

\bibitem[{Fryzlewicz(2018)}]{Fryzlewicz_tail-greedy}
Fryzlewicz, P. (2018). Tail-greedy bottom-up data decompositions and fast multiple
   change-point detection.
\newblock \textit{Annals of Statistics}, \textbf{{\bf 46}}, 3390--3421.

\bibitem[{Fryzlewicz(2020)}]{Fryzlewicz_WBS2}
Fryzlewicz, P. (2020). Detecting possibly frequent change-points: Wild Binary Segmentation 2 and steepest-drop model selection.
\newblock \textit{Journal of the Korean Statistical Society}, \textbf{{\bf 49}}, 1027--1070.



\bibitem[{Hampel(1974)}]{Hampel}
Hampel, F.~R. (1974). The influence curve and its role in robust estimation.
\newblock \textit{Journal of the American Statistical Association},
  \textbf{{\bf {69}}}, 383--393.

\bibitem[{Haynes et~al.(2017)Haynes, Fearnhead and Eckley}]{Haynes_npPelt}
Haynes, K., Fearnhead, P. and Eckley, I.~A. (2017). A computationally efficient
 nonparametric approach for changepoint detection.
\newblock \textit{Statistics and Computing}, \textbf{{\bf {27}}}, 1293--1305.


\bibitem[{Jackson et~al.(2005)Jackson, Sargle, Barnes, Arabhi, Alt, Gioumousis,
  Gwin, Sangtrakulcharoen, Tan and Tsai}]{Jackson}
Jackson, B., Sargle, J.~D., Barnes, D., Arabhi, S., Alt, A., Gioumousis, P.,
  Gwin, E., Sangtrakulcharoen, P., Tan, L. and Tsai, T.~T. (2005). An
  {A}lgorithm for {O}ptimal {P}artitioning of {D}ata on an {I}nterval.
\newblock \textit{IEEE Signal Processing Letters}, \textbf{{\bf {12}}},
  105--108.
  

\bibitem[{Killick et~al.(2012)Killick, Fearnhead and Eckley}]{Killick_PELT}
Killick, R., Fearnhead, P. and Eckley, I.~A. (2012). Optimal {D}etection of
  {C}hangepoints {W}ith a {L}inear {C}omputational {C}ost.
\newblock \textit{Journal of the American Statistical Association},
  \textbf{{\bf {107}}}, 1590--1598.

\bibitem[{Kim et~al.(2009)Kim, Koh, Boyd and Gorinevsky}]{KimTF}
Kim, S.-J., Koh, K., Boyd, S. and Gorinevsky, D. (2009). $\ell_1$ {T}rend
  {F}iltering.
\newblock \textit{SIAM Review}, \textbf{{\bf {51}}}, 339--360.

\bibitem[{Kov\'{a}cs et~al.(2020)Kov\'{a}cs, Li, B\"{u}hlmann and Munk}]{Seeded_2020}
Kov\'{a}cs, S., Li, H., B\"{u}hlmann, P.,  and Munk, A. (2020). Seeded {B}inary {S}egmentation: {A} general methodology for fast and optimal change point detection.
\newblock \url{https://arxiv.org/pdf/2002.06633.pdf}.

\bibitem[{Li et~al.(2016)Li, Munk and Sieling}]{LiMunk}
Li, H., Munk, A. and Sieling, H. (2016). {FDR}-control in multiscale
  change-point segmentation.
\newblock \textit{Electronic Journal of Statistics}, \textbf{{\bf {10}}},
  918--959.

\bibitem[\protect\citeauthoryear{Liu, Wu, and Zidek}{Liu et~al.}{1997}]{Liu}
Liu, J., Wu, S. and Zidek, J.~V. (1997).
\newblock On segmented multivariate regression.
\newblock {\em Statistica Sinica\/}~{\em {\bf {7}}}, 497--526.

\bibitem[{Maidstone et~al.(2017)Maidstone, Hocking, Rigaill and
  Fearnhead}]{Maidstone_et_al}
Maidstone, R., Hocking, T., Rigaill, G. and Fearnhead, P. (2017).
  On optimal multiple changepoint algorithms for large data.
\newblock \textit{Statistics and Computing}, \textbf{{\bf {27}}}, 519--533.

\bibitem[{Muggeo and Adelfio(2011)}]{Muggeo_Adelfio}
Muggeo, V. M.~R. and Adelfio, G. (2011). Efficient {C}hange {P}oint {D}etection for
  {G}enomic {S}equences of {C}ontinuous {M}easurements.
\newblock \textit{Bioinformatics}, \textbf{{\bf {27}}}, 161--166.



\bibitem[{Olshen et~al.(2004)Olshen, Venkatraman, Lucito and
  Wigler}]{Olshen_Venkatraman}
Olshen, A.~B., Venkatraman, E.~S., Lucito, R. and Wigler, M. (2004). Circular
  binary segmentation for the analysis of array-based {DNA} copy number data.
\newblock \textit{Biostatistics}, \textbf{{\bf {5}}}, 557--572.

\bibitem[\protect\citeauthoryear{Raimondo}{Raimondo}{1998}]{Raimondo}
Raimondo, M. (1998).
\newblock Minimax estimation of sharp change points.
\newblock {\em Annals of Statistics\/}~{\em {\bf {26}}}, 1379--1397.

\bibitem[{Rigaill(2015)}]{Rigaill}
Rigaill, G. (2015). A pruned dynamic programming algorithm to recover the best
  segmentations with 1 to ${K}_{max}$ change-points.
\newblock \textit{Journal de la Soci\'{e}t\'{e} Fran\c caise de Statistique},
  \textbf{{\bf {156}}}, 180--205.


\bibitem[{Ross(2015)}]{Ross_CPM}
Ross, G.~J. (2015). Parametric and {N}onparametric {S}equential {C}hange
  {D}etection in \textsf{{R}}: {T}he cpm {P}ackage.
\newblock \textit{Journal of Statistical Software}, \textbf{{\bf 66}, No.3},
  1--20.
  
\bibitem[{Rousseeuw(1993)}]{Rousseeuw1993}
Rousseeuw, P.~J. and Croux, C. (1993). Alternatives to the {M}edian {A}bsolute {D}eviation.
\newblock \textit{Journal of the American Statistical Association}, \textbf{{\bf 88}, No.424},
  1273--1283.



\bibitem[\protect\citeauthoryear{Spiriti, Eubank, Smith, and Young}{Spiriti
  et~al.}{2013}]{Spiriti}
Spiriti, S., R.~Eubank, P.~W. Smith, and D.~Young (2013).
\newblock Knot selection for least-squares and penalized splines.
\newblock {\em Journal of Statistical Computation and Simulation\/}~{\em {\bf
  {83}}}, 1020--1036.



\bibitem[{Tibshirani(2014)}]{TibshiraniTF}
Tibshirani, R.~J. (2014). Adaptive piecewise polynomial estimation via trend
  filtering.
\newblock \textit{Annals of Statistics}, \textbf{{\bf {42}}}, 285--323.

\bibitem[\protect\citeauthoryear{Truong, Oudre, and Vayatis}{Truong
  et~al.}{2020}]{Truong_Oudre_Vayatis}
Truong, C., L.~Oudre, and N.~Vayatis (2020).
\newblock Selective review of offline change point detection methods.
\newblock {\em Signal Processing\/}~{\em {\bf
  {167}}}, 1020--1036.

\bibitem[{Venkatraman(1992)}]{Venkatraman}
Venkatraman, E.~S. (1992). \textit{{Consistency results in multiple change-point
  problems}}.
\newblock Ph.D. thesis, Stanford University.

\bibitem[{Vostrikova(1981)}]{Vostrikova}
Vostrikova, L. (1981). Detecting ``disorder'' in multidimensional random
  processes.
\newblock \textit{Soviet Mathematics: Doklady}, \textbf{{\bf 24}}, 55--59.




\bibitem[{Yao(1988)}]{Yao}
Yao, Y.-C. (1988). Estimating the number of change-points via {S}chwarz'
  criterion.
\newblock \textit{Statistics \& Probability Letters}, \textbf{{\bf 6}},
  181--189.
  
\bibitem[{Yu (2020)}]{YiYu}
Yu, Y. (2020). A review on minimax rates in change point detection and localisation.
\newblock \url{https://arxiv.org/pdf/2011.01857.pdf}.

\end{thebibliography}

\begin{thebibliography}{}
\expandafter\ifx\csname natexlab\endcsname\relax\def\natexlab#1{#1}\fi

\bibitem[{Baranowski et~al.(2019)Baranowski, Chen and Fryzlewicz}]{NOT_paper}
Baranowski, R., Chen, Y. and Fryzlewicz, P. (2019). Narrowest-over-threshold detection of multiple change points and change-point-like features.
\newblock \textit{Journal of the Royal Statistical Society Series B}, \textbf{{\bf {81}}}, 649--672.

\bibitem[\protect\citeauthoryear{Fryzlewicz}{Fryzlewicz}{2014}]{Fryzlewicz_WBS}
Fryzlewicz, P. (2014).
\newblock Wild binary segmentation for multiple change-point detection.
\newblock {\em Annals of Statistics\/}~{\em {\bf 42}}, 2243--2281.

\bibitem[{Yao(1988)}]{Yao}
Yao, Y.-C. (1988). Estimating the number of change-points via {S}chwarz'
  criterion.
\newblock \textit{Statistics \& Probability Letters}, \textbf{{\bf 6}},
  181--189.
\end{thebibliography}
\end{document}